%% file: journal_TOCL.tex
\documentclass[prodmode]{acmsmall} 

\acmVolume{0}
\acmNumber{0}
\acmArticle{0}
\acmYear{2016}
\acmMonth{0}

\doi{0000001.0000001}

\issn{1234-56789}

\usepackage{amsfonts}
\usepackage{xspace}
\usepackage{amssymb}
\usepackage{amsmath}
\usepackage[T1]{fontenc}

\newtheorem{fact}{Fact}
\newtheorem{claim}{Claim}

\usepackage[hidelinks]{hyperref}

\usepackage{enumerate}
\usepackage{multicol}
\usepackage{array}
\usepackage{eqparbox}
\usepackage{tkz-graph}
\usetikzlibrary{arrows}

\usepackage{color}
\newcommand{\new}[1]{\textcolor{black}{#1}}

\input{macros}

\begin{document}

\markboth{A. Facchini, F. Murlak, M. Skrzypczak}{Index problems for game automata}

\title{Index problems for game automata}

\author{
Alessandro Facchini
\affil{IDSIA, Switzerland}
Filip Murlak
\affil{University of Warsaw}
Micha{\l} Skrzypczak\thanks{This author has been supported by Poland's National Science Centre (decision DEC-2012/05/N/ST6/03254).}
\affil{University of Warsaw}
}

\begin{abstract}
For a given regular language of infinite trees, one can ask about the
minimal number of priorities needed to recognize this language with a
non-deterministic,  alternating, or weak alternating parity
automaton. These questions are known as, respectively, the
non-deterministic, alternating, and weak Rabin-Mostowski index
problems. Whether they can be answered effectively is a long-standing
open problem, solved so far only for languages recognizable by
deterministic automata (the alternating variant trivializes). 

We investigate a wider class of regular languages, recognizable by
so-called game automata, which can be seen as the closure of
deterministic ones under complementation and composition. Game
automata are known to recognize languages arbitrarily high in the
alternating Rabin-Mostowski index hierarchy; that is, the alternating
index problem does not trivialize any more. 

Our main contribution is that all three index problems are decidable for languages recognizable by game automata. Additionally, we show that it is decidable whether a given regular language can be recognized by a game automaton.
\end{abstract}

%
%



\acmformat{Alessandro Facchini, Filip Murlak, Micha{\l} Skrzypczak, 2015. Index problems for game automata.}

\begin{bottomstuff}
The third author has been supported by Poland's National Science Centre (decision DEC-2012/05/N/ST6/03254).
\end{bottomstuff}

\maketitle

\input{1_introduction}

\input{2and3_preliminaries}

\input{4_nondet_index}

\input{5_alt_index}

\input{6_weak_index}

\input{7_game_decidable}

\bibliographystyle{ACM-Reference-Format-Journals}
\bibliography{mskrzypczak}

\received{June 2015}{XXX}{XXX}

\end{document}

%% file: macros.tex
\usepackage{xspace}

\newcommand{\ident}[2]{\newcommand{#1}{\ensuremath{\mathrm{#2}\xspace}}}

\newcommand{\defDom}[2]{\newcommand{#1}{\ensuremath{\mathbb {#2}}}\xspace}

\defDom{\N}{N}
\defDom{\Z}{Z}
\defDom{\Q}{Q}
\defDom{\R}{R}
\defDom{\C}{C}
\defDom{\Pri}{P}
\defDom{\D}{D}

\newcommand{\aut}[1]{\ensuremath{\mathcal {#1}}\xspace}

\newcommand{\autclass}[1]{\ensuremath{\new{\mathfrak {#1}}}\xspace}

\newcommand{\boldclass}[3]{\ensuremath{\mathbf{#1}^{#2}_{#3}}}

\newcommand{\bsigma}[1]{\boldclass{\Sigma}{0}{#1}}
\newcommand{\bpi}[1]{\boldclass{\Pi}{0}{#1}}
\newcommand{\bdelta}[1]{\boldclass{\Delta}{0}{#1}}

\newcommand{\asigma}[1]{\boldclass{\Sigma}{1}{#1}}
\newcommand{\api}[1]{\boldclass{\Pi}{1}{#1}}
\newcommand{\adelta}[1]{\boldclass{\Delta}{1}{#1}}

\newcommand{\mathnm}[1]{#1}

\newcommand{\wadgeq}{\leq_{\mathrm{\mathnm{W}}}}

\newcommand{\mathcalsym}[1]{\ensuremath{\mathcal{#1}}\xspace}

\newcommand{\Aa}{\mathcalsym{A}}

\usetikzlibrary{positioning}
\usetikzlibrary{decorations.pathreplacing}
\usetikzlibrary{automata,positioning}
\usetikzlibrary{decorations.pathmorphing}
\usetikzlibrary{decorations.markings}
\usetikzlibrary{decorations}
\usetikzlibrary{arrows}
\usetikzlibrary{patterns}
\usetikzlibrary{calc}
\usetikzlibrary{shapes}
\usetikzlibrary{fadings,	shadings}

\tikzstyle{ubrace} = [draw, thick, decoration={brace, mirror, raise=0.1cm}, decorate,
    every node/.style={anchor=north, yshift=-0.2cm}]
\tikzstyle{rbrace} = [draw, thick, decoration={brace, mirror, raise=0.1cm}, decorate,
    every node/.style={anchor=west, xshift= 0.2cm}]

\tikzstyle{obrace} = [draw, thick, decoration={brace, raise=0.1cm}, decorate,
    every node/.style={anchor=south, yshift= 0.2cm}]
\tikzstyle{lbrace} = [draw, thick, decoration={brace, raise=0.1cm}, decorate,
    every node/.style={anchor=east, xshift=-0.2cm}]
    
\newcommand{\fun}[3]{\ensuremath{#1\colon #2 \to #3}}
\newcommand{\parfun}[3]{\ensuremath{#1\colon #2 \rightharpoonup #3}}

\ident{\dom}{dom}
\ident{\rg}{rg}
\ident{\id}{id}
\ident{\sgn}{sgn}

\ident{\prof}{Profiles}

\ident{\type}{tp}
\ident{\types}{Tp}

\ident{\errTrace}{NotTrace}
\ident{\errProfile}{NoProfile}

\ident{\close}{cl}
\ident{\inter}{int}

\newcommand{\lang}{L}

\newcommand{\restr}{\!\upharpoonright}

\ident{\BC}{BC}

\newcommand{\trees}[1]{\ensuremath{\mathrm{Tr}_{{#1}}}\xspace}

\newcommand{\partrees}[1]{\ensuremath{\mathrm{PTr}_{{#1}}}\xspace}

\ident{\holes}{holes}

\ident{\graph}{Gph}

\ident{\cmp}{Comp}

\newcommand{\eve}{\ensuremath{\exists}\xspace}

\newcommand{\adam}{\ensuremath{\forall}\xspace}

\newcommand{\W}[2]{\ensuremath{W_{{#1},{#2}}}\xspace}

\ident{\RM}{RM}

\ident{\RMD}{RM\Delta}

\newcommand{\RModd}[1]{\ensuremath{\mathbf{\Sigma}^{RM}_{#1}}\xspace}
\newcommand{\RMeven}[1]{\ensuremath{\mathbf{\Pi}^{RM}_{#1}}\xspace}
\newcommand{\RMdelta}[1]{\ensuremath{\mathbf{\Delta}^{RM}_{#1}}\xspace}

\newcommand{\rmin}{i}

\newcommand{\rmax}{j}

\ident{\br}{Branching}

\ident{\level}{level}

\ident{\class}{class}

\ident{\hlev}{RMlev}

\newcommand{\tran}[1]{\xrightarrow{#1}}

\ident{\game}{\mathbf{G}}

\ident{\agame}{{\mathbf{G}}}

\ident{\rhogame}{{\mathbf{G_\rho}}}

\newcommand\thmheadfont{\upshape\bf}

\newcommand{\herewasA}{}

\renewcommand{\setminus}{-}

\newcommand{\dt}{\mathrm{det}}

\newcommand{\dL}{\mathtt{L}}
\newcommand{\dR}{\mathtt{R}}

\ident{\weak}{\mathit{w}}

\ident{\wclass}{wclass}

\ident{\leave}{\textit{leave}}
\ident{\stay}{\textit{stay}}

\newcommand{\wRMeven}[1]{\ensuremath{\mathbf{\Pi}^{\weak}_{#1}}\xspace}
\newcommand{\wRModd}[1]{\ensuremath{\mathbf{\Sigma}^{\weak}_{#1}}\xspace}
\newcommand{\wRMdelta}[1]{\ensuremath{\mathbf{\Delta}^{\!\weak}_{#1}}\xspace}

\newcommand{\wRM}{\textrm{\bf{RM}}^{\weak}}

\newcommand{\iloop}{
  \node (p)  [] at (0,0) {$p$};
  
  \path[->]
    (p) edge[in=-30, out=+30, looseness=0.8, loop, distance=2cm] node[xshift=-15pt] {$3$} node[xshift=7pt] {$w$} (p);   
}

\newcommand{\ijloop}{
  \node (p1)  [] at (6,0) {$p'$};
  \node (q1)  [] at (8,0) {$q$};
  \node       [] at (8.8,0) {$a$};
  \node (ql)  [] at (9,1) {$q_\dL$};
  \node (qr)  [] at (9,-1) {$q_\dR$};
  
  \path[->]
    (p1) edge node[yshift=5pt] {$w$} (q1)
    (q1) edge[out=0, in=-75] node {} (ql)
    (q1) edge[out=0, in=+75] node {} (qr)
    (ql) edge[in=+40, out=+125] node[yshift=+7pt] {$w_\dL$} node[yshift=-15pt, xshift=12pt] {$1$} (p1)
    (qr) edge[in=-40, out=-125] node[yshift=-7pt] {$w_\dR$} node[yshift=+15pt, xshift=12pt] {$4$} (p1);
}

\newcommand{\figILoop}{
\begin{figure}
\centering
\begin{tikzpicture}[->,>=stealth',scale=0.85, every node/.style={scale=0.77}]
\iloop
\ijloop
\end{tikzpicture}
\caption{A $3$-loop rooted in $p$ and a $(1,4)$-loop rooted in $p'$.}
\label{fig:iloop}
\end{figure}
}

\newcommand{\figEdelBasic}{
\begin{figure}
\centering
\begin{tikzpicture}[->,>=stealth',scale=0.85, every node/.style={scale=0.77}]

  \node (p1)  [] at (-5+6,0) {$p$};
  \node (q1)  [] at (-5+8,0) {$q^{\adam}$};
  \node (ql)  [] at (-5+9,1) {$q^{\adam}_\dL$};
  \node (qr)  [] at (-5+9,-1) {$q^{\adam}_\dR$};
  
  \path[->]
    (p1) edge (q1)
    (q1) edge[out=0, in=-75] (ql)
    (q1) edge[out=0, in=+75] (qr)
    (ql) edge[in=+40, out=+125] node[yshift=-15pt, xshift=12pt] {$1$} (p1)
    (qr) edge[in=-40, out=-125] node[yshift=+15pt, xshift=12pt] {$2$} (p1);

  \node (p1)  [] at (+5-6,0) {$p$};
  \node (q1)  [] at (+5-8,0) {$q^{\eve}$};
  \node (ql)  [] at (+5-9,1) {$q^{\eve}_\dL$};
  \node (qr)  [] at (+5-9,-1) {$q^{\eve}_\dR$};
  
  \path[->]
    (p1) edge (q1)
    (q1) edge[out=180, in=-105] (ql)
    (q1) edge[out=180, in=+105] (qr)
    (ql) edge[in=140, out=+55] node[yshift=-15pt, xshift=-12pt] {$1$} (p1)
    (qr) edge[in=-140, out=-55] node[yshift=+15pt, xshift=-12pt] {$0$} (p1);       
\end{tikzpicture}
\caption{$(0,1)$-edelweiss and $(1,2)$-edelweiss.}
\label{fig:edelweiss-basic}
\end{figure}
}

\newcommand{\edel}[8]{
  \node (p)  [] at (#1+0,0) {$p$};
  \node (e)  [] at (#1-1.5,1) {$q^{#2}$};  
  \node (a)  [] at (#1+1.5,1) {$q^{#3}$};     
  \node (el) [] at (#1-2.6,0.6) {$q^{#2}_\dL$};  
  \node (er) [] at (#1-1.5,2.15) {$q^{#2}_\dR$};  
  \node (al) [] at (#1+1.5,2.15) {$q^{#3}_\dL$};   
  \node (ar) [] at (#1+2.6,0.6) {$q^{#3}_\dR$};

  \node at (#1-1.5,0.3) {$#6$};
  \node at (#1-0.8,1.2) {$#7$};
  \node at (#1+0.8,1.2) {$#7$};
  \node at (#1+1.5,0.3) {$#8$};
  
  \path[->]
    (p) edge[in=-80, out=-30, looseness=0.8, loop, distance=2cm] node[above left] {$#4$} (p)
    (p) edge[in=-100, out=210, looseness=0.8, loop, distance=2cm] node[above right] {$#5$} (p)
    (p) edge node {} (e)
    (e) edge[out=150, in=75] node {} (el)
    (e) edge[out=150, in=215] node {} (er)
    (p) edge node {} (a)
    (a) edge[out=30, in=-35] node {} (al)
    (a) edge[out=30, in=105] node {} (ar)
    (al) edge[bend right] node {} (p)
    (ar) edge[bend left] node {} (p)
    (el) edge[bend right] node {} (p)
    (er) edge[bend left] node {} (p);   
}

\newcommand{\figEdelNew}{
\begin{figure}
\centering
\begin{tikzpicture}[->,>=stealth',scale=0.85, every node/.style={scale=0.77}]
\edel{-4}{\eve}{\adam}01234
\edel{ 4}{\adam}{\eve}12345
\end{tikzpicture}
\caption{$(0,4)$-edelweiss and $(1,5)$-edelweiss.}
\label{fig:edelweiss}
\end{figure}
}


%% file: 1_introduction.tex
\section{Introduction}

Finite state automata running over infinite words and infinite binary trees lie at the core of the seminal works of B\"uchi~\cite{buchi_decision} and Rabin~\cite{rabin_s2s}. Known to be equivalent to the monadic second-order (MSO) logic and the modal $\mu$-calculus on both classes of structures, they subsume all standard linear and branching temporal logics. Because of these properties, they constitute fundamental tools in the theory of verification and model-checking, where the model-checking problem is reduced to the non-emptiness problem for automata: a given formula is translated into an automaton recognizing its models. From this perspective, a natural question is, which parameter in the definition of an automaton reflects the complexity of the language recognized by it. A na\"ive approach is to look at the number of states;  a more meaningful one is to consider the infinitary behaviour of \new{an automaton}, captured by the complexity of its  acceptance condition. 

Out of different acceptance conditions proposed for tree automata, B\"uchi, Muller, Rabin, Streett, and parity~\cite{mostowski_parity_games,mostowski_standard},  the last one has proved to be the most appropriate, as it enabled unveiling the subtle correspondences between games, automata, and the modal $\mu$-calculus~\cite{niwinski_rudiments,jutla_determinacy}. In a parity automaton, each state is assigned a natural number, called its priority. A sequence of states is said to be accepting if the lowest priority occurring infinitely often is even (min-parity condition). The pair $(\rmin,\rmax)$ consisting of the minimal priority $\rmin$ and the  maximal priority $\rmax$ in a given automaton is called its Rabin-Mostowski index. The index of a language is the minimal index of a recognizing automaton. Practical importance of this parameter comes from the fact that the best known algorithms deciding emptiness of (non-deterministic) automata are exponential in the number of priorities.

Given a regular tree language, what is the minimal range of priorities needed to recognize it? The answer to this question depends on which mode of computation is used, i.e, whether the automata are deterministic, non-deterministic, alternating, or weak alternating. While weak alternating and deterministic automata are weaker, non-deterministic and alternating parity automata recognize all regular tree languages. Still, alternating automata often need less priorities than non-deterministic ones. Thus, for each of these four classes there is the respective index problem.

\smallskip

\noindent   {\bf $\autclass{C}$ Index Problem:} \emph{Given $i,j$ and a regular language $L$, decide if $L$ is recognized by an automaton in class
$\autclass{C}$ of Rabin-Mostowski index $(i,j)$. }

\smallskip

The solution of this problem for the most important cases---when $\autclass{C}$ is the class of non-deterministic, alternating, or weak alternating automata---seems still far away. The results of~\cite{otto_recursion,kuster_first_level,walukiewicz_low_levels}, later extended in~\cite{bojanczyk_boolean}, show that it is decidable if a given regular tree language can be recognized by a combination of reachability and safety conditions (which corresponds to the Boolean combination of open sets). It is also known that the non-deterministic (min-parity) index problem is decidable for $(i,j) = (1,2)$, and for $(i,j) = (0,1)$ if the input language is given by an alternating automaton of index $(1,2)$~\cite{boom_phd,kuperberg_phd,colcombet_weak}. The non-deterministic index problem has been reduced to the uniform universality problem for so-called distance-parity automata~\cite{loding_index_to_bounds}, but decidability of \new{the} latter problem remains open. 

The index problems become easier when we restrict the input to languages recognized by deterministic automata. This is mostly due to the fact that in a deterministic automaton, each sub-automaton can be substituted with any automaton recognizing a language of the same index, without influencing the index of the whole language. This observation has been essential in providing a full characterization of the combinatorial structure of a language $L$ in terms of certain patterns in a deterministic automaton recognizing $L$. This so-called pattern method~\cite{murlak_phd} has been successfully used for solving all four index problems for languages recognized by deterministic automata:
\begin{theorem}
\label{thm:detindex}
For languages recognized by deterministic automata the following problems are decidable:
\begin{enumerate}
\item the \emph{deterministic} index problem~\cite{niwinski_relating};
\item the \emph{non-deterministic} index problem~\cite{urbanski_det_buchi,niwinski_deterministic}; 
\item the \emph{alternating} index problem~\cite{niwinski_gap}; and
\item the \emph{weak alternating} index problem~\cite{murlak_weak_index}.
\end{enumerate}
\end{theorem}

The pattern method cannot be applied in general to non-deterministic or alternating automata; the reason is that both these types of automata naturally implement set-theoretic union of languages and union is not an operation that preserves the index of languages. But how far can we push the pattern method beyond deterministic automata? 

In this paper we give a precise answer to this question. We present a syntactic class of automata for which substitution preserves the index of languages---we call them \emph{game automata}---and show that it is the largest such class satisfying natural closure conditions. Relying on the first property we extend Theorem~\ref{thm:detindex} (2), (3), (4) and prove the following.  
\begin{theorem}\label{thm:gameindex}
For languages recognized by game automata the following problems are decidable:
\begin{enumerate}
\item the \emph{non-deterministic} index problem,
\item the \emph{alternating} index problem,
\item the \emph{weak alternating} index problem.
\end{enumerate}
\end{theorem}
\noindent 
Decidability of the non-deterministic index problem for languages recognized by game automata is obtained via an easy reduction to the non-deterministic index problem for deterministic automata (Section~\ref{sec:ndindex}). 

As game automata recognize the game languages $\W{\rmin}{\rmax}$~\cite{arnold_strict}, the alternating index problem does not trivialize, unlike for deterministic automata, and is much more difficult than the non-deterministic index problem. We solve it by providing a recursive procedure computing the alternating index of the language recognized by a given game automaton (Section~\ref{sec:altindex}). 

Similar techniques are applied to solve the weak alternating index problem (Section~\ref{sec:weak_index}).

Finally, we give an effective characterization of languages recognized by game automata, within the class of all regular languages (Section~\ref{sec:isgame}). As the characterization effectively yields an equivalent game automaton, we obtain  procedures computing the alternating, weak alternating, and non-deterministic index for a given alternating automaton equivalent to some game automaton. 

This paper collects results from two conference papers:~\cite{murlak_game_auto} and~\cite{game_wollic}. Additionally, it contains a
discussion of the maximality of the class of game automata, which
adapts a reasoning  from~\cite{murlak_weak_game} to the index problem.


%% file: 2and3_preliminaries.tex
\section{Preliminaries}
\label{sec:prelims}


To simplify the presentation of inductive arguments, all our definitions allow partial objects: trees have leaves, automata have exits (where computation stops) and games have final positions (where the play stops and no player wins). The definitions become standard when restricted to \emph{total} objects: trees without leaves, automata without exits, and games without final positions. We also do not distinguish the initial state of an automaton but treat it as an additional parameter for the recognized language. 

\subsection{Trees}

For a function $f$ we write $\dom(f)$ for the domain of $f$ and $\rg(f)$ for the range of $f$. For a finite alphabet $A$, we denote by $\partrees{A}$ the set of partial trees over $A$, i.e., functions $\fun{t}{\dom(t)}{A}$ from a prefix-closed subset $\dom(t)\subseteq \{\dL,\dR\}^\ast$ to $A$. By $\trees{A}$ we denote the set of \emph{total} trees, i.e., trees $t$ such that $\dom(t) = \{\dL,\dR\}^\ast$.  For a direction $d\in\{\dL,\dR\}$ by $\bar{d}$ we denote the opposite direction. For $v \in \dom(t)$, $t\restr_v$ denotes the subtree of $t$ rooted at $v$. The sequences $u,v\in\{\dL,\dR\}^\ast$ are naturally ordered by the prefix relation: $u\preceq v$ if $u$ is a prefix of $v$.

A tree that is not total contains \emph{holes}. A \emph{hole} of a tree $t$ is a minimal sequence $h\in\{\dL,\dR\}^\ast$ that does not belong to $\dom(t)$. By $\holes(t)\subseteq\{\dL,\dR\}^\ast$ we denote the set of holes of a tree $t$. If $h$ is a hole of $t\in\partrees{A}$, for $s\in\partrees{A}$ we define the partial tree $t [h:=s]$ obtained by putting the root of $s$ into the hole $h$ of $t$.

\subsection{Games}
\noindent A \emph{parity game} $\game$ is a tuple $\langle V=V_\eve\cup V_\adam, v_I, F, E, \Omega \rangle$, where
\begin{itemize}
\item $V$ is a countable \emph{arena};
\item $V_\eve,V_\adam\subseteq V$ are positions of the game \emph{belonging}, respectively, to player \eve and player \adam, $V_\eve\cap V_\adam=\emptyset$;
\item $v_I\in V$ is the initial position of the game;
\item $F$ is a countable set of \emph{final positions}, $F\cap V=\emptyset$;
\item $E\subseteq V{\times}\left(V\cup F\right)$ is the transition relation;
\item $\fun{\Omega}{V}{\{\rmin,\ldots,\rmax\}}\subseteq \mathbb{N}$ is a \emph{priority function}.
\end{itemize}
%
We assume that all parity games are finitely branching (for each $v\in V$ there are only finitely many $u\in V \cup F$ such that $(v,u)\in E$), and that there are no dead-ends (for each $v\in V$ there is at least one $u\in V\cup F$ such that $(v,u)\in E$). 

A \emph{play} in a parity game $\game$ is a finite or infinite sequence $\pi$ of positions starting from $v_I$. If $\pi$ is finite then $\pi=v_I v_1\ldots v_n$ and $v_n$ is required to be a final position (that is $v_n\in F$). In that case $v_n$ is called the \emph{final position of $\pi$}. An infinite play $\pi$ is \emph{winning} for $\eve$ if $\liminf_{n\to\infty} \Omega(\pi(n))$ is even. Otherwise $\pi$ is winning for $\adam$.

A (positional) \emph{strategy} $\sigma$ for a player
$P\in\{\eve,\adam\}$ in a game $\game$ is defined as usual, as a
function assigning to every $P$'s position $v\in V_P$  the chosen
successor $\sigma(v)\in V\cup F$ such that $(v,\sigma(v))\in E$. 
\new{Strategies can be also seen as trees labelled with
positions and final positions: we label the root with the initial
position $v_I$, and then for each node labelled with a (non-final) position
of the player $P$ we add one child,  corresponding to the move determined
by the strategy, and for each node labelled with a (non-final)
position of the opponent we add a child for each possible move.} 
A play $\pi$ \emph{conforms to $\sigma$} if whenever $\pi$ visits a
vertex $v\in V_P$, the next position of $\pi$ is $\sigma(v)$\new{;
that is, if $\pi$ is a prefix of a branch of the strategy $\sigma$
viewed as a tree.}  We
say that a strategy $\sigma$ is \emph{winning} for $P$ if every
infinite play conforming to $\sigma$ is winning for $P$. For a winning
strategy $\sigma$ we define the \emph{guarantee of $\sigma$} as the
set of all final positions that can be reached in plays conforming to
$\sigma$ \new{(the labels of the leaves of $\sigma$ viewed as a
  tree)}. \new{Due to final positions, both players can have a winning 
strategy; in such case the intersection of their guarantees is nonempty,
as two winning strategies used against each other must lead the play
to a final position.
Like for parity games without final positions,  at least one player} has a (positional)
winning strategy~\cite{jutla_determinacy,mostowski_parity_games}. 


\subsection{Automata}\label{subsec:automata}

For the purpose of the inductive argument we incorporate into the definition of automata a finite set of \emph{exits}. Therefore, an alternating automaton $\aut{A}$ is defined as a tuple $\langle A, Q, F,\delta,\Omega\rangle$, where $A$ is a finite alphabet, $Q$ is a finite set of states, $F$ is a finite set of exits disjoint from $Q$, $\Omega\colon Q \to \mathbb{N}$ is a function assigning to each state of $\aut{A}$ its priority, and $\delta$ assigns to each pair $(q,a)\in Q\times A$ the transition $b=\delta(q,a)$ built using the grammar
\[ b\ ::=\ \top\ \bigm|\ \bot\ \bigm|\ (q,d) \ \bigm|\ (f,d)\ \bigm|\ b\lor b\ \bigm|\ b \land b\]
for states $q\in Q$, $f\in F$, and directions $d\in\{\dL,\dR\}$.

For an alternating automaton $\aut{A}$, a state $q_I\in Q$, and a partial tree  $t\in\partrees{A}$ we define the game $\agame(\aut{A}, t, q_I)$ as follows:
\begin{itemize}
\item  $V=\dom(t)\times (S_\delta\cup Q)$, where $S_\delta$ is the set of all sub\-formulae of formulae in $\rg(\delta)$; all positions of the form $(v,b_1\lor b_2)$ belong to $\eve$ and the remaining ones to $\adam$;\footnote{Positions $(v,(q,d)), (v,q), (v,\bot),  (v,\top)$ offer no choice, so their owner is irrelevant.}
\item $F=\left(\holes(t)\times \left(Q\cup F\right)\right) \cup \dom(t)\times F$;
\item $v_I = (\epsilon, q_I)$;
\item $E$ contains the following pairs (for all $v \in \dom(t)$):
\begin{itemize}
\item $\big ((v, b), (v, b) \big)$ for $b\in\{\top,\bot\}$,
\item $\big ((v, b), (v, b_i) \big)$ for $b=b_1 \land b_2$ or $b=b_1 \lor b_2$,
\item $\big ((v, (q,d)), (vd,q) \big)$ for $d\in\{\dL,\dR\}$, $q\in Q\cup F$,
\item $\big ((v, q), (v, \delta(q,t(v))) \big)$ for $q\in Q$;
\end{itemize}
\item $\Omega(v,\top)=0$, $\Omega(v,\bot)=1$, $\Omega(v,q)=\Omega_\Aa(q)$ for $q\in Q$, $v\in\dom(t)$, and 
for  other positions $\Omega$ is  $\max(\rg(\Omega_\Aa))$, where
$\Omega_\Aa$ is the priority function of $\Aa$.
\end{itemize}


An automaton $\aut{A}$ is \emph{total} if $F=\emptyset$. A total automaton $\aut{A}$ \emph{accepts} a total tree $t\in\trees{A}$ from $q_I\in Q$ if $\eve$ has a winning strategy in $\agame(\aut{A}, t, q_I)$. By $\lang(\aut{A}, q_I)$ we denote the set of total trees accepted by a total automaton $\aut{A}$ from a state $q_I$. A total automaton $\aut{A}$ \emph{recognizes} a language $L\subseteq \trees{A}$ if $\lang(\aut{A},q_I)=L$ for some $q_I\in Q$. A state $q\in Q$ is \emph{non-trivial} if $\emptyset \subsetneq \lang(\aut{A},q) \subsetneq \trees{A}$. Without loss of generality, when a total automaton $\aut{A}$ recognizes a non-trivial language, i.e. $ \lang(\aut{A},q_I) \notin \{\emptyset, \trees{A}\}$ for some $q_I \in Q$, we implicitly assume that $\aut{A}$ has only non-trivial states.  

The \emph{(Rabin-Mostowski) index} of an automaton $\aut{A}$ is the pair $(\rmin,\rmax)$ where $\rmin$ is the minimal and $\rmax$ is the maximal priority of the states of $\aut{A}$ ($\bot$ and $\top$ are
counted as additional looping states with odd and even 
priority, respectively). In that case $\aut{A}$ is called an $(\rmin,\rmax)$-automaton.

An automaton $\aut{A}$ is \emph{deterministic} if all its transitions are deterministic, i.e., of the form $\top$, $\bot$, $(q_d,d)$, or $(q_\dL,\dL) \land (q_\dR,\dR)$, for $d\in\{\dL,\dR\}$. Similarly, $\aut{A}$ is \emph{non-deterministic} if its transitions are (multifold) disjunctions of deterministic transitions.

An automaton $\aut{A}$ is \emph{weak} if whenever $\delta(q,a)$
contains a state $q'$ then $\Omega(q)\leq \Omega(q')$. For weak
automata, allowing \new{\emph{trivial}} transitions $\top$ or
$\bot$, interferes with the index much more \new{than} for strong automata: essentially, it adds one
more change of priority. To reflect this, when defining the index of
the automaton, we count $\bot$ and $\top$ as additional looping states
with priorities assigned so that the weakness condition above is
satisfied: $\bot$ gets the lowest \emph{odd} priority $\ell$ such that $\bot$ is
accessible only from states of priority at most $\ell$, and dually for
$\top$. That is,  if the automaton uses priorities $i, i+1, \dots,
2k-1$, we can use $\bot$ for free (with priority $2k-1$), but for
$\top$ we may need to pay with an additional priority $2k$, yielding index
$(i, 2k)$. To emphasize the fact that an automaton in question is
weak, we often call its index the \emph{weak index}. 

 

\subsection{Compositionality}

Let $\aut{A}=\left<A,Q, F, \delta, \Omega\right>$ be an alternating
automaton and $Q'\subseteq Q$ be a set of states. By
$\aut{A}\restr_{Q'}$ we denote the \emph{restriction of $\aut{A}$ to
  $Q'$} obtained by replacing the set of states by $Q'$, the set of
exits by $F\cup \left(Q\setminus Q'\right)$, the priority function by
$\Omega\restr_{Q'}$, and the transition function by
$\delta\restr_{Q'\times A}$.
\new{Let us stress that in the restricted automaton exits are
either original exits or original states not in $Q'$ (see Fig.~\ref{fig:exits}).}
We say that $\aut{B}$ is a \emph{sub-automaton of $\aut{A}$} (denoted $\aut{B}\subseteq\aut{A}$) if $\aut{B} = \aut{A}\restr_{Q^\aut{B}}$.

\begin{figure}
\centering
\begin{tikzpicture}
\tikzstyle{edge}=[draw, thick, ->]

\coordinate (zero) at (0,0);

\path ($(zero)+(0,-1.4)$) edge[lbrace] node {$Q'$} ($(zero)+(0,-0.1)$);

\path[draw=black, rounded corners=5] ($(zero)+(0,0)$) rectangle ($(zero)+(2,-3)$);

\path[draw=black, rounded corners=5] ($(zero)+(2.3,0)$) rectangle ($(zero)+(3.3,-3)$);

\node at ($(zero)+(1,0.4)$) {$Q$};
\node at ($(zero)+(2.8,0.4)$) {$F$};

\foreach \y in {1,...,4} {
  \draw[edge] ($(zero)+(1.8, 0.5-\y*0.8)$) -- ++(0.7,0);
}

\foreach \x in {1,...,2} {
  \draw[edge] ($(zero)+(\x*0.4,-1.15)$) -- ++(0,-0.7);
}

\foreach \x in {1,...,2} {
  \draw[edge] ($(zero)+(2-\x*0.4,-1.85)$) -- ++(0,0.7);
}

\coordinate (cu0) at ($(zero)+(1.0,-0.6)$);

\draw[edge] ($(cu0)+(0,0.3)$) arc (90:-90:0.3);
\draw[edge] ($(cu0)-(0,0.3)$) arc (90:-90:-0.3);

\coordinate (cu1) at ($(zero)+(1.0,-2.4)$);

\draw[edge] ($(cu1)+(0,0.3)$) arc (90:-90:0.3);
\draw[edge] ($(cu1)-(0,0.3)$) arc (90:-90:-0.3);

\draw[dotted] ($(zero)+(0,-1.5)$) -- ($(zero)+(2,-1.5)$);

\coordinate (zero) at (5,0);

\path[draw=black, rounded corners=5] ($(zero)+(0,0)$) rectangle ($(zero)+(2,-1.35)$);

\path[draw=black, rounded corners=5] ($(zero)+(2.3,0)$) --
	++(1.0,0) --
    ++(0,-3) --
    ++(-3.3,0) --
    ++(0,1.35) --
    ++(2.3,0) -- cycle;


\node at ($(zero)+(1,0.4)$) {$Q'$};
\node at ($(zero)+(2.8,0.4)$) {$F'$};

\foreach \y in {1,...,2} {
  \draw[edge] ($(zero)+(1.8, 0.5-\y*0.8)$) -- ++(0.7,0);
}

\foreach \x in {1,...,2} {
  \draw[edge] ($(zero)+(\x*0.4,-1.15)$) -- ++(0,-0.7);
}

\coordinate (cu0) at ($(zero)+(1.0,-0.6)$);

\draw[edge] ($(cu0)+(0,0.3)$) arc (90:-90:0.3);
\draw[edge] ($(cu0)-(0,0.3)$) arc (90:-90:-0.3);

\end{tikzpicture}
\caption{\new{An alternating automaton $\aut{A}$ with states $Q$ and exits $F$; and the restriction $\aut{A}\restr_{Q'}$ for $Q'\subseteq Q$. The edges illustrate the transitions of $\aut{A}$.}}
\label{fig:exits}
\end{figure}
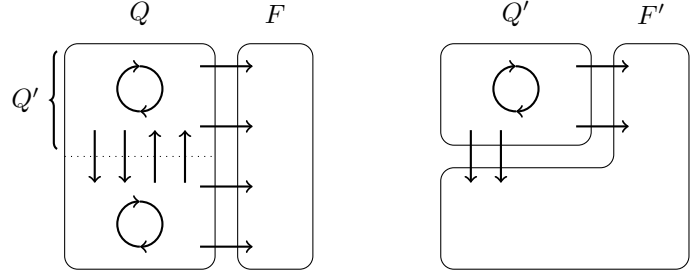

For automata $\aut{A}$, $\aut{B}$ over an alphabet $A$ with
$Q^\aut{A}\cap Q^\aut{B}=\emptyset$, we define the composition
$\aut{A}\cdot\aut{B}$ as the automaton over $A$, with states
$Q=Q^\aut{A}\cup Q^\aut{B}$, exits $\left(F^\aut{A}\cup
F^\aut{B}\right)\setminus Q$, transitions
$\delta^\aut{A}\cup\delta^\aut{B}$, and priorities
$\Omega^\aut{A}\cup\Omega^\aut{B}$.  Note here that some exits of
$\aut{A}$ may be states of $\aut{B}$ and \emph{vice versa}.

\begin{fact}
If $\aut{A}$ is an alternating automaton and $Q=Q_1\cup Q_2$ is a
partition of the states of $\aut{A}$ then
$\aut{A}\restr_{Q_1}\cdot\;\aut{A}\restr_{Q_2}=\aut{A}$.
\end{fact}

\section{Game automata}
\label{sec:game}


\new{Let ${\cal A}$ and ${\cal B}$ be automata over the same alphabet. For an occurrence of
a state (or an exit)} $q$ in a transition $\delta(p,a)$ of ${\cal A}$, and a state
\new{$q_0^{\aut{B}}$} of $\aut{B}$, the substitution ${\cal A}_{\cal B}$ is obtained
by \new{taking the disjoint union of $\aut{A}$ and $\aut{B}$ (the
  state space is the disjoint union of $Q^{\aut{A}}$ and
  $Q^{\aut{B}}$, etc.),} and replacing the occurrence of $q$ in $\delta(p,a)$ with \new{$q_0^{\aut{B}}$}. The
mapping ${\cal B} \mapsto {\cal A}_ {\cal B}$ induces an operation on
recognized languages, but it need not preserve coarser equivalence
relations, like having the same index. For a class of automata ${\autclass{C}}$ \new{over a common alphabet}, we say that \emph{substitution preserves
(alternating, nondeterministic, etc.) index in ${\autclass{C}}$} if
\new{for all ${\cal A}, {\cal B}_1, {\cal B}_2, \in {\autclass{C}}$, if
$\lang({\cal B}_1, q_0^{{\cal B}_1})$ and $\lang({\cal B}_2,
q_0^{{\cal B}_2})$ have the same (alternating, nondeterministic, etc.)
index, then so do 
$\lang(\aut{A}_{{\cal B}_1}, q_0^{\aut{A}})$ and 
$\lang(\aut{A}_{{\cal B}_2}, q_0^{\aut{A}})$ 
for any  $q_0^{\aut{A}}$.}

As pointed out in the introduction, the operation of union does not
preserve index. The same is true \new{for} intersection.  

\begin{example} \label{ex:union}
Take $A=\{0,1,2\}$ and consider \new{$\omega$-word languages}
$(A^*(1+2))^\omega$ and $(A^*2)^\omega$.  Both these languages can be
recognized by deterministic automata of index $(1,2)$, and not lower
\new{than} this. Taking union with $A^*0^\omega$, we obtain
$(A^*(1+2))^\omega \cup A^*0^\omega = A^\omega$, and $(A^*2)^\omega
\cup A^*0^\omega$. To recognize the language $(A^*2)^\omega \cup
A^*0^\omega$, a deterministic automaton requires three priorities and
an alternating one needs two. This makes it much more complex than the
whole space $A^\omega$, which can be recognized by a deterministic
automaton with a single state, whose priority is $0$. \new{Similarly,
intersecting $(A^*(1+2))^\omega$ and $(A^*2)^\omega$ with $(A^*(0 +
1))^\omega$ we obtain respectively $A^*(0^*1)^\omega$, and the empty
set}, which have very different complexity. \new{This example can be
transferred to trees by encoding $\omega$-words as \new{sequences} of
labels on the left-most branches.}
\end{example}

Example~\ref{ex:union} illustrates  a more general
phenomenon. The following notion is designed to capture how an
automaton can simulate union or intersection: 
we call a transition $\delta(q,a)$ \emph{ambiguous} if
it contains two occurrences of some direction
$d\in\{\dL,\dR\}$. Recall that a transition is trivial if it is of
the form $\bot$ or $\top$; as discussed in Section~\ref{subsec:automata},
trivial transitions are just a convenient notation for immediate
acceptance and rejection, and can be easily replaced with looping
states of appropriate priority.

\begin{fact} \label{fact:motivation}
Let $\autclass{C}$ be a class of (alternating) automata without trivial
transitions, over a fixed alphabet $A$ containing at least two
letters, that 
\begin{itemize}
\item is closed under substitution, and
\item contains automata recognizing $\emptyset$, $\trees{A}$, and some language $X$ of
non-trivial index.
\end{itemize} 
It is only possible that substitution preserves index in ${\autclass{C}}$, if no automaton in
$\autclass{C}$ contains an ambiguous transition.
\end{fact}

\begin{proof}  
Let us take an arbitrary automaton $\aut{A}\in {\autclass{C}}$ and a state
$p$ of $\aut{A}$.  Assume that $\delta(p,a)$ contains $(\dL, p_a)$ and
$\delta(p,b)$ contains $(\dL, p_b)$ for some distinct letters $a,b\in A$ and
some states $p_a, p_b$ (the remaining three cases are symmetric).
Starting from the automaton $\aut{A}$ and the automata recognizing
$\emptyset$, $\trees{A}$, and $X$, we obtain by substitution automata
$\aut{A}_a, \aut{A}_b, \aut{A}'_a, \aut{A}'_b \in {\autclass{C}}$ such that
\begin{align*}
\lang(\aut{A}_a, p)&=a(X,\trees{A}) \cup b(\trees{A}, \trees{A}),  
&\lang(\aut{A}'_a, p) &=a(X,\trees{A}), \\
\lang(\aut{A}_b, p)&=a(\trees{A},\trees{A}) \cup b(X, \trees{A}), 
&\lang(\aut{A}'_b, p)& = b(X,\trees{A}),
\end{align*}
where for $c\in A$ and $Y,Z\subseteq\trees{A}$,
\[c(Y,Z) = \left \{ t \in \trees{A} \bigm | t(\epsilon) = \new{c}, \; t\!\!\upharpoonright _\dL \in Y,
t\!\!\upharpoonright _\dR \in Z \right\}.\] Note that  
$\lang(\aut{A}_a, p) \cup \lang(\aut{A}_b, p)  = \trees{A}$ and
$\lang(\aut{A}'_a, p) \cap \lang(\aut{A}'_b, p) = \emptyset$. 

Let $\aut{B} \in {\autclass{C}}$ and let $q_0$ be a state of $\aut{B}$ such
that for some $c$, $\delta(q_0, c)$ is an ambiguous transition. By
substituting appropriately the automata recognizing $\emptyset$ and
$\trees{A}$ we can assume that $\delta(q_0,c) = (d, q_1) \lor (d,
q_2)$ or $\delta(q_0,c) = (d, q_1) \land (d, q_2)$ for some states
$q_1, q_2$, and no tree with a label $c'\neq c$ in the root is accepted from $q_0$.
Assume $\delta(q_0,c) = (d, q_1) \lor (d, q_2)$, and let
$\aut{B}'$ be the result of replacing the occurrence of $q_1$ with the
state $p$ of $\aut{A}_a$.  Now, $\lang(\aut{A}_a, p)$ and
$\lang(\aut{A}_b, p)$ have the same index, but by substituting in
$\aut{B}'$ at $(d, q_2)$ the automaton $\aut{A}_a$ or $\aut{A}_b$
(with initial state $p$), we get two languages of different index. For
$\delta(q_0,c) = (d, q_1) \land (d, q_2)$ the argument is analogous
but uses the other two automata.
\end{proof}

\medskip 

In the light of this result we propose the following definition.

\begin{definition} 
A {\em game automaton} is an alternating automaton 
without ambiguous transitions; that is, it has only transitions of the following forms: 
\[ \top\,, \;\; \bot\,, \;\; (q_\dL, \dL)\,, \;\; (q_\dR, \dR)\,, \;\; (q_\dL, \dL) \vee (q_\dR, \dR)\,, \;\; (q_\dL, \dL) \wedge (q_\dR, \dR)\]
for  $q_\dL,q_\dR\in Q\cup F$.
\end{definition}

\new{ 
In the course of the paper we shall see that for game automata
substitution preserves
the non-deterministic index (Proposition~\ref{pro:nondet-preserve}),
the alternating index (Proposition~\ref{pro:alt-preserve}), 
and the weak alternating index (Proposition~\ref{pro:weak-preserve}). Together with Fact~\ref{fact:motivation} this
will imply that game automata are the largest
non-trivial subclass of alternating automata closed under substitution
for which substitution preserves the index. 
}


The class of languages recognized by game automata is closed under
complementation: the usual complementation procedure of increasing the
priorities by one and swapping existential and universal transitions
works. However they are neither closed under union nor
intersection. For instance, let $L_{\new{c}} = \{t\in
T_{\{a,b\}}\colon t(\dL)=t(\dR)=\new{c}\}$ \new{for
$c=a,b$}. Obviously, $L_a$ and $L_b$ are recognizable by game
automata, but $L_a\cup L_b$ is not. Note that the last example also
shows that game automata do not recognize all regular languages. On
the other hand they extend across the whole alternating index
hierarchy, as they recognize so-called game languages
$\W{\rmin}{\rmax}$. We discuss this in more detail in
Section~\ref{ssec:lower}.

The main similarity between game automata and deterministic automata
is that their acceptance can be expressed in terms of \emph{runs},
which are relabelings of input trees induced uniquely by
transitions. For a total game automaton $\aut{A}$ and an initial state
$q_I$, with each partial tree $t$ one can associate the \emph{run}
\[\fun{\rho(\aut{A}, t,q_I)}{\dom(t) \cup \holes(t)}{Q^\aut{A} \cup\{\top, \bot, \ast\}}\]
such that $\rho(\varepsilon) = q_I$ and for all $v\in\dom(t)$, if
$\rho(v) = q$, $\delta(q,t(v))=b_v$, then
\begin{itemize}
\item if $b_v$ is $(q_\dL, \dL) \lor (q_\dR, \dR)$ or $(q_\dL, \dL) \land (q_\dR, \dR)$, then  $\rho(v\dL)=q_\dL$ and $\rho(v\dR)=q_\dR$;
\item if $b_v = (q_d, d)$ for some $d\in\{\dL,\dR\}$, then  $\rho(vd)=q_d$ and $\rho(v\bar d\,) = \ast$;
\item if $b_v = \bot$ then $\rho(v\dL)=\rho(v\dR)=\bot$, and dually for $\top$;
\end{itemize} 
and if $\rho(v)\in \{\top, \bot,\ast\}$, then $\rho(v\dL)=\rho(v\dR)=\ast$.
Observe that $\rho(v)$ is uniquely determined by the labels of $t$ on
the path leading to $v$.

The run $\rho=\rho(\aut{A}, t,q_I)$ is naturally interpreted as a game
$\rhogame(\aut{A}, t, q_I)$ with positions
$\dom(t)\setminus\rho^{-1}(\ast)$, final positions $\holes(t)$, where
edges follow the child relation and loop on
$\rho^{-1}(\{\top,\bot\})$, priority of $v$ is
$\Omega^\aut{A}(\rho(v))$ with $\Omega^\aut{A}(\bot)=1$,
$\Omega^\aut{A}(\top)=0$, and the owner of $v$ is \eve iff
$\delta(\rho(v), t(v)) = (q_\dL, \dL) \lor (q_\dR, \dR)$ for some
$q_\dL, q_\dR\in Q^\aut{A}$. Clearly $\rhogame(\aut{A}, t, q_I)$ is
equivalent to $\agame(\aut{A}, t, q_I)$. \new{Note that the arena of
$\rhogame(\aut{A}, t, q_I)$ is a subtree of $t$ (with additional edges
looping over positions labelled with $\top$ or $\bot$ in the run of
$\aut{A}$ on $t$). Consequently,  a strategy in $\rhogame(\aut{A}, t,
q_I)$ can be also viewed as a subtree of $t$ (we can ignore the
looping edges, as there is no choice there anyway.}  If $t$ is total, we say that
$\rho$ is \emph{accepting}, if $\eve$ has a winning strategy in
$\rhogame(\aut{A}, t, q_I)$.


\new{
A direct consequence of the acceptance being definable in terms of
runs is that each part of the automaton matters: for each state $q$
reachable from the given initial state, one can find a family of trees
that spans all possible behaviours of the automaton from the state $q$,
and the outcome of the computation depends exclusively on this
behaviour. As we show below, it can be done by taking all full trees
extending an appropriately chosen partial tree with a single hole,
corresponding to the state $q$.  
}
Let $t\in\partrees{A}$ be a partial tree and $\rho=\rho(\aut{A}, t,
q_I)$ be the run of an automaton $\aut{A}$ on $t$. We say that $t$
\emph{resolves $\aut{A}$ from $q_I\in Q^\aut{A}$} if $\rho(h) \neq
\ast$ for each hole $h$ of $t$ and whenever $t\restr_{vd}$ is the only
total tree in $\{t\restr_{v\dL}, t\restr_{v\dR}\}$, either $\rho(vd) =
\ast$ or $vd$ is losing for the owner of $v$ in $\rhogame(\aut{A}, t,
q_I)$. 
\new{The following establishes the property we discussed above and
  its analogue for transitions.}

\begin{fact}\label{ft:resolve}
Assume that $t$ resolves $\aut{A}$ from $q_I$ and $\rho=\rho(\aut{A},
t, q_I)$. If $t$ has a single hole $h$ then $t[h:=s]\in\lang(\aut{A},
q_I)$ iff $s\in\lang(\aut{A}, \rho(h))$ for all $s\in\trees{A}$. 

If $t$ has two holes $h, h'$ whose closest common ancestor $u$
satisfies $\delta_\aut{A}(\rho(u), t(u)) = (q_\dL,\dL) \land
(q_\dR,\dR)$ for some $q_\dL, q_\dR$, then $t[h:=s,
h':=s']\in\lang(\aut{A}, q_I)$ iff $s\in\lang(\aut{A}, \rho(h))$ and
$s'\in\lang(\aut{A}, \rho(h'))$ for all $s,s'$; dually for
$(q_\dL,\dL) \lor (q_\dR,\dR)$.  
\end{fact}

\begin{proof}
Let us prove the first claim. There are two cases. 
\begin{itemize}
\item 
One of the players $P\in\{\eve,\adam\}$ has a winning strategy
$\sigma$ in the game associated to $\rho$ such that node $h$ does not
belong to $\sigma$ \new{viewed as a subtree of $t$}. In that case,
there exists an ancestor $u$ of $h$ such that the player $P$ owns $u$
and $\sigma$ moves to from $u$ to $ud$ such that $ud$ is not ancestor
of $h$. In that case $\sigma$ is a winning strategy for $P$ in the
subtree under $ud$, which contradicts the definition of a resolving
tree.
\item 
Whenever $\sigma$ is a winning strategy for a player
$P\in\{\eve,\adam\}$ in the game associated to $\rho$, then the node
corresponding to $h$ belongs to $\sigma$ \new{viewed as a subtree of
$t$}. Take any total tree $s$. If $t[h:=s]\in \lang(\aut{A}, q_I)$
then $\eve$ has a winning strategy in the game associated with
$\rho(\aut{A}, t[h:=s], q_I)$. In particular, she can win from the
position $h$ in this game. Therefore, by the definition,
$s\in\lang(\aut{A}, \rho(\new{h}))$. If $t[h:=s]\notin \lang(\aut{A}, q_I)$
then the property is symmetrical: $\adam$ has a winning strategy and
$s\notin \lang(\aut{A}, \rho(h))$.
\end{itemize}

For the second claim, it follows easily that in this case the trees
$t\restr_{u\dL}$, $t\restr_{u\dR}$ and the tree obtained by putting a
hole in $t$ instead of $u$, resolve $\aut{A}$ from $q_\dL$, $q_\dR$,
and  $q_I$, respectively. We obtain the second claim by applying the
first claim three times. 
\end{proof}


%% file: 4_nondet_index.tex
\section{Non-deterministic index problem}
\label{sec:ndindex}

The decidability of the non-deterministic index problem for languages recognized by game automata is an immediate consequence of the decidability of the non-deterministic index problem for deterministic tree languages~\cite{niwinski_deterministic} and the following observation. 

\begin{proposition}
\label{prop:nondetgame}
For each game automaton $\aut{A}$ and a state $q_I^\aut{A}\in Q^\aut{A}$ one can effectively construct a deterministic automaton $\aut{D}$ with an initial state $q_I^\aut{D}$, such that $\lang(\aut{A},q_I^\aut{A})$ is recognized by a non-deterministic automaton of index $(i,j)$ if and only if so is $\lang(\mathcal{D},q_I^\aut{D})$.
\end{proposition}

\begin{proof}
Essentially, $\aut{D}$ recognizes the set of winning strategies for \eve in games induced by the runs of $\aut{A}$. For two total trees $t\in \trees{A}$, $s\in \trees{B}$ let $t\otimes s \in \trees{A\times B}$ be given by $(t \otimes s) (v) = (t(v), s(v))$. Let $W^\eve_{\aut{A}, q_I}$ be the set of all total trees $t \otimes s$ over the alphabet $A^\aut{A} \times \{ \dL,\dR, \star\}$ such that $s$ encodes a winning strategy for $\eve$ in the game $\rhogame(\aut{A}, t, q_I)$ in the following sense: if $s(v)\in\{\dL,\dR\}$, $\eve$ should choose $v\cdot s(v)$, and $s(v)=\star$ means that \eve has no choice in $v$. It is easy to see that $W^\exists_{\mathcal{A}, q_I}$ can be recognized by a deterministic automaton $\aut{D}$. It inherits the state-space and the \new{priority} function from $\aut{A}$ and its transitions are modified as follows: for all $q\in Q$, $a\in A$, $d\in \{\dL,\dR\}$, if $\delta_\aut{A}(q,a) = (q_\dL, \dL) \lor (q_\dR, \dR)$ for some $q_\dL, q_\dR$, then 
\[ \delta(q, (a,d)) = (q_{\new{d}}, d)\,, \quad \delta(q,(a,\star)) = \bot\,,\] 
otherwise, \[\delta(q,(a,d)) = \bot\,, \quad \delta(q,(a,\star)) = \delta_\aut{A}(q,a)\,.\] 
It is easy to check that $\lang(\aut{D},q_I) = W^\exists_{\aut{A}, q_I}$.

Note that
\[\lang(\mathcal{A},q_I^\aut{A}) = \big \{ t\in\trees{A^\aut{A}} \bigm | \exists\, s. \ t\otimes s \in W^\eve_{\mathcal{A}, q_I^\aut{A}} \big \}.\]
Hence, if $W^\exists_{\mathcal{A}, q_I^\aut{A}} = \lang(\aut{B}, q_I^\aut{B})$ for some non-deterministic automaton $\aut{B}$ then $\lang(\aut{A},q_I^\aut{A}) = \lang(\aut{B}',q_I^\aut{B})$, where $\aut{B}'$ is the standard projection of $\aut{B}$ on the alphabet $A^\aut{A}$: for all $q\in Q^\aut{A}$ and $a \in A^\aut{A}$, $\;\delta^{\aut{B}'}(q, a) = \delta^{\aut{B}}(q, (a,\dL)) \lor \delta^{\aut{B}}(q, (a,\dR))\lor \delta^{\aut{B}}(q, (a, \star))$. The projection does not influence the index. 

For the other direction, the proof is based on the following observation. For $t \in \trees{A^\aut{A}}$ and $s \in\trees{ \{ \dL,\dR, \star\}}$ let $ t\odot s \in \trees{A^\aut{A}}$ be the tree obtained from $t$ by the following operation:  for each $v$, if $\rho_{t,q_I} (v) = q$,  $\delta(q, t(v)) = (q_\dL, \dL) \lor (q_\dR, \dR)$, and $s(v)=\dL$,  then replace the subtree of $t$ rooted at $v\dR$ by some fixed regular tree in the complement of $\lang(\aut{A}, q_\dR)$; dually for $s(v)=\dR$.  (Recall that $\aut{A}$ has only non-trivial states, so $\lang(\aut{A}, q_\dR) \subsetneq \trees{A^\aut{A}}$\new{; by the Rabin's theorem the complement of $\lang(\aut{A},q_\dR)$ contains a regular tree}.) If $s$ encodes a strategy $\sigma_s$ for \eve in $\rhogame(\aut{A}, t, q_I^\aut{A})$, then $\sigma_s$ is winning if and only if $t\odot s \in \lang(\aut{A},q_I^\aut{A})$. Hence, $t\otimes s \in W^\exists_{\mathcal{A}, q_I^\aut{A}}$ if and only if  $s$ encodes a strategy for \eve in $\rhogame(\aut{A}, t, q_I^\aut{A})$ and $t\odot s \in \lang(\aut{A},q_I^\aut{A})$. These conditions can be checked by a non-deterministic automaton of index $(i,j)$ as soon as $\lang(\aut{A},q_I^\aut{A})$ can be recognized by such an automaton.

\ident{\strat}{St}
\ident{\expstrat}{StE}
\ident{\expgoodstrat}{StEW}
\ident{\goodstrat}{StW}

Indeed, assume that $\lang(\aut{A},q_I) = \lang(\aut{B}, q^\aut{B}_I)$ for some non-deterministic automaton $\aut{B}$ of index $(\rmin,\rmax)$. \new{ To construct a non-deterministic automaton $\aut{C}$ of index at most $(\rmin,\rmax)$ recognizing $W^\exists_{\aut{A}, q_I}$, we first define a sequence of auxiliary languages and argue that each of them can be recognized by such an automaton. Let}
\begin{align*}
\strat &=\left\{t\otimes s:\ \text{$s$ is a strategy for \eve in $\rhogame(\aut{A},t,q_I)$}\right\},\\
\expstrat&=\left\{t\otimes s\otimes t':\ t\otimes s\in \strat \wedge t\odot s=t'\right\},\\
\expgoodstrat&=\left\{t\otimes s\otimes t'\in \expstrat:\ t'\in\lang(\aut{B},q_I^B)=\lang(\aut{A},q_I)\right\},\\
\goodstrat&=\left\{t\otimes s\in \strat:\ t\odot s\in\lang(\aut{A},q_I)\right\}.
\end{align*}
\new{Then, }
\begin{itemize}
\item $\strat$ corresponds to a safety condition that can be verified both by a deterministic automaton of index $(0,1)$ and by a deterministic automaton of index $(1,2)$,
\item $\expstrat$ additionally enforces that the respective subtrees equal $t_q$, as above it can be checked both by a deterministic automaton of index $(0,1)$ and by a deterministic automaton of index $(1,2)$,
\item $\expgoodstrat$ can be recognized by a product of automata recognizing $\expstrat$ and $\aut{B}$---the resulting non-deterministic automaton can be constructed in such a way that its index is $(\rmin,\rmax)$,
\item $\goodstrat$ is obtained as the projection of $\expgoodstrat$ onto the first two coordinates, as such can also be recognized by a non-deterministic $(\rmin,\rmax)$-automaton.
\end{itemize}
\new{It remains to show that} 
\begin{equation}
W^\exists_{\mathcal{A}, q_I} = \goodstrat \nonumber
\end{equation} 

First assume that $t\otimes s \in W^\exists_{\mathcal{A}, q_I}$. In that case $s$ encodes a winning strategy $\sigma$ for \eve in $\rhogame(\aut{A}, t, q_I)$. Let $t'=t\odot s$ and $D=\dom(\sigma)$ be the set of vertices belonging to $\sigma$. Note that if $v\in D$ then $t(v)=t'(v)$, so also $\rho_{t,q_I}(v)=\rho_{t',q_I}(v)$. Therefore, the strategy $\sigma$ is also winning in $\rhogame(\aut{A}, t', q_I)$. So $t'\in \lang(\aut{A}, q_I)$\new{, which} implies that $t\otimes s\otimes t'\in \expgoodstrat$ and $t\otimes s\in\goodstrat$.

Now assume that $t\otimes s\in \goodstrat$. Let $t'=t\odot s$ and $\sigma$ be the strategy for \eve in $\rhogame(\aut{A},t, q_I)$ encoded by $s$. By the definition of $\expgoodstrat$ we obtain that $t'\in\lang(\aut{A},q_I)$ so there exists a winning strategy $\sigma'$ for \eve in $\rhogame(\aut{A}, t', q_I)$. Similarly as above, let $D$ (resp. $D')$ be the set of vertices in $\sigma$ (resp. $\sigma'$). If $D'\not\subseteq D$ then there exists a minimal (w.r.t. the prefix order) vertex $v\in D'\setminus D$. By the definition of $t\odot s$ we obtain that $t'\restr_v$ is $t_q$ for $q=\rho(\aut{A},t,q_I)(v)$. Therefore, since $t_q\notin\lang(\aut{A},q)$, so there is no winning strategy for \eve in $\rhogame(\aut{A}, t_q, q)$ and we obtain a contradiction. Therefore $D'\subseteq D$ and for every $v\in D'$ we have $\rho(\aut{A},t,q_I)(v)=\rho(\aut{A},t',q_I)(v)$, so $\sigma'$ is also a strategy in $\rhogame(\aut{A}, t', q_I)$. Since strategies form an anti-chain with respect to inclusion, so $\sigma=\sigma'$, $t'\in\lang(\aut{A},q_I)$, and $t\otimes s \in  W^\exists_{\mathcal{A}, q_I}$.
\end{proof}

As a direct corollary from the proof of Proposition~\ref{prop:nondetgame} and the criteria for the non-deterministic index of deterministic languages \cite{niwinski_deterministic}, we obtain the converse of Fact~\ref{fact:motivation} for the non-deterministic index. 

\begin{proposition}\label{pro:nondet-preserve}
For game automata, substitution preserves the non-deterministic index.
\end{proposition}

\begin{proof} 
For a given game automaton $\aut{A}$, let $\aut{A}'$ be the deterministic automaton constructed in the proof of Proposition~\ref{prop:nondetgame}. Recall that $\aut{A}'$ is obtained from $\aut{A}$ by adding some transitions of the form $\bot$, and uncoupling each disjunctive transition over letter $a$ into two non-branching transitions over letters $(a,\dL)$ and $(a,\dR)$; the state-space remains the same. It follows that the construction commutes with substitution: $(\aut{A}_{\aut{B}})'$ coincides with $(\aut{A}')_{\aut{B}'}$ for all game automata $\aut{A}$ and $\aut{B}$. 

The construction was designed to preserve the non-deterministic index of the recognized language.
Consequently, assuming that substitution preserves the index for deterministic automata, we can show that the same holds for game automata. Indeed, if $\lang(\aut{B}, q_I^{\aut{B}})$ and $\lang(\aut{C}, q_I^{\aut{C}})$ have the same index, then  
$\lang(\aut{B}', q_I^{\aut{B}})$ and $\lang(\aut{C}', q_I^{\aut{C}})$ have the same index. If substitution preserves the index for deterministic automata, we can conclude that $\lang((\aut{A}_{\aut{B}})', q_I^{\aut{A}})=\lang(\aut{A'}_{\aut{B}'}, q_I^{\aut{A}})$ and $\lang((\aut{A}_{\aut{C}})', q_I^{\aut{A}}) = \lang(\aut{A'}_{\aut{C}'}, q_I^{\aut{A}})$ have the same index. Hence, $\lang(\aut{A}_{\aut{B}}, q_I^{\aut{A}})$ and $\lang(\aut{A}_{\aut{C}}, q_I^{\aut{A}})$ have the same index. 

Preservation of the index for deterministic automata follows immediately from the characterization of the levels of the index hierarchy among deterministic languages \cite{niwinski_deterministic}: it asserts that for every $(i,j)$, a language $\lang(\aut{B}, q_I^{\aut{B}})$ is recognized by a non-deterministic automaton of index $(i,j)$ if and only if the automaton $\aut{B}$ does not contain a certain characteristic, strongly connected subgraph reachable from $q_I^{\aut{B}}$. Consequently, if $\lang(\aut{B}, q_I^{\aut{B}})$ and $\lang(\aut{C}, q_I^{\aut{C}})$ have the same index, $\aut{B}$ and $\aut{C}$ contain the same characteristic subgraphs reachable from the respective initial states.  By the definition of substitution, no strongly connected subgraph in $\aut{A}_{\aut{B}}$ can use states from $\aut{A}$ and from $\aut{B}$. Consequently, $\aut{A}_{\aut{B}}$ and $\aut{A}_{\aut{C}}$ contain the same characteristic subgraphs reachable from $q_I^{\aut{A}}$, and so $\lang(\aut{A}_{\aut{B}}, q_I^{\aut{A}})$ and $\lang(\aut{A}_{\aut{C}}, q_I^{\aut{A}})$ have the same index. 
\end{proof}


%% file: 5_alt_index.tex
\section{Alternating index problem}
\label{sec:altindex}

In this section we show that the alternating index problem is decidable for game automata. Let us start with some notation.

\begin{definition}
For $\rmin< \rmax\in\N$, let $\RM(\rmin, \rmax)$ denote the class of languages recognized by alternating tree automata of index $(\rmin, \rmax)$. Let
\begin{align*}
\RMeven{\rmax}&=\RM(0, \rmax),\\
\RModd{\rmax}&=\RM(1,\rmax+1),\\
\RMdelta{\rmax} &= \RM(0,\rmax) \cap \RM(1,\rmax+1).
\end{align*}
The above classes are naturally ordered by inclusion.
\end{definition}


The result we prove not only gives decidability of the alternating index problem but also shows that languages recognizable by game automata collapse inside the $\RMdelta i$ classes. To express it precisely we recall the so-called \emph{comp} classes~\cite{arnold_separation} that can be defined in terms of strongly connected components (SCCs) of a graph naturally associated with each alternating automaton.

\begin{definition}
Let $\aut{A}$ be an alternating automaton. Let $\graph(\aut{A})$ be the directed edge-labelled graph over the set of vertices $Q$ such that there is an edge $p \tran{(a,d)} q$ whenever $(q,d)$ occurs in $\delta(p,a)$. Additionally, vertices of $\graph(\aut{A})$ are labelled by values of $\Omega$. We write $p\tran{w}q$ if there is a path in  $\graph(\aut{A})$ whose edge-labels yield the word $w$.
\end{definition}

\begin{definition}\label{def:comp}
An alternating automaton $\aut{A}$ is in $\cmp(\rmin,\rmax)$ if (ignoring edge-labels) each SCC in $\graph(\aut{A})$ has priorities between $\rmin$ and $\rmax$ or between $\rmin+1$ and $\rmax+1$.
\end{definition}

It follows from the definition that each $\cmp(\rmin,\rmax)$ automaton is a $(\rmin,\rmax+1)$ automaton, and can be transformed into an equivalent $\cmp(\rmin+1,\rmax+2)$ automaton by scaling the priorities. We write $\cmp_{j}$ for the class of languages recognized by $\cmp(0, j)$ automata. We then have
\[\RMeven{j} \cup \RModd{j} \subseteq \cmp_{j} \subseteq \RMdelta{j+1}\,.\]
The class $\cmp_0$ corresponds to the class of weak alternating automata. 
An important result obtained in~\cite{rabin_separation} states that $\RMdelta{1}$ coincides with
the class of languages definable in weak monadic second order logic (WMSO).
Since WMSO definability and weak recognizability are coextensive
concepts \cite{MullerSS86}, Rabin's result proves that classes $\cmp_{0}$ and $\RMdelta{1}$ coincide.
However, as shown by Arnold and Santocanale~\cite{arnold_separation}, for higher levels the inclusion is strict
\[\cmp_{j} \subsetneq \RMdelta{j+1} \quad \text{ for } j>0\,,\] 
i.e., there are examples of regular languages in $\RMdelta{j+1}$ but not in $\cmp_{j}$. It turns out that, as a consequence of our characterization, in the case of languages recognizable by game automata the respective classes $\cmp_{j}$ and $\RMdelta{j+1}$ coincide for all levels.

\begin{theorem}\label{thm:alt_index}
For each game automaton $\aut{A}$ and an initial state $q_I$, the language $\lang(\aut{A}, q_I)$ belongs to exactly one of the classes: $\cmp_{0}$, $\RMeven{i}\setminus\RModd{i}$, $\RModd{i}\setminus \RMeven{i}$, or $\cmp_{i} \setminus\left(\RMeven{i} \cup \RModd{i}\right )$, for $i>0$. 
Moreover, it can be effectively decided which class it is and an automaton from this class can be constructed.
\end{theorem}

The rest of this section is devoted to showing this result. Section~\ref{ssec:algo} describes a recursive procedure to compute the \emph{class} of the given language $\lang(\aut{A},q_I)$, i.e., $\RMeven{i}$, $\RModd{i}$, or $\cmp_{i}$, depending on which of the possibilities holds. Sections~\ref{ssec:upper},~\ref{ssec:lower} show that the procedure is correct. The estimation of Section~\ref{ssec:upper} is in fact an effective construction of an automaton from the respective class.

\subsection{The algorithm}
\label{ssec:algo}

Let $\aut{A}$ be an alternating automaton of index $(\rmin, \rmax)$. For $n\in\N$ we denote by $\aut{A}^{\geq n}$ the sub-automaton obtained from $\aut{A}$ by restricting to states of priority at least $n$. Observe that the index of $\aut{A}^{\geq n}$ is at most $(n,\rmax)$. A sub-automaton $\aut{B}\subseteq\aut{A}$ is an \emph{$n$-component of $\aut{A}$} if $\graph(\aut{B})$ is a strongly connected component of $\graph(\aut{A}^{\geq n})$. We say that $\aut{B}$ is \emph{non-trivial} if $\graph(\aut{B})$ contains at least one edge. Our algorithm computes the class of each $n$-component $\aut{B}$ of $\aut{A}$, based on the classes of $(n+1)$-components of $\aut{B}$ and transitions between them. (We shall see that for $n$-components the class does not depend on the initial state.)

We begin with a simple preprocessing. An automaton $\aut{A}$ is \emph{priority-reduced} if for all $n>0$, each $n$-component of $\aut{A}$ is non-trivial and contains a state of priority $n$.

\begin{lemma}
\label{lemma:reduced}
Each game automaton can be effectively transformed into an equivalent priority-reduced game automaton.
\end{lemma}

\begin{proof}
We iteratively decrease priorities in the $n$-components of $\aut{A}$, for $n\geq 1$. As long as there is an $n$-component that is not priority-reduced, pick any such $n$-component, if it is trivial, set all its priorities to $n-1$, if it is non-trivial but does not contain a state of priority $n$, decrease all its priorities by $2$ (this does not influence the recognized language). After finitely many steps the automaton is priority-reduced. Note that no trivial states are introduced. 
\end{proof}

The main algorithm uses three simple notions. An $(n+1)$-component $\aut{B}_0$ of $\aut{B}$ is \emph{\eve-branching} if $\aut{B}$ contains a transition
\[\delta(p,a) = (q_\dL,\dL) \lor (q_\dR,\dR)\]
with $p, q_\dL \in Q^{\aut{B}_0}$ or $p, q_\dR\in Q^{\aut{B}_0}$. For $\adam$ replace $\lor$ with $\land$.

For a class $K$, operations $K^{\eve}$ and $K^{\adam}$ are defined as
\begin{align*}
\left(\RMeven m\right)^{\eve} = \left(\RModd{m-1}\right)^{\eve} = \left(\cmp_{m-1}\right)^{\eve} = \RMeven{m},\\
\left(\RModd m\right)^{\adam} = \left(\RMeven{m-1}\right)^{\adam} = \left(\cmp_{m-1}\right)^{\adam} = \RModd{m}.
\end{align*}
We write $\bigvee_{\ell=1}^k K_\ell$ for the largest class among $K_1, K_2, \dots, K_\ell$ if it exists, or $\cmp_{m}$ if among these classes there are two maximal ones, $\RMeven{m}$ and $\RModd{m}$.

Let $\aut{A}$ be a priority-reduced game automaton of index $(\rmin, \rmax)$. The algorithm starts from $n=\rmax$ and proceeds downward. Let $\aut{B}$ be an $n$-component. 

\begin{itemize}
\item If $\aut{B}$ has only states of priority $n$, set $\class(\aut{B})=\cmp_0$.

\item If $\aut{B}$ has no states of priority $n$, it coincides with a single $1$-component $\aut{B}_1$. Set $\class(\aut{B})=\class(\aut{B}_1)$.

\item Otherwise,  assume that $n$ is even (for odd $n$ replace \eve with \adam).  Let $\aut{B}_1,\aut{B}_2,\ldots,\aut{B}_k$, be the $(n+1)$-components of $\aut{B}$ that are \eve-branching, and let $\aut{C}_1,\aut{C}_2,\ldots,\aut{C}_{k'}$ be the ones that are not  \eve-branching.  We set 
\[\class(\aut{B}) = \bigvee_{\ell=1}^k \class(\aut{B}_\ell)^{\eve} \vee \bigvee_{\ell=1}^{k'} \class(\aut{C}_\ell) \,,\] 
\end{itemize}
Let $\class(\aut{A}, q_I) = \bigvee_{\ell=1}^k \class(\aut{A}_\ell)$ where $\aut{A}_1, \aut{A}_2, \dots, \aut{A}_k$ are the $i$-components  of  $\aut{A}$ reachable from $q_I$ in $\graph(\aut{A})$.

\subsection{Upper bounds}
\label{ssec:upper}

In this subsection we show that $L(\aut{A},q_I)$ can be recognized by a $\class(\aut{A},q_I)$-automaton. The argument will closely follow the recursive algorithm, pushing through an invariant guaranteeing that  each $n$-component $\aut{B}$ of $\aut{A}$ can be replaced with an ``equivalent''  $\class(\aut{B})$-automaton. The notion of equivalence for non-total automata is formalized by \herewasA simulations. 

\begin{definition} \label{def:Asim}
An alternating automaton $\aut{S}$ \emph{\herewasA simulates} a game automaton $\aut{A}$ if  $F^{\aut{S}}\subseteq F^{\aut{A}}$ and 
there exists an embedding $\fun{\iota}{Q^{\aut{A}}}{Q^{\aut{S}}}$ (usually $Q^{\aut{A}}\subseteq Q^{\aut{S}}$)  such that for all $t\in\trees{A}$, $q_I^\aut{A}\in Q^{\aut{A}}$,  and for each winning strategy $\sigma$ for player $P$ in  $\agame(\aut{A}, t, q_I^\aut{A})$ there is a winning strategy $\sigma^S$ for $P$ in  $\agame(\aut{S}, t, \iota(q_I^\aut{A}))$ such that the guarantee of $\sigma^S$ is contained in the guarantee of $\sigma$, and if there is an infinite play conforming to $\sigma^S$ then there is an infinite play conforming to $\sigma$.
\end{definition}

\noindent Note that if $\aut{A}$ and $\aut{S}$ are total and $\aut{S}$ \herewasA simulates $\aut{A}$ then $\lang(\aut{A}, q_I^\aut{A})=\lang(\aut{S}, \iota(q_I^\aut{A}))$.

\begin{lemma}
For each $n$-component $\aut{B}$ of a game automaton $\aut{A}$, 
$\aut{B}$ can be \herewasA simulated by a $\class(\aut{B})$-automaton.
\end{lemma}

\begin{proof} Assume that the index of $\aut{A}$ is $(i,j)$. We proceed by induction on $n=j, j-1, \dots, i$. If all states of $\aut{B}$ have priority $n$ or all have priority strictly greater \new{than} $n$,  the claim is immediate.
Let us assume that neither is the case. By symmetry it is enough to give the construction for even $n$.

Suppose $\aut{B}$ has only \eve-branching $n+1$ components,  $\aut{B}_1, \aut{B}_2, \dots, \aut{B}_k$.  Then $\class(\aut{B}) = \bigvee_{\ell}\class(\aut{B}_\ell)^{\eve} = \RMeven m$ for some $m\geq 1$. By the inductive hypothesis  we get a $\class(\aut{B}_\ell)$-automaton $\aut{B}_\ell^S$,  \herewasA simulating $\aut{B}_\ell$. Since  $\RMeven m \geq \class(\aut{B}_\ell)^{\eve}$, $\aut{B}_\ell^S$ can be assumed to be an $(n,n+m)$-automaton.  Hence, we can put 
\[\aut{B}^S =  \aut{B} \restr_{\Omega^{-1}(n)} \cdot \,\aut{B}_1^S \cdot  \aut{B}_2^S \cdot \ldots \cdot \aut{B}_k^S\] 
to get an $(n,n+m)$-automaton \herewasA simulating $\aut{B}$.

Now, assume that $\aut{B}$ contains also $n+1$ components  $\aut{C}_1, \aut{C}_2, \dots, \aut{C}_{k'}$ that are not \eve-branching. Repeating the construction above would now result in an automaton of index $\bigvee_\ell \class(\aut{B}_\ell)^{\eve} \vee \bigvee_\ell \class(\aut{C}_\ell)^{\eve} $, potentially higher than $\class(\aut{B}) = \bigvee_\ell \class(\aut{B}_\ell)^{\eve} \vee \bigvee_\ell \class(\aut{C}_\ell)$. Hence, instead of $\aut{C}^S_\ell$ we shall use $\aut{C}^R_\ell \cdot \aut{C}^T_\ell$, where
\begin{itemize}
\item $\aut{C}_\ell^T$ is a copy of $\aut{C}_\ell^S$ with each transition leading to an exit of $\aut{C}_\ell^S$ that is not an exit of $\aut{B}$, replaced  with a transition to $\top$ (losing for \adam);  
\item $\aut{C}^R_\ell$ is $\aut{C}^S_\ell$ with all priorities set to $n$ and additional $\varepsilon$-transitions (which can be eliminated in the usual way): for each state $q$ of $\aut{C}_\ell^R$ allow \adam to decide to stay in $q$ or move to the copy of $q$ in $\aut{C}_\ell^T$ (treated as an exit in $\aut{C}_\ell^R$). 
\end{itemize}
Thus, 
\[\aut{B}^S =  \aut{B} \restr_{\Omega^{-1}(n)} \cdot \,\aut{B}_1^S \cdot \ldots \cdot \aut{B}_k^S \cdot \aut{C}_1^R \cdot \aut{C}_1^T \cdot  \ldots \cdot \aut{C}_{k'}^R \cdot \aut{C}_{k'}^T\,.\] 
The composition of automata $\aut{B} \restr_{\Omega^{-1}(n)}$,  $\aut{B}^S_\ell$, $\aut{C}^R_\ell$ gives a $\class(\aut{B})$-automaton (each $\aut{C}^S_\ell$ was replaced with an $(n,n)$-automaton $\aut{C}^R_\ell$). This is further composed with $\class(\aut{C}_\ell)$-automata $\aut{C}^T_\ell$ in a loop-less way.  Hence, $\aut{B}^S$ is a $\class(\aut{B})$-automaton. 

Let us see that $\aut{B}^S$ \herewasA simulates $\aut{B}$. Let $\iota$ be defined as identity on $\aut{B}\restr_{\Omega^{-1}(n)}$, on $Q^{\aut{B}_\ell}$ as the embedding $Q^{\aut{B}_\ell} \to Q^{\aut{B}^S_\ell}$, and on  $Q^{\aut{C}_\ell}$ as the embedding $Q^{\aut{C}_\ell} \to Q^{\aut{C}^R_\ell}$. Consider a tree $t\in \trees{A}$, a state $q_I^\aut{B}$ of $\aut{B}$, and games $\agame(\aut{B},t,q_I^\aut{B})$ and $\agame(\aut{B}^S,t,\iota(q_I^\aut{B}))$.

First, consider a strategy $\sigma$ for \eve in $\agame(\aut{B},t,q_I^\aut{B})$. We decompose this strategy into parts corresponding to the sub-automata $\aut{B}_\ell$ and $\aut{C}_\ell$, for each part we use the fact that $\aut{B}_\ell^S$ \herewasA simulates $\aut{B}_\ell$ and $\aut{C}_\ell^S$ \herewasA simulates $\aut{C}_\ell$. This gives us a strategy for \eve on parts of $\agame(\aut{B}^S, t, \iota(q_I^\aut{B}))$ corresponding to sub-automata $\aut{B}_\ell^S$, $\aut{C}_\ell^R$, $\aut{C}_\ell^T$. Outside of $\aut{B}_\ell^S$, $\aut{C}_\ell^R$, and $\aut{C}_\ell^T$, \eve has the same choices in $\aut{B}^S$ as in $\aut{B}$. Therefore, she can make her choices according to $\sigma$. This gives a complete strategy $\sigma^S$. Now consider any play conforming to $\sigma^S$. Such a play either visits infinitely many times a state of priority $n$ in $\aut{B}^S$, and so is winning for \eve, or from some point on it stays in some sub-automaton $\aut{B}_\ell^S$, $\aut{C}_\ell^R$ or $\aut{C}_\ell^T$. In this case the play is also winning for \eve, by the assumption on $\sigma$ and by the fact that all the changes of priorities in $\aut{C}_\ell^R\,$'s and transitions in $\aut{C}_\ell^T\,$'s are favourable to \eve. By the definition of $\sigma^S$, the guarantee of $\sigma^S$ is contained in the guarantee of $\sigma$, and if there is an infinite play conforming to $\sigma^S$ then there is an infinite play conforming to $\sigma$.

For a winning strategy $\sigma$ for \adam in $\agame(\aut{B},t,q_I^\aut{B})$, we construct a winning strategy $\sigma^S$ for \adam in $\agame(\aut{B}^S,t,\iota(q_I^\aut{B}))$ as follows:
\begin{itemize}
\item in positions corresponding to states of priority $n$ in $\aut{B}$ the strategy $\sigma^S$ follows the decisions of $\sigma$;
\item in components $\aut{B}_\ell^S$, $\aut{C}_\ell^R$, $\aut{C}_\ell^T$ the strategy $\sigma^S$ simulates $\sigma$ (using the fact that $\aut{C}_\ell^R$ and $\aut{C}_\ell^T$ have the same states and exits as the automaton $\aut{C}_\ell^S$ that \herewasA simulates $\aut{C}_\ell$) with the following exception:
\adam immediately moves from $\aut{C}^R_\ell$ to $\aut{C}^T_\ell$ whenever each extension of the current play, conforming to the simulating strategy, stays forever in $\aut{C}^R_\ell$ (possibly reaching an exit that is also an exit of $\aut{B}^S$).
\end{itemize}
An easy inductive argument shows that 
\begin{enumerate}
\item each position $(v,p)$ with $p\in \aut{B}\restr_{\Omega^{-1}(n)}$ that is reached in some play conforming to $\sigma^S$ is also reached in some play in $\agame(\aut{B}, t, q_I^\aut{B})$ conforming to  $\sigma$;
\item whenever a play conforming to $\sigma^S$ enters $\aut{B}_\ell^S$ (resp.~$\aut{C}^R_\ell$) in a position $(v,p)$, then $p=\iota(q)$ for some $q\in \aut{B}_\ell$ (resp.~$q\in \aut{C}_\ell$) and $(v,q)$ is reached in some play in $\agame(\aut{B}, t,q_I^\aut{B})$ conforming to $\sigma$. 
\end{enumerate}


Consider any play $b^S$ conforming to $\sigma^S$.

Assume that $b^S$ is a finite play leading to a final position $(v,f)$. Unless $(v,f)$ is entered directly from some $\aut{C}_\ell^T$, by the two observations above (and by the definition of $\sigma^S$) it follows that $(v,f)$ can also be reached in some play conforming to $\sigma$. Assume that $(v,f)$ is entered directly from some $\aut{C}_\ell^T$. Let $(w,\iota(q))$ be the last moment when $b^S$ entered $\aut{C}_\ell^R$ (recall that $\aut{C}_\ell^T$ is only entered from $\aut{C}_\ell^R$). Since $\sigma^S$ in $\aut{C}_\ell^R$ and  $\aut{C}_\ell^T$ mimics the simulating strategy in $\aut{C}_\ell^S$, the final position $(v,f)$ can be reached in some play in $\agame(\aut{B}, t, q_I^\aut{B})$ starting in $(w,q)$, conforming to $\sigma$. By observation 2 it follows that $(v,f)$ is reached in a play conforming to $\sigma$ and starting in $(\varepsilon, q_I^\aut{B})$.

The remaining case is when $b^S$ is an infinite play. Should $b^S$ visit infinitely often positions of priority $n$, by the observation 2 and by the definition of $\sigma^S$ we would define a play in $\agame(\aut{B}, t, q_I^\aut{B})$ conforming to $\sigma$ that visits infinitely often positions of priority $n$. This is impossible since $\sigma$ is winning for \adam. It follows that from some point on $b^S$ stays in some sub-component. If the sub-component is $\aut{B}^S_\ell$,  \adam wins as he is playing with a winning strategy in $\aut{B}_\ell^S$.  The other possibility is that $b^S$ stays forever in $\aut{C}_\ell^R \cdot \aut{C}_\ell^T$ for some $\ell$. Since $\aut{C}_\ell$ is not \eve-branching, in each transition of the form $\delta(p,a) = (q_\dL,\dL) \lor (q_\dR,\dR)$ with $p\in Q^{\aut{B}_\ell}$, at least one of the states $q_\dL,q_\dR$ is an exit state in $\aut{B}$, or both are outside of $\aut{C}_\ell$. Hence, after entering $\aut{C}_\ell$, $\sigma$ becomes a single path in $\aut{C}_\ell$, with all the branchings (choices of \eve) going directly to exits of $\aut{B}$. In general, this path may end in a position belonging to \eve, such that both choices lead outside of $\aut{C}_\ell$ (not necessarily to exits of $\aut{B}$.) In our case the path must stay in $\aut{C_\ell}$ forever: since $b^S$ is infinite and stays forever in  $\aut{C}_\ell^R \cdot \aut{C}_\ell^T$, there is an infinite play conforming to the strategy simulating $\sigma$ in $\aut{C}^S_\ell$ and, by Definition~\ref{def:Asim}, an infinite play conforming to $\sigma$ in $\aut{C}_\ell$. Consequently, all exits reachable with $\sigma$ in $\aut{C}_\ell$ are also exits of $\aut{B}$. Hence, as soon as $b^S$ enters $\aut{C}^R_\ell$ for the last time, $\sigma^S$ tells  \adam to move to $\aut{C}^T_\ell$ where \adam wins all infinite plays.
\end{proof}

\medskip

It follows easily that $\lang(\aut{A}, q_I)$ can be recognized by a $\class(\aut{A}, q_I)$-automaton: the automaton can be obtained as a loop-less composition of the $\class(\aut{A}_\ell)$-automata \herewasA simulating the $i$-components $\aut{A}_\ell$ of $\aut{A}$ reachable from $q_I$. In other words, the alternating index bounds \new{as} computed by the algorithm in Section~\ref{ssec:algo} are correct. 

\subsection{Lower bounds}
\label{ssec:lower}

It remains to see that $\lang(\aut{A},q_I)$ cannot be recognized by an alternating automaton of index lower than $\class(\aut{A},q_I)$. Our proof uses the concept of topological hardness. A classical notion of topological hardness relies on the Borel hierarchy and the projective hierarchy~\cite{kechris_descriptive}, but  these notions are not suitable for us, since most regular tree languages live on the same level of these hierarchies: $\adelta 2$.  We use a more refined notion based on continuous reductions~\cite{wadge_phd} and so-called game languages~\cite{arnold_strict,bradfield_simplifying,niwinski_strict}.

\begin{definition} 
For $\rmin< \rmax$ consider the following alphabet 
\[A_{\rmin,\rmax} = \{\eve, \adam\}\times \{\rmin, \rmin+1,\ldots, \rmax\}.\]
With each  $t\in\partrees{A_{\rmin,\rmax}}$ we associate a parity game $\game_t$ where
\begin{itemize}
\item $V=\dom(t)$, $F=\holes(t)$,
\item $E =  \big \{(v,vd) \bigm| v\in \dom(t),  d\in\{\dL,\dR\} \big\}$,
\item if $t(v){=}(P,n)$ then $\Omega(v){=} n$ and $v\in V_P$ for $P\in\{\eve,\adam\}$.
\end{itemize}
Let $\W{i}{j}$ be the set of \emph{total trees} over $A_{i,j}$ such that $\eve$ has a winning strategy in $\game_t$.
\end{definition}


Let us assume the usual
Cantor-like topology on the space of trees, with the open sets
defined as arbitrary unions of \new{finite intersections of} sets of the form $\{t\in\trees{A} \bigm
| t(v) = a\}$ for $v\in\{\dL,\dR\}^*$ and $a\in A$.
Topological hardness of languages can be compared using continuous reductions. 
A \emph{continuous reduction} of $L_1\subseteq X$ to $L_2\subseteq Y$ is a continuous function $\fun{f}{X}{Y}$ such that $f^{-1}(L_2)=L_1$. The fact that $L_1$ can be continuously reduced to $L_2$ is denoted by $L_1\leq_W L_2$. On Borel sets,  the pre-order $\leq_W$ induces the so-called Wadge hierarchy (see~\cite{wadge_phd}) which greatly refines the Borel hierarchy and has the familiar ladder shape with pairs of mutually dual classes alternating with single self-dual classes. Here, we are interested in the following connection between continuous reductions, languages $\W{i}{j}$, and the alternating index hierarchy.

\begin{fact}[\cite{arnold_strict,bradfield_simplifying,niwinski_strict}]\mbox{} For all $\rmin < \rmax$,
\begin{enumerate}
\item $\W{\rmin}{\rmax}$ is regular and $\W{\rmin}{\rmax}\in\RM(\rmin, \rmax)$, 
\item $L\leq_W \W{\rmin}{\rmax}$ for each $L\in\RM(\rmin, \rmax)$,
\item $\W{\rmin}{\rmax} \not \leq_W \W{\rmin+1}{\rmax+1}$,
\item $\W{\rmin}{\rmax} \in \adelta 2$, $\W{0}{1}$ is $\asigma 1$-complete, $\W{1}{2}$ is $\api 1$-complete.
\end{enumerate}
\end{fact}

This gives  a criterion for proving index lower bounds. 

\begin{corollary} \label{cor:criterion}
If $\W{\rmin}{\rmax}  \leq_W L$ then $L\notin \RM(\rmin+1, \rmax+1)$.
\end{corollary}



In consequence, in order to show that the index bound computed by the algorithm from Section~\ref{ssec:algo} is tight, it suffices to show that if $\RM(\rmin, \rmax) \leq \class(\aut{A}, q_I)$, then $\W{\rmin}{\rmax} \leq_W \lang(\aut{A},q_I)$. We construct the reduction is three steps:
\begin{enumerate}
\item we show that if the class computed by the algorithm (i.e.~$\class(\aut{A}, q_I)$) is at least $\RM(i,j)$, then this is witnessed with a certain subgraph in $\graph(\aut{A})$, called $(i,j)$-edelweiss; 
\item we introduce intermediate languages $\widehat W_{i,j}$, whose internal structure corresponds precisely to $(i,j)$-edelweisses, and \new{in consequence} $\widehat W_{i,j} \leq_W \lang(\aut{A},q_I)$ if only  $\aut{A}$ contains an $(i,j)$-edelweiss reachable from $q_I$;
\item we prove that $W_{i,j} \leq_W \widehat W_{i,j}$.
\end{enumerate}
The combinatorial core of the argument is the last step.

\figILoop

\begin{definition} \label{def:ijloop}
We say that in a game automaton $\aut{B}$ there is an \emph{$i$-loop rooted in $p$} if there exists a word $w$ such that on the path $p\tran{w} p$ in $\graph(\aut{B})$ the minimal priority is $i$ \new{(see the left-hand side of Fig.~\ref{fig:iloop})}.

\new{A game} automaton $\aut{B}$ contains an \emph{$(i, j)$-loop for $\eve$ rooted in $p$} \new{(see the right-hand side of Fig.~\ref{fig:iloop})}, if there exist states $q, q_\dL,q_\dR$ of $\aut{B}$, a letter $a$, and words $w,w_\dL,w_\dR$ such that:
\begin{itemize}
\item $\delta(q,a) = (q_\dL, \dL) \lor (q_\dR,\dR)\,$; 
\item $p\tran{w} q$; $q_\dL\tran{w_\dL} p$; $q_\dR\tran{w_\dR} p\,$; 
\item on one of the paths $p\tran{w\, (a,\dL)\, w_\dL} p$ or $p\tran{w\, (a,\dR)\, w_\dR} p$ the minimal priority is $i$ and on the other it is $j$.
\end{itemize}
For \adam dually, with $\lor$ replaced with $\land$.

For an \emph{even} $j > i$, $\aut{B}$ contains an \emph{$(i,j)$-edelweiss rooted in $p$} (see \new{Fig.~\ref{fig:edelweiss-basic} and} Fig.~\ref{fig:edelweiss})
if for some even $n$ it contains 
\begin{itemize}
\item $(n+k)$-loops for $k = i, i+1, \dots,  j-3\,$,
\item  $(n+j-2,n+j-1)$-loop for $\eve\,$, if $i\leq j-2$; and
\item  $(n+j-1,n+j)$-loop for $\adam\,$; 
\end{itemize}
all rooted in $p$. For odd $j$ swap \adam and \eve but keep $n$ even.
\end{definition}

\figEdelBasic

\begin{lemma} \label{lm:invariant}
Let $\aut{A}$ be a game automaton and $q_I$ a state of $\aut{A}$. If $\class(\aut{A}, q_I) \geq\RM(i,j)$ then $\aut{A}$ contains an $(i,j)$-edelweiss rooted in a state reachable from  $q_I$.
\end{lemma}

\begin{proof} Let us first assume that $(i,j)=(0,1)$. Analyzing the algorithm we see that the only case when $\class(\aut{A}, q_I)$ jumps to $\RM(0,1)$ is when for some even $n$ there is an $n$-component $\aut{B}$ in $\aut{A}$, reachable from $q_I$, and containing states of priority $n$, such that some $n+1$ component $\aut{B}_\ell$ of $\aut{B}$ is \eve-branching in $\aut{B}$, i.e., $\aut{B}$ contains a transition of the form 
\[\delta(p,a) = (q_\dL,\dL) \lor (q_\dR,\dR)\]
with $p,q_\dL \in Q^{\aut{B}_\ell}$, $q_\dR\in Q^{\aut{B}}$ (or symmetrically, $p,q_\dR \in Q^{\aut{B}_\ell}$, $q_\dL\in Q^{\aut{B}}$). 
Since $\aut{A}$ is priority-reduced, $p$ is reachable from $q_\dL$ within $\aut{B}_\ell$ via a state of priority $n+1$, and from $q_\dR$ within $\aut{B}$ via a state of priority $n$. This gives an $(n, n+1)$-loop for \eve (a $(0,1)$-edelweiss) rooted in a state reachable from $q_I$. The argument for $(1,2)$ is entirely dual. 

Next, assume that $(i,j)=(0,2)$. It follows immediately from the algorithm that $\aut{A}$ contains an $n$-component $\aut{B}$ (reachable from $q_I$, containing states of priority $n$) such that $n$ is even and there exists an \eve-branching $(n+1)$-component $\aut{B}_\ell$ in $\aut{B}$ such that $\class(\aut{B}_\ell)=\RModd 1$ or $\class(\aut{B}_\ell)=\cmp_1$. In either case, $\class(\aut{B}_\ell)\geq \RM(1,2)$ and by the previous case $\aut{B}_\ell$ contains an $(n', n'+1)$-loop for \adam, for some odd $n'\geq n$. Since $\aut{A}$ is priority-reduced, for each state $q$ in $\aut{B}_\ell$  and each $r$ between $n$ and $\Omega(q)$, there is a loop from $q$ to $q$ with the lowest priority $r$. Hence, the $(n', n'+1)$-loop can be  turned into an $(n+1, n+2)$-loop. Thus, $\aut{B}_\ell$ contains an $(n+1, n+2)$-loop for \adam, rooted in a state $p$. We claim that $\aut{B}$ contains an $(n, n+1)$-loop for \eve, also rooted in $p$ (giving a $(0,2)$-edelweiss rooted in $p$). Indeed, since $\aut{B}_\ell$ is \eve-branching, arguing like for $(0,1)$, we obtain an $(n, n+1)$-loop for \eve rooted in a state $p'$ in $\aut{B}_\ell$. Since $\aut{B}_\ell$ is an $n+1$-component, there are paths in $\aut{B}_\ell$ from $p$ to $p'$ and back; the lowest priority on these paths is at least $n+1$. Using these paths one easily transforms the $(n, n+1)$-loop rooted in $p'$ into an $(n, n+1)$-loop rooted in $p$.


The inductive step is easy. Suppose that $j-i>2$. Then, for some even $n$, $\aut{A}$ contains an $(n+i)$-component $\aut{B}$ (reachable from $q_I$, containing states of priority $n+i$), which has an $(n+i+1)$-component $\aut{B}_\ell$ such that $\class(\aut{B}_\ell) = \RM(i+1,j)$ or $\class(\aut{B}_\ell) = \cmp(i+1,j)$. Since for each state $p$ in $\aut{B}_\ell$, $\aut{B}$ contains an $(n+i)$-loop rooted in $p$, we can conclude by the inductive hypothesis. 
\end{proof}






\figEdelNew

\begin{definition} 
For $\rmin\leq 2k-2$ consider the alphabet 
\[\widehat A_{\rmin, 2k} = \{\rmin, \rmin+1,\ldots, 2k-3, e, a \}.\]
With each  $t\in\partrees{\widehat A_{\rmin,2k}}$ we associate a parity game $\widehat \game_t$ with positions $\dom(t)$ and final positions $\holes(t)$ such that 
\begin{itemize}
\item if $t(v)=a$, then in $v$ player $\adam$ can choose to go to  $vL$ or to $v\dR$, and  $\Omega(v\dL)=2k-1$,  $\Omega(v\dR)=2k$,
\item if $t(v)=e$, then in $v$ player $\eve$ can choose to go to  $v\dL$ or to $v\dR$, and  $\Omega(v\dL)=2k-2$,  $\Omega(v\dR)=2k-1$,
\item if $t(v)\in\{i, i+1, \dots, 2k-3\}$, the only move from $v$ is to $v\dL$ and $\Omega(vL) = t(v)$.
\end{itemize}
For $\rmin = 2k-1$, let $\widehat A_{\rmin, 2k} = \{a,\top\}$, and let $\widehat \game_t$ be defined like above, except that if $t(v)=\top$ then $\Omega(v)=2k$ and the only move from $v$ is back to $v$.

Let $\widehat W_{i,2k}$ be the set of all total trees over $A_{i,2k}$ such that $\eve$ has a winning strategy in $\widehat\game_t$.

The languages $\widehat W_{i,2k+1}$ are defined dually, with $e,a$ and $\eve,\adam$ swapped, and $\top$ replaced with $\bot$.
\end{definition}

\begin{lemma}\label{lemma:edel_hard} 
If a total game automaton $\aut{A}$ contains an $(i,j)$-edelweiss rooted in a state reachable from an initial state $q_I$ then $\widehat W_{i,j} \leq_W \lang(\aut{A},q_I)$.
\end{lemma}

\begin{proof} 
We only give the proof for $(i,j)=(1,2)$; for other values of $(i,j)$ the argument is entirely analogous. By the definition, $\aut{A}$ contains an $(1,2)$-loop for \adam, rooted in a state $p$ reachable from $q_I$. Since $\aut{A}$ is a game automaton and has no trivial states, it follows that there exist
\begin{itemize}
\item a partial tree $t_I$ resolving $\aut{A}$ from $q_I$, with a single hole $v$,  labelled with $p$ in $\rho(\aut{A},t_I,q_I)$;
\item a partial tree $t_a$ resolving $\aut{A}$ from $p$ with two holes $v_1, v_2$, such that in $\rho(\aut{A},t_a,p)$ both holes are labelled $p$, the lowest priority on the path from the root to $v_i$ is $i$, and the closest common ancestor $u$ of $v_1$ and $v_2$ is labelled with a state $q$ such that $\delta_\aut{A}(q, t(u)) = (q_\dL,\dL)\land(q_\dR,\dR) $ for some $q_\dL, q_\dR$; and 
\item a total tree $t_\top \in \lang(\aut{A}, p)$.
\end{itemize}
Let us see how to build $t_a$. The paths $p\tran{w\, (a,\dL)\, w_\dL} p$, $p\tran{w\, (a,\dR)\, w_\dR} p$ guaranteed by Definition~\ref{def:ijloop} give as a partial tree $s$ with a single branching in some node $u$ and two leaves $v_1, v_2$, which we replace with holes. For $\rho=\rho(\aut{A},s,p)$, $\rho(v_1) = \rho(v_2) = p$ and $\delta(\rho(u),t(u)) = (q_\dL,\dL)\land(q_\dR,\dR)$.  At each hole of $s$, except $v_1$ and $v_2$, we substitute a total tree such that the run on the resulting tree with two holes resolves $\aut{A}$ from $p$, e.g., if $vL$ is a hole and $\delta(s(v), \rho(v)) = (q',\dL) \lor(q'', \dR)$, we substitute at $v\dL$ any tree that is not in $\lang(\aut{A}, q')$, relying on the assumption that $\aut{A}$ has no trivial states.

Let us define the reduction $g \colon \trees{\{a,\top\}} \to \trees{A^\aut{A}}$.
Let $t\in \trees{\{a, \top\}}$. For $v \in \dom(t)$, define $t_v$ co-inductively as follows: if $t(v) = \top$, set $t_v=t_\top$; if $t(v)=a$, then $t_v$ is obtained by plugging in the holes $v_1, v_2$ of $t_a$ the trees $t_{v\dL}$ and $t_{v\dR}$. Let $g(t)$ be obtained by plugging $t_{\varepsilon}$ in the hole of $t_I$. It is easy to check that $g$ continuously reduces $\widehat W_{1,2}$ to $\lang(\aut{A},q_I)$. 
\end{proof}

\medskip

It remains to see that $\W{i}{j} \leq_W \widehat W _{i,j}$. For the lowest level we give a separate proof.

\begin{lemma} \label{lm:12_hardness}  
$\W{0}{1} \leq_W \widehat W_{0,1}$ and $\W{1}{2} \leq_W \widehat W_{1,2}$.
\end{lemma}

\begin{proof} 
By the symmetry it is enough to prove the first claim. Let us take $t\in \trees{A_{0,1}}$. By K\"onig's lemma, Player $\exists$ has a winning strategy in $\game_t$ if and only if  she can produce a sequence of finite strategies $\sigma_0, \sigma_1, \sigma_2, \dots$ (viewed as subtrees of $t$) such that
\begin{enumerate}
\item $\sigma_0$ consists of the root only;
\item for each $n$ the strategy $\sigma_{n+1}$ extends $\sigma_n$ in such a way that below each leaf of $\sigma_n$ a non-empty subtree is added, and all \new{the} leaves of $\sigma_{n+1}$ have priority $0$.
\end{enumerate}

\new{
Clearly, the union of such a sequence of finite strategies $(\sigma_n)_{n\in\N}$ is a total strategy for \eve in $\game_t$. Additionally, the strategies $\sigma_n$ witness that their union visits a node of priority $0$ infinitely many times on each branch. Therefore, \eve wins $\game_t$.}

\new{Consider the opposite direction: we assume that \eve wins in $\game_t$ using a strategy $\sigma$ and we want to define the strategies $\sigma_n$. Let $\sigma_0$ consist of the root only and let $\sigma_{n+1}\subseteq \sigma$ extend the strategy $\sigma_n$ until the next node of priority $0$ is seen on every branch. We need to prove that all the strategies $\sigma_n$ are finite. Assume contrarily that $\sigma_n$ is finite but $\sigma_{n+1}$ is not. Let $v$ be a leaf of $\sigma_n$ such that $\sigma_{n+1}\restr_v$ is infinite. By K\"onig's lemma we know that there exists an infinite branch $\pi$ of $\sigma_{n+1}$ such that $v\prec \pi$. In that case there is no node of priority $0$ on $\pi$ after $v$. Therefore, $\pi$ treated as a play is winning for \adam and is consistent with $\sigma$. It contradicts the assumption that $\sigma$ was a winning strategy for \eve.}

Using \new{such approximating strategies $\sigma_n$} we can define the \new{required} reduction. Let $\left(\tau_i\right)_{i\in\N}$ be the list of all finite unlabelled binary trees. Some of these trees naturally induce a strategy for \eve in $\game_t$. For those we define $t_{\tau_i}\in \trees{\{e,\bot\}}$ co-inductively, as follows:
\begin{itemize}
\item $t_{\tau_i}(\dR^j) = e$ for all $j$;
\item if $\tau_j$ induces in $\game_t$ a strategy that is a legal extension of the strategy induced by $\tau_i$ in the sense of item 2) above, then the subtree of $t_{\tau_i}$ rooted at $\dR^j \dL$  is  $t_{\tau_j}$;
\item otherwise, all \new{the} nodes in this subtree are labelled with $\bot$. 
\end{itemize}
Let $f(t) = t_{\sigma_0}$. By the initial observation, $t_{\sigma_0} \in \widehat W_{0,1}$ if and only if \eve has a winning strategy in $\game_t$\new{: a winning strategy for \eve in $t_{\sigma_0}$ corresponds to the successive choices of strategies $\sigma_0\subseteq\sigma_1\subseteq\ldots\;$.}

Additionally, the function $f$ is continuous: to determine the labels in nodes $\dR^{n_1}\dL\dR^{n_2}\dL\dots \dR^{n_k}$ and $\dR^{n_1}\dL\dR^{n_2}\dL\dots \dR^{n_k}\dL$ we only need to know the restriction of $t$ to the union of the domains of $\tau_{n_1}, \tau_{n_2}, \dots, \tau_{n_k}$. Hence, $f$ continuously reduces $\W{0}{1}$ to $\widehat W_{0,1}$.
\end{proof}






\new {Our aim now is to prove the following proposition, which forms the technical core of this section. }

\begin{proposition}\label{pro:123_hardness}
For all $i$ and  $j \geq i+2$, $W_{i,j} \leq _W \widehat W_{i,j}$.
\end{proposition}

\new{The rest of this section is devoted to the proof of the proposition above. We begin by defining an auxiliary game $\tilde\game_t$ and proving that it is equivalent with $\game_t$. The structure of the game $\tilde\game_t$ corresponds to the possible choices of players in an edelweiss.}

By duality we can assume that $j=2k$.
For $t\in\trees{A_{i,2k}}$, let us consider \new{a} game $\tilde \game_t$ defined as follows. The positions are pairs $(v, \sigma)$,  where $v$ is a node of $t$, and $\sigma$ is \new{a} finite strategy from $v$ for \adam (viewed as a subtree of $t\restr_v$). Initially $v=\varepsilon$ is the root of $t$ and $\sigma=\{\varepsilon\}$. In each round, in a position $(v, \sigma)$, the players make the following moves:
\begin{itemize}
\item  \adam extends $\sigma$ under leaves of priority $2k-1$ to $\sigma'$ in such a way that on every path leading from a leaf of $\sigma$ to a leaf of $\sigma'$ all nodes have priority $2k$, except the leaf of $\sigma'$, which has priority at most $2k-1$;
\item \eve has the following possibilities:
\begin{itemize}
\item select a leaf $v'$ of $\sigma'$ with priority at most $2k-2$, and let the next round start with $(v', \{v'\})$, or
\item if $\sigma'$ has some leaves of priority $2k-1$, continue with $(v,\sigma')$.
\end{itemize}
\end{itemize}
A play is won by \eve if she selects a leaf infinitely many times and the least priority of these leaves seen infinitely often is even, or \adam is unable to extend $\sigma$ in some round. Otherwise, the play is won by \adam.

\begin{lemma}\label{lem:tilde-and-g}
A player $P$ has a winning strategy in $\game_t$ if and only if $P$ has a winning strategy in $\tilde\game_t$.
\end{lemma}

\begin{proof} 
For a winning strategy $\sigma_\eve$ for \eve in $\game_t$, let $\tilde\sigma_\eve$ be the strategy in $\tilde\game_t$ in which \eve selects a leaf $v'$ in $\sigma'$ if and only if $v'\in\sigma_\eve$. 
Consider an infinite play conforming to $\tilde\sigma_\eve$. If in the play  \eve selects a leaf infinitely many times, she implicitly defines a path in $t$ conforming to $\sigma_\eve$, and so the play must be winning for $\eve$. Assume that \eve selects a leaf only finitely many times. Then, \adam produces an infinite sequence of finite strategies $\{v\}=\sigma_0\subseteq\sigma_1\subseteq\ldots$ in $\game_t$. Let $\sigma_\infty$ be the union of these strategies. Consider the play $\pi$ in $\game_t$ passing through $v$ and conforming to $\sigma_\infty$ and $\sigma_\eve$. 
Observe that for each $\sigma_i$, the strategy $\sigma_\eve$ must choose some path; hence, either \eve selects a leaf of $\sigma_i$, or this path goes via a leaf of priority $2k-1$. 
Thus, $\pi$ is infinite and  by the rules of $\tilde\game_t$ priorities at most $2k-1$ are visited infinitely often. Since \eve selects a leaf only finitely many times, priorities strictly smaller than $2k-1$ are visited finitely many times in $\pi$. Hence, $\pi$ is won by \adam, what contradicts the assumption that $\sigma_\eve$ is winning for \eve. 

Now, let $\sigma_\adam$ be a winning strategy for \adam in $\game_t$. Then, for each $v\in \sigma_\adam$ there exists a finite sub-strategy $\sigma'$ of $\sigma_\adam$ from $v$ such that all internal nodes of $\sigma'$ have priority $2k$ and leaves have priority at most \ $2k-1$. This shows that for each current strategy $\sigma\subseteq\sigma_\adam$, $\adam$ is able to produce a legal extension $\sigma' \subseteq \sigma_\adam$. Let $\tilde \sigma_\adam$ be a strategy of $\adam$ in $\tilde\game_t$ that extends every given $\sigma$ by $\sigma'$ as above. Consider any play conforming to $\tilde\sigma_\adam$. By the initial observation, the play is infinite, so priorities strictly smaller then  $2k$ are visited infinitely often. If \eve selects a leaf only finitely many times, priorities strictly smaller then $2k-1$ occur only finitely many times and \adam wins. If \eve selects a leaf infinitely many times, then the lowest priority seen infinitely often must be odd, as otherwise \eve would show a losing path in $\sigma_\adam$. Hence, \adam wins in this case as well. 
 \end{proof} 

Now it remains to encode the game $\tilde\game_t$ as a tree $f(t) \in \trees{\widehat A_{i,2k}}$ in a continuous manner. The argument is similar to the one in Lemma~\ref{lm:12_hardness}. 
Let $\left(\tau_n\right)_{n\in\N}$ be the list of all unlabelled finite trees.
For some pairs $(v,\tau_n)$,  $\tau_n$ induces a strategy in $\game_t$ from the node $v$. For such $(v,\tau_n)$ we define  $t^\adam_{v,\tau_n}$ and $t^\eve_{v,\tau_n}$ co-inductively, as follows: 
\begin{itemize}
\item $t^\adam_{v,\tau_n}(\dR^m) = a$ for all $m$;
\item the subtree of $t^\adam_{v,\tau_n}$ rooted at $\dR^m\dL$ is $t^\eve_{v,\tau_m}$ if $\tau_m$ induces a strategy from $v$ that is a legal extension of $\tau_n$ according to the rules of $\tilde\game_t$, and otherwise the whole subtree is labelled with $e$'s (losing choice for \adam);
\end{itemize}
\begin{itemize}
\item $t^\eve_{v,\tau_n}(\dR^m) = e$ for $m=0, 1, \dots, \ell$, where $v_0, v_1, \dots, v_\ell$ are the leaves in the strategy induced by $\tau_n$ from $v$;
\item the subtree of $t^\eve_{v,\tau_n}$ rooted at $\dR^{\ell+1}$ is $t^\adam_{v,\tau_n}$ if the strategy induced by $\tau_n$ from $v$ has some leaves of priority $2k-1$, otherwise the whole subtree is labelled with $a$'s (losing choice for \eve);
\item for $m\leq \ell$, consider the following cases to define the subtree $s_m$ of $t^\eve_{v,\tau_n}$ rooted at $\dR^{m}\dL$:
\begin{itemize}
\item if $\Omega(v_m)\in\{2k-1,2k\}$ then $s_m$ is labelled everywhere with $a$'s (losing choice for \eve),
\item if $\Omega(v_m)=2k-2$ then $s_m=t^\adam_{v_m, \{v_m\}}$,
\item if $\Omega(v_m)=r<2k-2$ then $s_m(\varepsilon)=r$, the left subtree of $s_m$ is $t^\adam_{v_m, \{v_m\}}$, and the right subtree of $s_m$ is labelled with $a$'s (irrelevant for $\game_t$).
\end{itemize}


\end{itemize}
Let $f(t)$ be $t^\adam_{\varepsilon, \{\varepsilon\}}$. Checking that $f$ is continuous does not pose any difficulties. \new{Lemma~\ref{lem:tilde-and-g} implies that $f$ reduces $W_{i,j}$ to $\widehat W_{i,j}$, which concludes the proof of Proposition~\ref{pro:123_hardness}.}

\subsection{Corollaries}\label{ssec:alt-corols}

As a by-product of the decision procedure we have described, we obtain the following characterization of the levels of the alternating index hierarchy for game languages. 

\begin{proposition}\label{pro:index-depends-on-edels}
For a priority-reduced game automaton $\aut{A}$, $\lang(\aut{A},q_I)\in\RM(\rmin,\rmax)$ if and only if there is no $(\rmin+1,\rmax+1)$-edelweiss reachable from $q_I$ in $\aut{A}$.
\end{proposition}

\begin{proof}
One direction follows immediately from Lemma~\ref{lm:invariant}: if $\lang(\aut{A},q_I)\notin\RM(\rmin,\rmax)$ then $\class(\aut{A},q_I)\geq\RM(\rmin+1,\rmax+1)$ and $\aut{A}$ contains an $(\rmin+1,\rmax+1)$-edelweiss reachable from $q_I$.
For the opposite direction assume that $\aut{A}$ contains an $(\rmin+1,\rmax+1)$-edelweiss reachable from $q_I$. By Lemma~\ref{lemma:edel_hard} it implies that $\widehat W_{\rmin+1,\rmax+1} \leq_W \lang(\aut{A},q_I)$. Lemma~\ref{lm:12_hardness} together with Proposition~\ref{pro:123_hardness} imply that in that case $W_{\rmin+1,\rmax+1} \leq_W \lang(\aut{A},q_I)$. By Corollary~\ref{cor:criterion}, it means that $\lang(\aut{A},q_I)\notin\RM(\rmin,\rmax)$.
\end{proof}

A further corollary is the converse of Fact~\ref{fact:motivation} for the alternating index. 
 
\begin{proposition}\label{pro:alt-preserve}
For game automata, substitution preserves the alternating index.
\end{proposition}

\begin{proof} First note that without changing the outcome of the substitution $\aut{A}_{\aut{B}}$, we can always assume that the substituted state of $\aut{A}$ is an exit. We would like to use Proposition~\ref{pro:index-depends-on-edels}, but we first need to ensure that our automata are priority-reduced. The preprocessing that turns a given automaton into a priority reduced one, described in the proof of Lemma~\ref{lemma:reduced}, works independently in each strongly-connected component of the automaton. Hence, as long as the substituted state of $\aut{A}$ is an exit, the preprocessing commutes with substitution; that is, $(\aut{A}_{\aut{B}})' = \aut{A}'_{\aut{B}'}\,$, where primes are used to denote the preprocessed automata. Consequently, we can assume that our initial automata are priority reduced, and so are the results of the substitution.  The claim now follows from Proposition~\ref{pro:index-depends-on-edels}: since the characterization it offers is in terms of strongly connected subgraphs in the automaton, we can reason just like for non-deterministic index of deterministic automata in Proposition~\ref{pro:nondet-preserve}. If the languages recognized by automata $\aut{B}$ and $\aut{C}$ have the same index, then $\aut{B}$ and $\aut{C}$ contain the same edelweisses.  By the definition of substitution, no strongly connected subgraph in $\aut{A}_{\aut{B}}$ can use states from $\aut{A}$ and from $\aut{B}$. Consequently, $\aut{A}_{\aut{B}}$ and $\aut{A}_{\aut{C}}$ contain the same edelweisses reachable from $q_I^{\aut{A}}$, and so $\lang(\aut{A}_{\aut{B}}, q_I^{\aut{A}})$ and $\lang(\aut{A}_{\aut{C}}, q_I^{\aut{A}})$ have the same index. 
%
\end{proof}


%% file: 6_weak_index.tex
\section{Weak alternating index problem}
\label{sec:weak_index}

In this section we provide a procedure computing the weak index for languages given via a game automaton recognizing them.
Of course, a game language need not be weakly recognizable. This is because, as we have mentioned in Section \ref{sec:altindex}, languages recognized by weak alternating automata coincide with the class $\cmp_{0}\subseteq \RModd{1}\cap \RMeven{1}$, and, for each $i>0$, there is a language recognized by a game automaton belonging to $\RModd{i}\setminus \RMeven{i}$. 

As an immediate corollary of the proof of Theorem \ref{thm:alt_index}, we have the following decidable characterization of being weakly recognizable for languages recognized by game automata.
\begin{fact}
\label{ft:when-too-high}
Let $\aut{A}$ be a game automaton, and $q$ one of its states. The
language  $\lang(\aut{A},q)$ is weakly recognizable if and only if
$\aut{A}$ contains neither a $(0,1)$-edelweiss  nor a
$(1,2)$-edelweiss reachable from state $q$.   
\end{fact}
\begin{proof}
For the direction from left to right we reason as follows. Assume
$\lang(\aut{A},q)$ is weakly recognizable but  contains, say, a
$(0,1)$-edelweiss reachable from $q$. Then from Lemmas
\ref{lemma:edel_hard} and \ref{lm:12_hardness}, and Corollary
\ref{cor:criterion}, we have that $\lang(\aut{A},q)\notin \cmp_0$, a
contradiction. 
The other direction is an immediate consequence of Lemma
\ref{lm:invariant}. 
\end{proof}
This characterization of weak recognizability within the class of game
automata was essentially already provided in~\cite{niwinski_gap}. In
this paper it is shown that a deterministic automaton recognizes a
weakly recognizable language if and only if it does not contain a
forbidden pattern called \emph{split}, which corresponds to a
$(1,2)$-edelweiss. 
Since game automata are closed under dualization, they can also
contain a \emph{dual split}, that is a $(0,1)$-edelweiss.   
Fact \ref{ft:when-too-high}  is thence an immediate corollary from
the proof of the result of~\cite{niwinski_gap}.

Analogously to what we have done in Section \ref{sec:altindex}, in the aim of providing a precise formulation of the problem we want to solve, we start by introducing some useful notation.

\begin{definition}
For $\rmin< \rmax\in\N$, let $\wRM(\rmin, \rmax)$ denote the class of languages recognized by weak alternating tree automata of index $(\rmin, \rmax)$. Let
\begin{align*}
\wRMeven{\rmax}&=\wRM(0, \rmax),\\
\wRModd{\rmax}&=\wRM(1,\rmax+1),\\
\wRMdelta{\rmax} &= \wRM(0,\rmax) \cap \wRM(1,\rmax+1).
\end{align*}
These classes, naturally ordered by inclusion, constitute the \emph{weak index hierarchy}. The \emph{weak index of a language $L$} is the least class $\autclass{C}$ in the weak index hierarchy such that $L\in {\autclass{C}}$.
\end{definition}

Now we can properly formulate the main result of this section.

\begin{theorem}
\label{thm:weak-dec}
For a game automaton $\aut{A}$ and a state $q$, if $\aut{A}$ does not contain neither a $(0,1)$-edelweiss  nor a $(1,2)$-edewelweiss reachable from the state $q$, then $\lang(\aut{A}, q)$ is weakly recognizable and its weak index can be computed effectively.
\end{theorem}

The proof consists in a recursive procedure computing the weak class of $\lang(\aut{A}, q)$, denoted $\wclass(\aut{A}, q)$. The procedure itself is given in Subsection~\ref{sec:weak-proc}; Subsections~\ref{sec:weak-upper} and~\ref{sec:weak-lower} prove its correctness by  providing upper and lower bounds, respectively. The upper bounds are simply constructions of a weak alternating automaton of appropriate weak index recognizing the language $\lang(\aut{A}, q)$. The lower bounds show that for lower indices such constructions are impossible; they are obtained by means of simple tools from descriptive set theory. In the final Subsection \ref{sec:weak_vs_Borel}, we obtain as a Corollary that, as for deterministic languages, the weak index and the Borel rank coincide for tree languages recognised by game automata.

\subsection{The algorithm}
\label{sec:weak-proc}

\new{Like for the strong index we assume without loss of generality
that a given automaton $\aut{A}$ is priority-reduced (see
Lemma~\ref{lemma:reduced}).} The procedure works recursively on the
DAG of strongly-connected components, or SCCs, of $\aut{A}$ (maximal
sets of mutually reachable states).  We identify each SCC $\aut{B}$ of
an automaton $\aut{A}$ with the automaton obtained by restricting
$\aut{A}$ to the set of states in $\aut{B}$; the states outside of
$\aut{B}$ accessible via a transition originating in $\aut{B}$ become
the exits of the new automaton (cf. Subsection
\ref{subsec:automata}). \new{Note that the resulting automaton is also
priority-reduced.}
Our procedure computes $\wclass(\aut{A}, q)$ based on
$\wclass(\aut{A}, p)$ for exits $p$ of the SCC $\aut{B}$ containing
$q$. Those classes are aggregated in a way dependent on the internal
structure of $\aut{B}$, or more precisely, on the way in which the
state $p$ is reachable from $\aut{B}$. The aggregation \new{is done} by
means of auxiliary operations on classes. Two most characteristic are
\[\big(\wRMeven{n-1}\big)^\eve =
\big(\wRMdelta{n}\big)^\eve =
\big(\wRModd{n}\big)^\eve = 
\wRModd{n}, \quad \big(\wRMeven{n}\big)^\adam =
\big(\wRMdelta{n}\big)^\adam =
\big(\wRModd{n-1}\big)^\adam = 
\wRMeven{n}.\]
We also use the bar notation for the dual classes,
\[\overline{\wRMeven{n}} =\wRModd{n}, \quad  \overline{\wRModd{n}}
=\wRMeven{n},\quad  \overline{\wRMdelta{n}} =\wRMdelta{n},\] 
and $\Phi\lor\Psi$ for the least class containing $\Phi$ and $\Psi$.

Let us now describe the conditions that will trigger applying the
operations above to previously computed classes.  We begin with some
shorthand notation. \new{Recall that an $n$-path is a path in which the minimal
priority is $n$, and analogously for $n$-loop.}
Let $q'$, $q''$ be a pair of states in
$\aut{B}$. Let $\max_\Omega(q'\to q'')$ be the maximal $n$ such that
there exists an $n$-path from $q'$ to $q''$ in $\aut{B}$. Observe that
since $\aut{B}$ is an SCC, such $n$ is well-defined (at least
$0$). Also, since the automaton is priority-reduced, for each $n'\leq
\max_\Omega(q'\to q'')$ there exists an $n'$-\new{path} from $q'$ to $q''$
in $\aut{B}$.

\new{A} \emph{\adam-branching} transition in $\aut{B}$ is a transition
of the form $\delta(q', a) = (q_\dL,\dL)\land(q_\dR,\dR)$ with all
three states $q'$, $q_\dL$, $q_\dR$ in $\aut{B}$; dually for
$\eve$. \new{Note that these notions are similar but not entirely
analogous to  $\adam$-branching and $\eve$-branching components
from Section~\ref{ssec:algo}.} 

We say that a state $p$ is \emph{$(\eve,n)$-replicated by $\aut{B}$}
if there are states $q'$, $q''$ in $\aut{B}$ and a letter $a$ such
that $\delta(q',a)=(q'',\dL)\lor (p,\dR)$ (or symmetrically) and
$\max_\Omega(q''\to q') \geq n$.  Dually, $p$ is
\emph{$(\adam,n)$-replicated} if the transition above has the form
$\delta(q',a)=(q'',\dL)\land (p,\dR)$ (or the symmetrical).

We can now describe the procedure.  By duality we can assume
that the minimal priority in $\aut{B}$ is $0$. If $\aut{A}$ contains
no loop reachable from $q$, set $\wclass(\aut{A}, q) =\wRMdelta{1}$.
If it contains an accepting loop reachable from $q$, but no rejecting
loop reachable from $q$, set $\wclass(\aut{A}, q) =
\wRMeven{1}$. Symmetrically, if it contains a rejecting loop reachable
from $q$, but no accepting loop reachable from $q$, set
$\wclass(\aut{A}, q) = \wRModd{1}$. Otherwise, consider \new{the
  following} two cases.    

Assume first that $\aut{B}$ contains no \adam-branching transition. In
that case, for every transition $\delta(q,a)$ of $\aut{B}$ that is
controlled by \adam, at most one of the successors of $\delta(q,a)$ is
a state of $\aut{B}$. Hence, $\aut{B}$ can be seen as a
co-deterministic tree automaton (exits are removed from the
transitions; if both states in a transition are exits, the transition
is set to $\bot$). Thus, the automaton $\aut{\bar {B}}$ dual to
$\aut{B}$ is a deterministic tree automaton. For deterministic tree
automata it is known how to compute the weak
index~\cite{murlak_weak_index}. \new{Denote the weak index of $\aut{\bar{B}}$ as $\wclass(\aut{\bar B},q)$.}

\new{Now, set $\wclass(\aut{A},q)$ to  }
\begin{equation}
\label{eq:weak-no-A}
\wRMdelta{2} \lor \overline{\wclass(\aut{\bar B},q)} \,\lor 
\bigvee_{p\in F} \!\wclass(\aut{A},p) \,\lor 
\bigvee_{p\in F_{\eve,1}} \wclass(\aut{A},p) ^\eve \,\lor 
\bigvee_{p\in F_{\adam,0}} \wclass(\aut{A},p) ^\adam 
\end{equation}
where $F\subseteq Q^{\aut{A}}$ is the set of exits of $\aut{B}$, $F_{\eve,1}\subseteq F$ is the set of states $(\eve,1)$-replicated by $\aut{B}$, and similarly for $F_{\adam,0}$. 

Assume now that $\aut{B}$ does contain an \adam-branching transition. By the hypothesis of the theorem, for every \adam-branching transition $\delta(q', a) = (q_\dL,\dL)\land(q_\dR,\dR)$ in $\aut{B}$, it must hold that $\max_\Omega(q_\dL \to q')\leq 1$ and $\max_\Omega(q_\dR \to q')=0$, or symmetrically. We call a state $q''$ (either $q_\dL$ or $q_\dR$) in an \adam-branching transition \emph{bad} if $\max_\Omega(q''\to q')=0$. Let $\aut{B}^{-}$ be the automaton $\aut{B}$ with all these \emph{bad} states in the \adam-branching transitions changed into $\top$, and let $\aut{A}^{-}$ be the automaton $\aut{A}$ with $\aut{B}$ replaced by $\aut{B}^{-}$. Observe that $\aut{B}^{-}$ contains no \adam-branching transitions. Let us put
\begin{equation}
\label{eq:weak-A}
\wclass(\aut{A},q)= \wRMdelta{2} \lor \big(\wclass(\aut{A}^{-},q)\big)^\adam.
\end{equation}

\subsection{Upper bounds}
\label{sec:weak-upper}

In this section we prove that $\wclass(\aut{A}, q)$ is an upper bound for the weak index of $\lang(\aut{A}, q)$. More precisely, we show the following.  

\begin{lemma}
\label{lem:weak-upper}
If $\wclass(\aut{A}, q)\leq \wRM(\rmin,\rmax)$ then the language
$\lang(\aut{A}, q)$ can be recognised by a weak alternating automaton
of index $(\rmin,\rmax)$. 
\end{lemma}

Let us first deal with the lowest levels. The algorithm never returns $\wRMeven{0} = \wRM(0,0)$ nor $\wRModd{0}=\wRM(1,1)$, so the lowest $(\rmin, \rmax)$ we need to consider are $(0,1)$ and $(1,2)$. Suppose that $(\rmin, \rmax)=(0,1)$. Examining the algorithm we immediately see that this is possible only if automaton $\aut{A}$ does not contain a rejecting loop reachable from state $q$. Since our automaton is priority reduced, it means that it uses only priority $0$. Hence, it is already a $(0,1)$ weak automaton (not $(0,0)$, because of allowed $\bot$ transitions). For $(\rmin, \rmax)=(1,2)$ the argument is entirely analogous. 

For higher indices we consider three cases, leading to three different constructions of weak alternating automata recognizing $\lang(\aut{A},q)$.

\subsubsection{\texorpdfstring{$\aut{B}$ has no $\adam$-branching transitions and $(\rmin,\rmax)=(1,\rmax)$ with $\rmax\geq 3$}{B has no A-branching transitions and (i,j)=(1,j)}}

In an initial part of the weak automaton recognizing $\lang(\aut{A},q)$ the players declare whether during the play on a given tree they would leave the $\aut{B}$ component or not. Since $\aut{B}$ has no $\adam$-branching transitions, as long as the play has not left $\aut{B}$, each choice of $\adam$ amounts to leaving $\aut{B}$ or staying in $\aut{B}$. Hence, each strategy of $\eve$ admits exactly one path staying in $\aut{B}$, finite or infinite. We first let $\eve$ declare $l_\eve\in\{\leave,\stay\}$\new{, where $\leave$ means that the path is finite, and $\stay$ means that it is infinite.}

\begin{itemize}
\item If $l_\eve=\leave$, we move to a copy of $\aut{B}$ with all the priorities set to $1$. By Equation~\eqref{eq:weak-no-A}, for every exit $f$ of $\aut{B}$ we have $\wclass(\aut{A},f)\leq \wRM(1,\rmax)$. Therefore, we can compose this copy of $\aut{B}$ with all the automata for $\lang(\aut{A}, f)$ to obtain an automaton of index $(1,\rmax)$.

\item Assume that $l_\eve=\stay$. Given the special shape of $\eve$'s strategies, this means that $\eve$ claims that the play will only leave $\aut{B}$ if at some point $\adam$ chooses an exit $f$ in a transition whose other end is in $\aut{B}$. Since the minimal priority in $\aut{B}$ is $0$, all these exists are $(\adam,0)$-replicated. We check $\eve$'s claim by substituting all other exits in transitions with rejecting states, i.e.~weak alternating automata of index $(3,3)$ (recall that $\rmax$ is at least $3$). Thus, the only exits that are not substituted are the $(\adam,0)$-replicated ones.  Now, we ask \adam whether he plans to take one of these exists: he declares $l_\adam\in\{\leave,\stay\}$, accordingly.

\begin{itemize}
\item If $l_\adam=\stay$, the play moves to the weak alternating automaton of index $\overline{\wclass_\dt(\aut{\bar B})}$, corresponding to the co-deterministic automaton $\aut{B}$ with the remaining exits removed from transitions (they were only present in transitions of the form $(q_\dL,\dL) \land (q_\dR,\dR)$, with the other state in $\aut{B}$). 

\item Assume that $l_\adam=\leave$. In that case we move to a copy of
  $\aut{B}$ with all the priorities set to $2$. The only exits left
  are the $(\adam,0)$-replicated ones. By
  Equation~\eqref{eq:weak-no-A}, for all such exists $f$,
  \[\wclass(\aut{A}, f)\leq \wRM(0,\rmax-2)\,,\] for otherwise
  $\wclass(\aut{A}, f)\geq \wRM(1,\rmax-1)$, so  $\big(\wclass(\aut{A}, p)\big)^\adam\geq \wRM(0,\rmax-1)$ and $\wRM(0,\rmax-1)$ is not smaller than $\wRM(1,\rmax)$. In particular, we can find a weak alternating automaton of index $(2,\rmax)$ recognizing $\lang(\aut{A}, f)$. So the whole sub-automaton is a weak alternating automaton of index $(2,\rmax)$.
\end{itemize}
\end{itemize}

\subsubsection{\texorpdfstring{$\aut{B}$ has no $\adam$-branching transitions and $(\rmin,\rmax)=(0,j)$ with $\rmax\geq 2$}{B has no A-branching transitions and (i,j)=(0,j)}}

The simulation starts in a copy of $\aut{B}$ with all the priorities set to $0$. If the play leaves $\aut{B}$ at this stage then we move to the appropriate automaton of index $(0,\rmax)$. At any moment $\adam$ can pledge that:
\begin{itemize}
\item the play will no longer visit transitions $\delta(q',a)$ of the form $(f_\dL,\dL)\land(f_\dR,\dR)$, $(f_\dL,\dL)\lor(f_\dR,\dR)$, $(q_\dL,\dL)\lor(f_\dR,\dR)$, $(f_\dL,\dL)\lor(q_\dR,\dR)$, or  $(q_\dL,\dL)\lor(q_\dR,\dR)$, where $\max_{\Omega}(q_\dL\to q')=\max_{\Omega}(q_\dR\to q')=0$ and $f_\dL$, $f_\dR$ are exits of $\aut{B}$;

\item in the transitions he controls, he will always choose the state in $\aut{B}$, and win regardless of $\eve$'s choices.
\end{itemize}
If the play stays forever in $\aut{B}$ but \adam is never able to make such a pledge, he loses by the parity condition---it means that infinitely many times a loop from $q_\dL\to q'$ or $q_\dR\to q'$ is taken with $\max_\Omega(q_d\to q')=0$ therefore, the minimal priority occurring infinitely often is $0$.

After \adam has made the above pledge, \eve has the following choices:
\begin{itemize}
\item She can challenge the first part of $\adam$'s pledge, declaring that at least one such transition is reachable. In that case we move to a copy of $\aut{B}$ with all the priorities set to $1$ and all the transitions controlled by \eve. In this copy, reaching any of the disallowed transitions entails acceptance---the play immediately moves to a $(2,2)$ final component.
\item She can accept the first part of $\adam$'s pledge. 
\end{itemize}

After \eve has accepted the first part of \adam's pledge, we can assume that the rest of the game in $\aut{B}$ is a single infinite branch. Indeed, by the hypothesis of the theorem, for every $\eve$-branching transition $\delta(q',a)=(q_\dL,\dL)\lor(q_\dR,\dR)$ in $\aut{B}$ it must hold that $\max_{\Omega}(q_\dL\to q')=\max_{\Omega}(q_\dR\to q')=0$; otherwise, $\aut{B}$ would contain $(0,1)$-edelweiss.
Thus, no $\eve$-branching transition can be reached, and since $\aut{B}$ contains no $\adam$-branching transitions at all, the game can continue in $\aut{B}$ in only one way.

Now $\eve$ must challenge the second part of $\adam$'s pledge. We ask her whether she plans to leave $\aut{B}$ or not, and she declares $l_\eve\in\{\leave,\stay\}$.
\begin{itemize}
\item If $l_\eve=\stay$ then we proceed to the weak automaton of index
  $\wclass(\aut{B},q)$, corresponding to $\aut{B}$ treated as a co-deterministic automaton. We are only interested in the behaviour of this automaton over trees in which there is exactly one branch in $\aut{B}$, and it is infinite. Over such trees we want to make sure that \new{neither player} ever chooses to exit. This is already ensured: when $\aut{B}$ is turned into a co-deterministic tree automaton, the exits are simply removed from transitions (if both states are exits, the transition is changed to a transition to a $(2,2)$ automaton, but such transitions will never be used over trees we are interested in).

\item If $l_\eve=\leave$ then we move to a copy of $\aut{B}$ with all the priorities set to $1$. The only available exits of $\aut{B}$ in this copy are those in transitions of the form $\delta(q',q) = (q_\dL,\dL)\lor(f,\dR)$ (or symmetrical) with $\max_{\Omega}(q_\dL\to q')>0$ (in other transitions the exits are removed, if both states are exits, they are replaced by a final $(2,2)$-component); therefore $\wclass(\aut{A}, f)\leq \wRM(1,\rmax)$ and we can simulate it with a $(1,\rmax)$-automaton.  
\end{itemize}

\subsubsection{\texorpdfstring{$\aut{B}$ contains $\adam$-branching}
{B contains A-branching} transitions}

If $\aut{B}$ contains an $\adam$-branching transition, the algorithm returns $\wclass(\aut{A}, q)$ of the form $\wRM(0,\rmax)$. Let us construct a weak automaton of index $(0,\rmax)$ that recognizes $\lang(\aut{A},q)$. The automaton starts in a copy of $\aut{B}$ with all the priorities set to $0$. At any moment \adam can declare that no-one will ever take any \emph{bad} transition in $\aut{B}$.  
If he cannot make such a declaration, it means that \eve can force infinitely many bad transitions to be taken, and she wins. After $\adam$ has made such declaration, we need to recognize the language $\lang(\aut{A}^{-},q)$ (note that the bad transitions in $\aut{A}^{-}$ are made directly losing for $\adam$). For this we can use a weak automaton of index $\wclass(\aut{A}^{-})\leq\wRM(0,\rmax)$,
already constructed.

\subsection{Lower bounds}
\label{sec:weak-lower}

Now we prove lower bounds for the weak index computed by our procedure, as expressed by Lemma~\ref{lem:weak-lower}.

\begin{lemma}
\label{lem:weak-lower}
If $\wclass(\aut{A}, q)\geq \wRM(\rmin, \rmax)$ then $\lang(\aut{A}, q)$ cannot be recognised by a (total) weak alternating automaton of index $(\rmin+1,\rmax+1)$.
\end{lemma}

For this we use a topological argument, relying on the following simple observation~\cite{duparc_weak}, essentially proved already by Mostowski~\cite{mostowski_hierarchies}. Let 
$\bpi{n}$, $\bsigma{n}$, and $\bdelta{n}$ be the finite Borel classes; that is,
$\bsigma{1}$ is the class of the open sets, $\bpi{n}$ consists of the
complements of sets from $\bsigma{n}$, $\bdelta{n} =
\bsigma{n}\cap\bpi{n}$, and $\bsigma{n+1}$ consists of countable
unions of sets from $\bpi{n}$. 

\begin{fact}\label{fact:weak_borel}
If $L$ is recognizable by a weak alternating automaton of index $(0,\rmax)$ then $L\in\bpi{\rmax}$. Dually, for index $(1,\rmax+1)$, we have $L\in\bsigma{\rmax}$.
\end{fact}

Thus, in order to show that a language is \emph{not} recognizable by weak alternating automaton of index $(0,\rmax)$ it is enough to show that it is not in $\bpi{\rmax}$. This can be shown by providing a continuous reduction to $L$ from some language not in $\bpi{\rmax}$, e.g.~a $\bsigma{\rmax}$-complete language. We shall use languages introduced by Skurczy{\'n}ski~\cite{skurczynski_borel_infinite}. 

One can define Skurczy{\'n}ski's languages by means of two dual operations on tree languages.

\begin{definition} For $L \subseteq \trees{A}$ define 
\[ L^\adam = \left \{t \in \trees{A}\bigm | \forall_{n\in\N}\; t\restr_{\dL^n\dR}\in L\right\}, \quad 
L^\eve = \left \{t  \in \trees{A}\bigm | \exists_{n\in\N}\; t\restr_{\dL^n\dR}\in L\right\}.\]
\end{definition}

It is straightforward to check that these operations are monotone with
respect to the Wadge ordering; that is,
\[ L\wadgeq M \text{ implies } L^\adam \wadgeq M^\adam \text{ and } L^\eve \wadgeq M^\eve.\]

Moreover, for all $n>0$, 
\begin{itemize}
\item if $L$ is $\bsigma{n}$-complete, $L^\adam$ is $\bpi{n+1}$-complete, and 
\item if $L$ is $\bpi{n}$-complete, $L^\eve$ is $\bsigma{n+1}$-complete. 
\end{itemize}

This allows us to define simple tree languages complete for finite levels of the Borel hierarchy.

\begin{definition}[\cite{skurczynski_borel_infinite}]
Consider the alphabet $A=\{\bot,\top\}$. Let \[S_{(0,1)} = \left \{t \in \trees{A} \bigm | t(\epsilon)=\top\right\}^\adam\,,\quad S_{(1,2)} = \left \{t \in \trees{A} \bigm | t(\epsilon)=\bot \right\}^\eve\,.\]  The remaining languages are defined inductively,
\[ S_{(0,\rmax+1)}=(S_{(1,\rmax+1)})^\adam\,, \quad
S_{(1,\rmax+1)}=(S_{(0,\rmax-1)})^\eve\,.\]
For notational convenience, let $S_{(0,0)}=\trees{A}$ and $S_{(1,1)}=\emptyset$.
\end{definition}

Note that the languages are dual to each other:
$S_{(1,\rmax+1)}=\trees{A}\setminus S_{(0,\rmax)}$. A straightforward
reduction shows that $S_{(i',j')}\wadgeq S_{(i,j)}$ whenever $(i,j)$
is at least $(i',j')$. But the crucial property is the following. 

\begin{fact}[\cite{skurczynski_borel_infinite}]
\label{fact:sku}
 $S_{(0,n)} \in \bpi{n} \setminus \bsigma{n}$ and $S_{(1,n+1)} \in \bsigma{n}\setminus \bpi{n}$.
\end{fact}

Summing up, from Facts~\ref{fact:weak_borel} and~\ref{fact:sku} it follows immediately that if $S_{(\rmin,\rmax)}\leq_{W} L$ then $L$ is not recognizable by a weak alternating automaton of index $(\rmin+1,\rmax+1)$.

Observe that $S_{(\rmin,\rmax)}$ can be recognised by a weak game automaton of index $(\rmin,\rmax)$.  One consequence of this---and Facts~\ref{fact:weak_borel} and~\ref{fact:sku}---is the strictness of the hierarchy.

\begin{corollary}
The weak index hierarchy is strict, even when restricted to languages recognizable by game automata. 
\end{corollary}

Another consequence is that it is relatively easy to give the reductions we need to prove Lemma~\ref{lem:weak-lower}, summarized in the claim below.

\begin{claim}\label{lemma:hardness}
If $\wclass(\aut{A}, q)\geq \wRM(\rmin,\rmax)$ then $S_{(\rmin,\rmax)} \wadgeq \lang(\aut{A}, q)$.
\end{claim}

We prove this claim by induction on the structure of the DAG of SCCs
of $\aut{A}$ reachable from $q$, following the cases of the algorithm
just like for the upper bound. One of the cases is covered by the
procedure for deterministic automata, which we use as a black box. But
in order to prove Lemma~\ref{lem:weak-lower} we need to know that it
preserves our invariant. And indeed, just like here, it is a step in
the correctness proof: if the procedure returns at least
$\wRM(\rmin,\rmax)$, then $S_{(\rmin,\rmax)}$ continuously reduces to the
recognised language~\cite{murlak_weak_index}. 

The remaining cases essentially correspond to the items in the following lemma. 

\begin{lemma}\label{lem:replication}
Assume that $q$ is a state of $\aut{A}$, $\aut{B}$ is the SCC of $\aut{A}$ containing $q$, and $p$ is a state of $\aut{A}$ reachable from $q$ (from the same or different SCC). 
\begin{enumerate}
\item \label{item:reach}
$\lang(\aut{A},p)\wadgeq \lang(\aut{A},q)$. 
\item \label{item:minus}
$\lang(\aut{A}^-,q)\wadgeq \lang(\aut{A},q)$.
\item \label{item:aloop}
If an accepting loop is reachable from $q$, then
$S_{(0,1)} \wadgeq \lang(\aut{A},q)$.
\item \label{item:rloop}
If a rejecting loop is reachable from $q$, then
  $S_{(1,2)} \wadgeq \lang(\aut{A},q)$.
\item \label{item:areplicated} 
If $p$ is $(\adam,0)$-replicated by $\aut{B}$ then $\left(\lang(\aut{A}, p)\right)^\adam \wadgeq \lang(\aut{A}, q)$.
\item \label{item:ereplicated}
If $p$ is $(\eve,1)$-replicated by $\aut{B}$ then $\left(\lang(\aut{A}, p)\right)^\eve \wadgeq \lang(\aut{A}, q)$.
\end{enumerate}
\end{lemma}

\begin{proof}
The proof is based on Fact~\ref{ft:resolve}. Let us begin with~\eqref{item:reach}. Since all the states of $\aut{A}$ are non-trivial, we can construct a tree $t$ with a hole $h$ such that $t$ resolves $\aut{A}$ from $q$ and the state $\rho(\aut{A}, t, q)(h)$ is $p$. In that case $t[h:=s]\in\lang(\aut{A},q)$ if and only if $s\in\lang(\aut{A}, p)$. Therefore, the function $s\mapsto t[h:=s]$ is a continuous reduction witnessing that $\lang(\aut{A},p)\wadgeq \lang(\aut{A},q)$.

For~\eqref{item:minus}, recall that $\aut{A}^-$  is obtained from
$\aut{A}$ by turning some choices for $\adam$ to $\top$; that is, some
transitions $\delta(r,a)$ of the form $(r_\dL, \dL) \land (r_\dR, \dR)$ are
set to $(r_\dL, \dL)$, $(r_\dR, \dR)$, or $\top$. This means that if a node
$v$ of tree $t$ has label $a$ and gets state $p$ in the associated run
$\rho(\aut{A}^-, t, q)$, then $t\restr_{v\dL}$, $t\restr_{v\dR}$, or both of them,
respectively, are immediately accepted by $\aut{A}^-$. In the
corresponding run of the original automaton $\aut{A}$, however, these
subtrees will be inspected by the players and we should make sure they
are accepted. The way to do it is simple: since $r_\dL$ and $r_\dR$ are
non-trivial in $\aut{A}$, we can replace these subtrees with $t_{r_\dL}
\in \lang(\aut{A},r_\dL)$, or $t_{r_\dR} \in \lang(\aut{A},r_\dR)$,
accordingly. This gives a continuous reduction of $\lang(\aut{A}^-,q)$
to $\lang(\aut{A},q)$. 

To prove~\eqref{item:aloop}, let us fix  a state $p$ on an accepting loop $C$, reachable from $q$. By~\eqref{item:reach} and transitivity of $\wadgeq$, it is enough to show that $S_{(0,1)}\wadgeq \lang(\aut{A},p)$. Let $t$ be a tree with hole $h$ such that $t$ resolves $\aut{A}$ from $p$, the state $\rho(\aut{A},t,p)$ is $p$, and the states on the shortest path from the root to $h$ correspond to the accepting loop $C$. Since all states in $\aut{A}$ are non-trivial, we can also find a full tree $t' \notin \lang(\aut{A}, p)$. Let $t_0=t'$ and $t_n = t[h:=t_{n-1}]$ for $n>0$, and let $t_{\infty}$ be the tree defined co-inductively as \[t_{\infty} = t[h:=t_{\infty}]\,.\] Then, $t_n\notin \lang(\aut{A},p)$ for all $n\geq 0$, but $t_{\infty} \in \lang(\aut{A},p)$. To get a continuous function reducing $S_{(0,1)}$ to $\lang(\aut{A},p)$, map tree $s\in \trees{\{\bot,\top\}}$ to $t_m$, where $m = \min \left \{ i \bigm| s(\dL^i\dR)=\bot\right\}$, or to $t_{\infty}$ if $\left \{ i \bigm| s(\dL^i\dR)=\bot\right\}$ is empty. 

Item~\eqref{item:rloop} is analogous.

For~\eqref{item:areplicated}, let us assume that
$\delta(q,a)=(q_\dL,\dL)\land (p,\dR)$ is the transition witnessing that $p$
is $(\adam,0)$-replicated by $\aut{A}$. Let us also fix the path
$q_\dL\to q$ with minimal priority $0$. Now, let $t$ be a tree with a
hole $h$ that resolves $\aut{A}$ from $q$ and the value of the run of
$\aut{A}$ in $h$ is $q$. Similarly, let $t'$ be the tree with a hole
$h'$ that resolves $\aut{A}$ from $q_L$ and the value of the
respective run is $q$. Let us construct a continuous function that
reduces $\left(\lang(\aut{A}, p)\right)^\adam$ to $\lang(\aut{A},
q)$. Assume that a given tree $s$ has subtrees $s_i$ under the nodes
$L^i R$. Let us define co-inductively $t_i$ as 
\[t_i=a(t'[h':=t_{i+1}], s_i),\]
i.e.~the tree with the root labelled by $a$ and two subtrees:
$t'[h':=t_{i+1}]$ and $s_i$. Finally, let $f(s)$ be $t[h:=t_0]$. Note
that the run $\rho(\aut{A}, f(s), q)$ labels the hole $h$ of $t$ by
$q'$. Therefore, $f(t)\in\lang(\aut{A},q)$ if and only if
$t_0\in\lang(\aut{A}, q')$ and $t_i\in\lang(\aut{A}, q)$ if and only
if $t_{i+1}\in\lang(\aut{A},q)$ and $s_i\in\lang(\aut{A},p)$. Since
the minimal priority on the path from $t_i$ to $t_{i+1}$ is $0$, if no
$s_i$ belongs to $\lang(\aut{A}, p)$ then $f(t)\notin\lang(\aut{A},
q)$. Therefore, $f$ is in fact the desired reduction. 

The proof of~\eqref{item:ereplicated} is entirely analogous.
\end{proof}

Using Lemma~\ref{lem:replication}, and the guarantees for 
deterministic automata discussed earlier, we prove
Lemma~\ref{lem:weak-lower} as follows. 

\begin{proof}[of Lemma~\ref{lem:weak-lower}] 
By induction on the recursion depth of the algorithm execution we
prove that if $\wclass(\aut{A},p)\geq \wRM(\rmin,\rmax)$ then
$S_{(\rmin,\rmax)}\wadgeq \lang(\aut{A}, p)$.

Let us start with the lowest level. Assume that $(\rmin,\rmax)=(0,1)$
(for $(1,2)$ the proof is analogous). Examining the algorithm we see
that this is only possible if there is an accepting loop in $\aut{A}$,
reachable from $q$. Then, by Lemma~\ref{lem:replication} Item~\eqref{item:aloop}, $S_{(0,1)} \wadgeq \lang(\aut{A},q)$.

For higher levels we proceed by case analysis. First we cover the
possible reasons why equation~\eqref{eq:weak-no-A} can give at
least $\wRM(\rmin,\rmax)$.
If $\overline{\wclass(\aut{\bar B},q)}\geq \wRM(\rmin,\rmax)$, the
invariant follows immediately from the guarantees for deterministic
automata, and the duality between indices and between Skurczy{\'n}ski languages.
If $\wclass(\aut{A},p)\geq \wRM(\rmin,\rmax)$ for some $p\in F$,  we use the fact that $\lang(\aut{A},p)\wadgeq \lang(\aut{A},q)$, and get $S_{(\rmin,\rmax)}\wadgeq\lang(\aut{A},q)$ by transitivity.
Then, assume that $\wclass(\aut{A},p)^\eve\geq \wRM(\rmin,\rmax)$ for some $p\in
F_{\eve,1}$ (for $p\in F_{\adam,0}$ the proof is analogous). That
means that $\wclass(\aut{A},p)\geq \wRM(0,\rmax')$ such that
$\left(\wRM(0,\rmax')\right)^\eve = \wRM(1,\rmax'+2) \geq \wRM(\rmin,\rmax)$. 
By the inductive hypothesis
$S_{(0,\rmax')}\wadgeq \lang(\aut{A}, p)$, so by the monotonicity of
$\eve$ and Lemma~\ref{lem:replication} Item~\eqref{item:ereplicated}, 
$S_{(1,\rmax'+2)} = \left(S_{(0,\rmax')}\right)^\eve\wadgeq\left(\lang(\aut{A},p)\right)^\eve\wadgeq
\lang(\aut{A},q)$. But since $\wRM(1,\rmax'+2) \geq
\wRM(\rmin,\rmax)$, by the Wadge ordering of Skurczy{\'n}ski's
languages $S_{(\rmin,\rmax)} \wadgeq S_{(1,\rmax'+2)}$, and \new{consequently}
$S_{(\rmin,\rmax)} \wadgeq \lang(\aut{A},q)$ follows by transitivity.

Finally, assume that $\wclass(\aut{A},q)$ is computed according to~\eqref{eq:weak-A};
that is, the component $\aut{B}$ contains an $\adam$-branching
transition $\delta(q',a) = (q_\dL,\dL)\land(q_\dR,\dR)$. As we have already
observed, the hypothesis of the theorem implies that in that case
$\max_\Omega(q_\dL\to q') = 0$ and $\max_\Omega(q_\dR\to q') \leq 1$ (or
symmetrically). That means that $q_\dR$ is $\adam,0$-replicated by
$\aut{B}$, so by Lemma~\ref{lem:replication} Item~\eqref{item:areplicated},
$\left(\lang(\aut{A}, q_\dR)\right)^\adam \wadgeq \lang(\aut{A},
q)$. But since $\aut{B}$ is strongly connected, $q$ is reachable from
$q_\dL$ and $q_\dL$, so by Lemma~\ref{lem:replication} Item~\ref{item:reach} we have
$\lang(\aut{A}, q)\wadgeq \lang(\aut{A}, q_\dR)$. Since $\forall$ is
monotone, we conclude that
\begin{equation}
\left(\lang(\aut{A},
q)\right)^\adam \wadgeq \lang(\aut{A}, q)\,. \label{eqn:autoreplication}
\end{equation}
(\new{Although it  looks paradoxical, it is not the case since}
$(L^\forall)^\forall \wadgeq L^\forall$ for all $L$.)
Since $\wclass(\aut{A}, q) \geq \wRM(\rmin,\rmax)$, it must hold that
$\wclass(\aut{A}^-, q) \geq \wRM(1,\rmax')$, such that
$\left(\wRM(1,\rmax')\right)^\adam = \wRM(0,\rmax')\geq \wRM(\rmin,\rmax)$.
By the induction hypothesis, $S_{(0,\rmax')}\wadgeq
\lang(\aut{A}^-, q)$. Consequently, by Lemma~\ref{lem:replication} Item~\eqref{item:minus}
and by transitivity, $S_{(0,\rmax')}\wadgeq \lang(\aut{A},
q)$. Thus, by the monotonicity of $\adam$, from~\eqref{eqn:autoreplication} we get 
$ S_{(0,\rmax')} = \left( S_{(1,\rmax')}\right)^\adam\wadgeq \lang(\aut{A},
q)$, and we conclude by the Wadge ordering of Skurczy{\'n}ski's languages. 
\end{proof}

\subsection{Substitution preserves the weak alternating index}

A quick look at how the weak index is computed suffices to show the converse of Fact~\ref{fact:motivation} for weak alternating index. 

\begin{proposition}\label{pro:weak-preserve}
For game automata, substitution preserves the weak alternating index.
\end{proposition}

\begin{proof}
As we have argued in the proof of Proposition~\ref{pro:alt-preserve}, we can assume that our automata are priority reduced and so are the results of the substitutions. Notice that $\wclass(\aut{A}, q)$ depends only on the internal structure of the SCC containing $q$ and the value of $\wclass(\aut{A}, q)$ computed for the states outside this component. Therefore, substituting either of two game automata recognizing languages of the same weak index results in the same outcome of the procedure. The claim follows by the correctness of the procedure. 
\end{proof}

\subsection{Weak index VS Borel rank}\label{sec:weak_vs_Borel}

Another notable feature of tree languages recognised by deterministic
automata is that within this class, the properties of being Borel and being weakly
recognizable are coextensive. Since the former is
decidable~\cite{niwinski_gap}, the latter is also decidable. This
correspondence can be made even more precise: for languages recognised
by deterministic automata, the weak index and the Borel rank
coincide~\cite{murlak_weak_index}. Notice that this implies that the
Borel rank for deterministic languages is also decidable, a result
originally proved in~\cite{Murlak05}.  As a corollary of the work
presented in the previous part of this Section, we obtain that the same is true
for game automata.

\begin{corollary}
Under restriction to languages recognised by game automata, 
 the weak index hierarchy
coincides with the Borel hierarchy, and both are decidable.
\end{corollary}

\begin{proof} 
From~\cite{murlak_weak_index}, we know that if 
$\wclass(\aut{A}, q)\leq \wRM(\rmin,\rmax)$ then 
$\lang(\aut{A}, q) \wadgeq S_{(\rmin,\rmax)}$.
The coincidence between weak index and Borel rank thence follows by applying Claim~\ref{lemma:hardness} and Fact~\ref{fact:sku}. Decidability is a consequence of Theorem~\ref{thm:weak-dec}.
\end{proof}


%% file: 7_game_decidable.tex
\section{Recognizability by game automata}
\label{sec:isgame}

In this section we give an effective characterization of the class of languages recognized by game automata within the class of all regular languages. The characterization is inspired by the one for deterministic automata~\cite{niwinski_gap}, however, due to the alternation of players, the arguments here are more involved.

We begin with a handful of definitions. Let us fix a finite alphabet $A$. A \emph{trace} is a finite word $w$ over $A \cup \{\dL,\dR\}$, with letters from $A$ on even positions, and directions from $\{\dL,\dR\}$ on odd positions. If the last symbol of $w$ is a letter, the trace is \emph{labelled}, otherwise it is \emph{unlabelled}.  A trace $w$ can be seen as a partial tree $t_w \in\partrees{A}$ consisting of a single path: 
for a labelled trace $w=a_0d_1a_1\dots d_k a_k$,  $\dom(t_w)=\{d_1 d_2 \dots d_i \bigm | i \leq k\}$ and $t_w(d_1 d_2 \dots d_i) = a_i$ for all $i\leq k$. Abusing the notation, we write $w$ instead of $d_1d_2\dots d_k$. The tree $t_w$ has two \emph{final holes}, $w\dL$ and $w\dR$, and \emph{side holes} $d_1d_2\dots d_{i-1} \bar d_i$ for $i\leq k$. For an unlabelled trace $w=a_0d_1a_1\dots d_k$, $t_w$ is defined similarly, but this time it has only one final hole: $d_1d_2\dots d_k$. We shall also write $w$ for this hole. 

A partial tree $t\in\partrees{A}$ is a \emph{realization} of a trace $w$ if it is obtained from $t_w$ by putting some \emph{total} trees in all the side holes of $t_w$. If $w$ is an unlabelled trace, $t$ still has a hole $w$. We write $t(t')$ for the  tree obtained by putting $t'$ in the hole $w$, and $t^{-1}M$ for $\{ t' \bigm |  t(t')\in M\}$. Similarly, if $w$ is a labelled trace, we write $t(t_\dL, t_\dR)$ for the total tree obtained by putting $t_\dL$, $t_\dR$ in the holes $w\dL$ and $w\dR$, respectively, and we define $t^{-1}M$ as $\{ (t_\dL, t_\dR) \bigm | t(t_\dL,t_\dR)\in M \}$. Additionally, $a^{-1}M$ stands for $t_a^{-1}M$ for the root-only tree $t_a$ with $t_a(\varepsilon)=a$.

A language  $Z$ is \emph{non-trivial} if neither $Z$ nor its complement $Z^\complement$ is empty. The following notions are semantic counter-parts of states and transitions of game automata. 

\begin{definition} \label{def:profile}
A \emph{unary profile} is $\ast$ (standing for trivial) or a non-trivial regular tree language $Z$.
A \emph{binary profile} is $\ast$, $\emptyset$,  $\trees{A} \times \trees{A}$,  or \new{a subset} of $\trees{A} \times \trees{A}$ in one of the forms $Z_\dL \times \trees{A}$, $\trees{A}\times Z_\dR$,  $(Z_\dL \times \trees{A}) \cup (\trees{A}\times Z_\dR)$, or $Z_\dL \times Z_\dR$, 
for some non-trivial regular tree languages $Z_\dL,Z_\dR$\new{; the first three (i.e.~$\ast$, $\emptyset$, and $\trees{A} \times \trees{A}$) are called trivial and the remaining ones non-trivial}. 
\end{definition}

We shall see that the binary profiles (except $\ast$) correspond to transitions of the form $\bot$, $\top$, $(q_\dL, \dL)$, $(q_\dR, \dR)$, $(q_\dL, \dL) \lor (q_\dR, \dR)$, and $(q_\dL, \dL) \land (q_\dR, \dR)$, respectively. As a first step, let us relate traces to profiles. 

\begin{definition}\label{def:profsat}
A trace $w$ has non-trivial profile $Z$ in a regular language $M$, if for each realization $t$ of $w$  either $t^{-1}M$ is trivial or  $t^{-1}M = Z$, and for some realization $t_0$, $t_0^{-1}M = Z$; here $Z$ is unary for unlabelled $w$ and  binary for labelled $w$. 

An unlabelled trace $w$ has profile $\ast$ in $M$ if for each realization $t$ of $w$, $t^{-1}M$ is trivial. A labelled trace $wa$ has profile $Z\in\{ \emptyset, \trees{A}\times\trees{A}\}$ in $M$ if $w$ has a non-trivial profile $Z'$ and $a^{-1} Z' = Z$; if $w$ has profile $\ast$, so does the trace $wa$.
\end{definition}

Note that each trace, labelled or unlabelled, has at most one profile in $M$. We write $p_M$ for the partial function assigning profiles to traces. We say that $M$ is \emph{locally game} if each trace has a profile in $M$. The following lemma shows that it is equivalent to assume that all unlabelled traces have profiles in $M$.

\begin{lemma}\label{lm:unarygivebinary}
Given the profiles of traces $w$, $wa\dL$ and $wa\dR$ in $M$, one can effectively compute the profile of $wa$ in $M$.
\end{lemma}

\begin{proof}
Let us \new{assume that} $w$ has a non-trivial profile $K \subseteq \trees{A}$, and  $wa\dL$ and $wa\dR$ have profiles $K_\dL$ and $K_\dR$. Then, by Definition~\ref{def:profsat},  $wa$ cannot have profile $\ast$. \linebreak Let $a^{-1} K = \left \{ (s,t) \in \trees{A}\times\trees{A} \bigm | a(s,t) \in K\right \}$. Is easy to see that $wa$ has a profile $Z\subseteq \trees{A} \times \trees{A}$ if and only if $Z=a^{-1} K$. Thus, it remains to check that $a^{-1} K$ is of one of the forms allowed by Definition~\ref{def:profile}. 

For a set $U \subseteq X\times Y$, we define the \emph{lower section of  $U$ by $x\in X$} as $U_x = \{ y\in Y \bigm | (x,y) \in U\}$, and the \emph{upper section of $U$ by $y\in Y$} as $U^y = \{ x\in X \bigm | (x,y) \in U\}$.

Since $wa\dL$ and $wa\dR$ have profiles $K_\dL$ and $K_\dR$, respectively, it follows easily that
\begin{itemize}
\item each lower section of  $a^{-1} K$ is \new{either $\emptyset$, or $\trees{A}$, or $K_\dR$}; and 
\item each upper section of $a^{-1} K$ is \new{either $\emptyset$, or $\trees{A}$, or $K_\dL$}. 
\end{itemize}

The following three sets form a partition of  $\trees{A}$:
\begin{align*}
X_{\trees{A}} &= \left \{ s_L \bigm| (a^{-1} K)_{s_\dL} =\trees{A} \right\}, \\
X_{K_\dR} &= \left \{ s_L \bigm | (a^{-1} K)_{s_\dL} = K_\dR \right \} \quad
\new{\text{(if $K_{\dR}$ is trivial, take $X_{K_\dR}=\emptyset$)}},\\
X_\emptyset &=\left \{ s_L \bigm| (a^{-1} K)_{ s_\dL} =\emptyset \right \},
\end{align*}
and $a^{-1} K =  X_{\trees{A}} \times \trees{A} \cup  X_{K_\dR}  \times K_\dR \cup X_\emptyset  \times \emptyset $.

 First, assume that $K_\dR$ is trivial. Then $a^{-1}K=X_{\trees{A}}\times\trees{A}$ is a binary profile---either $\emptyset$, \new{or} $\trees{A}\times \trees{A}$, or $Z_\dL\times \trees{A}$, depending on $X_{\trees{A}}$. 

Now, assume that $X_{\dR}$ is non-trivial and  fix $s_\dR\in K_\dR$ and  $s'_\dR\in  \trees{A} \setminus K_\dR$. It follows that
\begin{align*}
(a^{-1}K)^{s_\dR} &=  
 X_{\trees{A}}  \cup X_{K_\dR} \,,\\
(a^{-1}K)^{s'_\dR} &= 
 X_{\trees{A}}\,.
\end{align*}
Hence, by the initial observation on upper sections, $X_{\trees{A}}$ is $\trees{A}$, $K_\dL$, or $\emptyset$, and similarly for  $X_{\trees{A}}  \cup X_{K_\dR}\,$.
\new{
We distinguish three cases (see Fig.~\ref{fig:profile-sections}).
\begin{enumerate}
\item If $X_{\trees{A}}=\emptyset$ then $a^{-1}K=K_\dR\times X_{K_\dR}$ is a binary profile.
\item If $X_\emptyset=\emptyset$ then $a^{-1}K = \trees{A}\times X_{\trees{A}}\cup K_\dR\times \trees{A}$ is a binary profile (either trivial or of the form $Z_\dR\times\trees{A}\ \cup\ \trees{A}\times Z_\dL$).
\item Finally, assume that both $X_{\trees{A}}$ and $X_\emptyset$ are non-empty. In that case both upper sections $(a^{-1}K)^{s_\dR}$ and $(a^{-1}K)^{s'_\dR}$ are non-trivial and therefore are both equal to $K_\dL$. It means that $a^{-1}K=\trees{A}\times \big(X_{\trees{A}}\cup X_{K_{\dR}}\big)$ is a binary profile (either trivial or of the form $\trees{A}\times Z_\dL$).
\end{enumerate}
This concludes the proof.}
\end{proof}

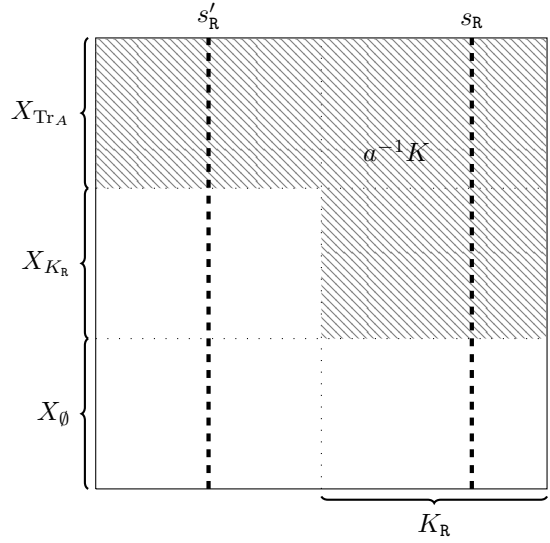
\begin{figure}
\centering
\begin{tikzpicture}

\tikzstyle{split} = [draw, loosely dotted]
\tikzstyle{section} = [draw, ultra thick, dashed]

\path[draw] (0,0) -- (6,0) -- (6,6) -- (0,6) -- (0,0);

\path[split] (3,0) -- (3,6);

\path[split] (0,2) -- (6,2);
\path[split] (0,4) -- (6,4);

\path[fill, pattern=north west lines, pattern color=gray] (0,6) -- (6,6) -- (6,2) -- (3,2) -- (3,4) -- (0,4) -- (0,6);

\node at (4,4.5) {$a^{-1}K$};

\path[ubrace] (3,0) --node {$K_\dR$} (6,0);

\path[lbrace] (0,0) --node {$X_{\emptyset}$} (0,2);
\path[lbrace] (0,2) --node {$X_{K_{\dR}}$} (0,4);
\path[lbrace] (0,4) --node {$X_{\trees{A}}$} (0,6);

\path[section] (1.5,0) -- (1.5,6);
\path[section] (5.0,0) -- (5.0,6);

\node[anchor=south] at (1.5, 6) {$s_\dR'$};
\node[anchor=south] at (5.0, 6) {$s_\dR$};

\end{tikzpicture}
\caption{An illustration of the set $a^{-1}K$ split into a union of products.}
\label{fig:profile-sections}
\end{figure}


Let us now examine the \new{connections} between profiles, and states and transitions of game automata. Let $\aut{B}$ be a total game automaton and  let $q_I$ be a state of $\aut{B}$. For a trace $w$, let $\rho_w=\rho(\aut{B}, t_w, q_I)$ be the run over the tree $t_w$ associated with $w$.

 If $w$ is an unlabelled trace, define $p_{\aut{B},q_I}(w)=\lang(\aut{B}, q)$, if $\rho_w(w)=q\in Q^\aut{B}$; if $\rho_w(w) \notin Q^\aut{B}$, set $p_{\aut{B},q_I}(w)=\ast$.

If $w$ is a labelled trace, set $p_{\aut{B},q_I}(w) =\ast$ if $\rho_w(w) \notin Q^\aut{B}$; otherwise, let $\rho_w(w) = q$ and $b_w = \delta_B(q,a)$, where $a$ is the last symbol of $w$, and set $p_{\aut{B},q_I}(w) = \lang(\aut{B}, b_w)$ where $\lang(\aut{B}, b)$ is the \emph{profile of the transition $b$ in $\aut{B}$}, defined as 
\begin{align*}
\emptyset &\text{ for } b=\bot;\\
\trees{A}\times\trees{A} &\text{ for } b=\top;\\
\lang(\aut{B},q_\dL)\times \trees{A} & \text{ for } b=(q_\dL,\dL);  \\
\trees{A} \times \lang(\aut{B},q_\dR)& \text{ for }b=(q_\dR,\dR);\\
\lang(\aut{B},q_\dL){\times}\trees{A}\,  \cup\,  \trees{A}{\times}\lang(\aut{B},q_\dR)& \text{ for }b=(q_\dL,\dL) \lor (q_\dR,\dR);\\
\lang(\aut{B},q_\dL) \times \lang(\aut{B},q_\dR)& \text{ for }b=(q_\dL,\dL)\land (q_\dR,\dR).
\end{align*}

The following is an easy consequence of Fact~\ref{ft:resolve}.

\begin{lemma}\label{lem:is_locally}
For each total game automaton $\aut{B}$ and state $q_I\in Q^\aut{B}$, for each trace $w$, \[p_{\aut{B}, q_I}(w) = p_{\lang(\aut{B}, q_I)}(w).\]
\end{lemma}

\begin{proof}
Let $M=\lang(\aut{B}, q_I)$.

First consider the case of an unlabelled trace $w$. Let $\rho=\rho(\aut{B}, t_w, q_I)$ be the run. If $\rho(w)=\ast$ then by the definition $p_{\aut{B},q_I}(w)=\ast$. For every realization $t$ of $w$ the position $w$ is not accessible in the game $\rhogame(\aut{B}, t, q_I)$ so $t^{-1}(M)$ is either $\emptyset$ or $\trees{A}$ and by the definition $w$ has profile $\ast$ in $M$.

Now let $\rho(w)=q\in Q^\aut{B}$. In that case $p_{\aut{B}, q_I}(w)=\lang(\aut{B}, q)$. Let $t$ be any realization of $w$. Observe that either:
\begin{enumerate}
\item Player $P$ has a winning $\sigma$ strategy in $\rhogame(\aut{B}, t, q_I)$ such that $w\notin\sigma$. Then $t^{-1}(M)$ is either $\emptyset$ or $\trees{A}$ depending on $P$.
\item Every winning strategy $\sigma$ of $P$ in $\rhogame(\aut{B}, t, q_I)$ contains $w$. In that case $t^{-1}(M)=\lang(\aut{B}, q)$ since the following conditions are equivalent:
\begin{itemize}
\item a composition $t[w:=s]$ belongs to $M$,
\item there exists a winning strategy for \eve in the game $\rhogame(\aut{B}, t[w:=s], q_I)$,
\item \eve can win $\rhogame(\aut{B}, t[w:=s], q_I)$ from $w$,
\item \eve has a winning strategy in the game $\rhogame(\aut{B}, s, q)$,
\item $s\in\lang(\aut{B}, q)$.
\end{itemize}
\end{enumerate}
Recall that there exists a tree $t_0$ that realizes $w$ and resolves $\aut{B}$ from $q_I$---we plug subtrees in the side holes of $t_w$ accordingly to the states assigned by $\rho$. By Fact~\ref{ft:resolve} we obtain that $t_0^{-1}M=\lang(\aut{B}, q)$ so $t_0$ is a witness that $w$ has profile $\lang(\aut{B}, q_I)$.

For the case when $w$ is a labelled trace we use Lemma~\ref{lm:unarygivebinary}---since every unlabelled trace has a profile, we know that every labelled trace also has a profile. It is then easy to verify that the respective equality holds. 
\end{proof}

\begin{corollary}\label{rem:is_locally}
Languages recognized by game automata are locally game. 
\end{corollary}

Being locally game is necessary but not sufficient to be recognizable by a game automaton.

\begin{proposition}
There exists a regular tree language $L$ such that $L$ is locally game but $L$ cannot be recognized by a game automaton.
\end{proposition}

\ident{\thins}{\mathrm{Thin}}
\ident{\cut}{Cut}

\begin{proof}
Consider the alphabet $A=\{a,b\}$. Let $t\in\trees{A}$ be a tree. Let us denote $\cut(a,t)$ as the subtree of $t$ containing those nodes that are accessible by only letters $a$ from the root of $t$. A total tree $t\in\trees{A}$ is called \emph{thin} if $\cut(a,t)$ has only countably many infinite branches. Let $\thins$ \new{be the language of all thin trees}. This language is regular by the \new{equivalence of the following conditions for each tree $t\in\trees{A}$}:
\begin{enumerate}
\item \new{$t$ is not thin},
\item there exists an embedding of the full binary tree $\{\dL,\dR\}^\ast$ into $\cut(a,t)$.
\end{enumerate}

Note that every trace $w$ has a profile $Z_w$ in $\thins$:
\begin{itemize}
\item if $w$ contains a letter $b$ then $Z_w=\ast$,
\item otherwise either $w$ is labelled and therefore $Z_w=\thins\times\thins$, or $w$ is unlabelled and $Z_w=\thins$.
\end{itemize}
This means that $\thins$ is locally game.

Assume that $\thins$ is recognized by a game automaton $\aut{B}$. In that case all transitions of $\aut{B}$ have profile $\thins \times \thins$ (see Lemma~\ref{lem:is_locally} in Section~\ref{sec:isgame}), so $\aut{B}$ is a deterministic automaton. However, a standard argument shows that $\thins$ is not recognizable by any deterministic automaton.
\end{proof}

\medskip

In what follows, for a given locally game language $M$ we construct a game automaton $\aut{G}_M$ that locally computes the profiles and globally reflects the infinitary aspects of $M$. We show that $M$ is recognized by a game automaton if and only if it is recognized by $\aut{G}_M$.


We say that a DFA $\aut{A} = \langle A, Q, q_I, \delta, F\rangle$ computes a \emph{partial} function $\parfun{f}{A^*}{X}$ if $\aut{A}$ recognizes $\dom(f)$ and it comes equipped with a function $\tau^\aut{A} \colon F \to \rg(f)$,  such that $\tau^\aut{A} (\delta(q_I, w)) = f(w)$ for each $w\in\dom(f)$, where $\delta(q, v)$ is the state of $\aut{A}$ after reading word $v$ from state $q$.

The following lemma shows that for each regular tree language $M$ one can effectively construct a DFA $\aut{A}$ that computes (a finite representation of) the profile in $M$ of given trace $w$. In particular, it is decidable whether $M$ is locally game, and the set $\prof_M$ of all possible profiles of traces in $M$ is finite and can be computed from $M$. 


\begin{lemma}\label{lm:regular}
Let $M$ be a regular tree language over an alphabet $A$. There exists a finite automaton that reads a word $w$ over $A\cup\{\dL,\dR\}$ and outputs:
\begin{itemize}
\item $\errTrace$ if $w$ is not a trace;
\item $\errProfile$ if $w$ is a trace but $w$ has no profile in $M$; and 
\item a finite representation of $p_M(w)$ if $w$ is a trace and has a profile in $M$.
\end{itemize}
\end{lemma}

\new{A proof could easily be obtained by} the composition method~\cite{shelah_composition}. However, to make the paper self-contained, we give a direct reasoning. The crucial observation is that if a tree $t'$ is put in a hole of a tree $t$, then the only thing that matters for the acceptance of $t$ is \emph{the type of $t'$}. For the sake of this proof let us fix a regular tree language $M$ recognized by a non-deterministic tree automaton $\aut{B}$ from an initial state $q_I\in Q$.

The type of a total tree $t\in\trees{A}$ is defined as follows:
\[\type(t)=\{q\in Q:\ t\in\lang(\aut{B},q)\}\subseteq Q.\]

The set of types of all total trees is finite and effective, we denote it by $\types\subseteq\mathcal{P}(Q)$. For a set $T\subseteq\types$, by $\lang(T)$ we denote the language of all total trees $t$ such that  $\type(t)\in T$. 

\begin{fact}\label{ft:substitute}
Let $t_\dL,t_\dR,t_\dL',t_\dR'\in\trees{A}$, let $q$ be a state of $\aut{B}$, and $t$ be a tree with two holes. If
\[(\type(t_\dL),\type(t_\dR))=(\type(t_\dL'),\type(t_\dR'))\]
then
\[t(t_\dL,t_\dR)\in\lang(\aut{B}, q) \iff t(t_\dL',t_\dR')\in\lang(\aut{B}, q).\]
In particular, the type $\type\left(t(t_\dL,t_\dR)\right)$ does not depend on the choice of representatives $t_\dL,t_\dR$.
\end{fact}

By the fact above, we can write $t(\tau_\dL,\tau_\dR)$ for the type of $t(t_\dL,t_\dR)$ for any $t_\dL,t_\dR$ with $(\type(t_\dL),\type(t_\dR))=(\tau_\dL,\tau_\dR)$.

Our aim is to construct a finite automaton $\aut{A}$ that reads a finite word $w \in (A\cup\{\dL,\dR\})^*$, checks that $w$ is a trace, and computes a representation of $p_M(w)$, provided that $w$ has profile.

First let us fix
\begin{align*}
Q_1 &= {\mathcal{P}}(\types),\\
Q_2 &= \left\{S\subseteq\types^2\bigm |\ \text{$\lang(S)$ is a profile}\right\}\cup\{\emptyset, \types^2\},\\
Q_E &= \{\errTrace, \errProfile\},\\
Q^\aut{A} &= Q_1\cup Q_2\cup Q_E,\\
q_I^\aut{A} &= \{T\subseteq\types: q_I\in T\}\in Q_1.
\end{align*}

Our aim is to define the transition function $\delta^\aut{A}$ in such a way that Lemma~\ref{lm:Acorrectness}, given below, is satisfied. First, for every $T\in Q_1$, $S\in Q_2$, $U\in Q^\aut{A}$, $a\in A$, $d\in\{\dL,\dR\}$, and $l\in A\cup\{\dL,\dR\}$ we put
\begin{itemize}
\item $\delta^\aut{A}(T,d)=\errTrace$,
\item $\delta^\aut{A}(S,a)=\errTrace$,
\item $\delta^\aut{A}(U,l)=U$.
\end{itemize}

Second, assume that the current state is $T\in Q_1$ and a letter $a$ is given. We define the successive state $\delta^\aut{A}(T,a)\in Q_2\cup\{\errProfile\}$. Let us define the following set of pairs of types
\begin{equation}\label{eq:defS}
S=\left\{(\tau_\dL,\tau_\dR): a(\tau_\dL,\tau_\dR)\in T\right\}.
\end{equation}
Note that, given $S$, we can decide if $\lang(S)$ is a profile and, if it is, we define $\delta^\aut{A}(T,a)=S$. Otherwise we put $\delta^\aut{A}(T,a)=\errProfile$.

Third, assume that the current state is $S\in Q_2$ and a direction $d$ is given. We define the successive state $\delta^\aut{A}(S,d)\in Q_1$. By the symmetry assume that $d=\dL$. Consider the following cases:
\begin{itemize}
\item if $S=T_\dL{\times}\types\cup\types{\times}T_\dR$ for some $T_\dL,T_\dR\subseteq \types$ then $\delta^\aut{A}(S,d)=T_\dL$,
\item otherwise, $\delta^\aut{A}(S,d)=\pi_1(S)$---the projection of $S$ onto the first coordinate.
\end{itemize}
In the case $d=\dR$ we consider $T_\dR$ instead of $T_\dL$ and the projection onto the second coordinate of $S$.

\begin{lemma}\label{lm:Acorrectness}
Let $w$ be a word and $U$ be the state of $\aut{A}$ after reading $w$. The following conditions hold
\begin{enumerate}[(i)]
\item if $U\in \left(Q_1\cup Q_2\right)\setminus\{\emptyset,\types,\types^2\}$ then $w$ is a trace and $\lang(U)$ is the profile of $w$ in $M$,\label{it:corr_prof}
\item if $U\in \{\emptyset,\types,\types^2\}$ then $w$ is a trace and has profile $\ast$ in $M$,\label{it:corr_ast}
\item if $U=\errTrace$ then $w$ is not a trace,
\item if $U=\errProfile$ then $w$ is a trace but has no profile in $M$.
\end{enumerate}
\end{lemma}

\begin{proof}
The first three items follow easily from the definition of profile. 

What remains is to show that if $U=\errProfile$ then the trace $w$ has no profile in $M$.

Assume the contrary and consider a minimal counterexample. Notice that, by definition, such minimal counterexample is a trace of the form $wa$ for some letter $a$. Assume that $wa$ has profile $Z$ in $M$. Let $T$ be the state of $\aut{A}$ after reading $w$ and let $S$ be the set computed according to equation \eqref{eq:defS}. If $T\in\{\emptyset, \types\}$ then $w$ has profile $\ast$ in $M$ by Item~(\ref{it:corr_ast}). Then $wa$ has also profile $\ast$ in $M$ and the state of $\aut{A}$ after reading $wa$ belongs to $\{\emptyset,\types^2\}$. Assume that $T\notin\{\emptyset,\types\}$. By Item~(\ref{it:corr_prof}) we obtain that $w$ has profile $\lang(T)$ in $M$.

Let $t_0$ be a realization of $w$ such that $t_0^{-1}(M)=\lang(T)$. Then, by the definition of $t^{-1}$, we obtain that
\begin{equation}\label{eq:TandZ}
\left(t_0[w:=a]\right)^{-1}(M)=a^{-1}(\lang(T))=Z.
\end{equation}
It is enough to show that $\lang(S)=Z$ and thus $S\in Q_2$ and $S\neq \errProfile$.

To see this, note that the following conditions are all equivalent for a pair of trees $t_\dL,t_\dR\in\trees{A}$:
\begin{enumerate}[(a)]
\item $(t_\dL,t_\dR)\in Z$,
\item $a(t_\dL,t_\dR)\in \lang(T)$,
\item $a(\type(t_\dL),\type(t_\dR))\in T$;
\item $(\type(t_\dL),\type(t_\dR))\in S$,
\item $t_\dL,t_\dR\in\lang (S)$.
\end{enumerate}
 Indeed, the equivalence between conditions (a) and (b) is given by equation~\eqref{eq:TandZ}, the equivalence between conditions (b) and (c) follows by the definition of $a(\tau_\dL,\tau_\dR)$, the equivalence between conditions (c) and (d) is a consequence of equation~\eqref{eq:defS}, and finally the equivalence between conditions (d) and (e) is by the definition of $\lang(S)$.
\end{proof}

\medskip


The infinitary aspects of $M$ are captured by the notion of  \emph{correct} infinite traces. An infinite trace is an infinite word $\pi$ over $A\cup\{\dL,\dR\}$ with letters from $A$ on even positions and directions from $\{\dL,\dR\}$ on odd positions. Just like a finite trace, $\pi$ can be seen as a tree $t_\pi$ consisting of a single infinite branch which has only side holes. A tree $t$ realizes $\pi$ if it is obtained by plugging total trees in the side holes of $t_\pi$.

Assume that $M$ is locally game and let $p_M(w)$ be the profile of $w$ in $M$. We say that \emph{$t$ resolves $M$ up to $\pi$} if $t$ realizes $\pi$ and for each labelled trace $w$ that is a prefix of $\pi$, if $wL$ is a prefix of $\pi$ then
\begin{itemize}
\item $t\restr_{w\dR} \notin Z_\dR$  if $p_M(w) = (Z_\dL \times \trees{A}) \cup (\trees{A} \times Z_\dR)$,
\item $t\restr_{w\dR} \in Z_\dR$ if $p_M(w) = Z_\dL \times Z_\dR$,
\end{itemize}
and symmetrically if $w\dR$ is a prefix of $\pi$.
An infinite trace $\pi$ is \emph{$M$-correct} if some tree $t\in M$ resolves $M$ up to $\pi$. 

\new{
There is an automata-theoretic counterpart of the notion of $M$-correct infinite traces. Consider a game automaton $\aut{C}$, an initial state $q_I$, and an infinite trace $\pi$. Notice that $\pi$ corresponds to a play of the game $\rho=\rho(\aut{C}, t_\pi, q_I)$ associated with $\aut{C}$. We say that $\aut{C}$ \emph{accepts $\pi$ from $q_I$} if either $\rho(v) =\ast$ for some $v\in\dom(t_\pi)$ and $\rho(w) \neq \bot$ for all $w\in\dom(t_\pi)$; or \eve wins the play corresponding to $\pi$ in $\rho$.
}

 
\begin{lemma}\label{lm:is_branches}
\new{A game automaton} $\aut{C}$ accepts $\pi$ from $q_I$ if and only if $\pi$ is $\lang(\aut{C},q_I)$-correct. 
\end{lemma}

\begin{proof}
Let use denote $\lang(\aut{C},q_I)$ as $M'$. Let $\pi\in (A\times \{\dL,\dR\})^\omega$ be an infinite trace. First assume that $\pi$ is $M'$-correct. Let $t'\in M'$ be a tree witnessing it. Recall, that $t'$ is obtained by putting some total trees in the holes of $t_\pi$. Since $t'\in M'$ so there exists a winning strategy $\sigma$ for \eve in $\rhogame(\aut{C}, t', q_I^D)$. Since, whenever \eve could make a choice to leave the branch $\pi$, the respective subtree in $t'$ is losing for her. \new{Let $\rho=\rho(\aut{C}, t', q_I^D)$. If there exists $v\in\dom(t_\pi)$ such that $\rho(v)=\ast$, then there cannot be any $w\prec \dom(t_\pi)$ with $\rho(w)=\bot$ (otherwise the strategy $\sigma$ would not be winning). In the opposite case the whole branch corresponding to $\pi$ must be contained in the strategy $\sigma$. In both cases $\aut{C}$ accepts $\pi$ from $q_I^{\aut{C}}$.}

Now assume that $\aut{C}$ accepts $\pi$ from $q_I^{\aut{C}}$. Let $t'$ be some realization of $\pi$. Consider $\sigma$ to be the strategy of \eve in $\rhogame(\aut{C}, t', q_I^{\aut{C}})$ defined as follows:
\begin{itemize}
\item in all the nodes along $\pi$ follow this branch,
\item whenever \adam selects to go off the branch $\pi$, use some winning strategy in the respective subtree (it exists by the construction).
\end{itemize}

By the definition, $\sigma$ is a winning strategy for \eve in the game $\rhogame(\aut{C}, t', q_I^{\aut{C}})$. Therefore, $t'\in M'$ so $t'$ is a witness that $\pi$ is an $M'$-correct branch.
\end{proof}

\ident{\trep}{TRep}

\begin{lemma} \label{lm:composition_branch}
For each regular tree language $M$ one can effectively construct a deterministic parity automaton $\aut{D}$ recognizing the set of $M$-correct infinite traces. 
\end{lemma}

\begin{proof}
Let $\aut{B}$ be a non-deterministic automaton recognizing the given regular tree language $M$ from a state $g_I$. We show how to express in monadic second-order logic over $\omega$ the fact, that a given $\omega$-word $\pi$ is an $M$-correct branch. By the results \new{of} B\"uchi and McNaughton such a formula $\varphi$ can be effectively translated into a deterministic parity $\omega$-word automaton $\aut{D}$.

\new{As in the proof of Lemma~\ref{lem:is_locally}, we make use of compositional tools. By }$\types\subseteq\mathcal{P}(Q^\aut{B})$ we denote the set of all types of total trees with respect to the automaton~$\aut{B}$.

Intuitively, the formula $\varphi$ guesses the types of the total subtrees that need to be put in the side holes of $t_\pi$ to obtain a tree $t$ that resolves $M$ up to $\pi$. Basing on \new{these} guessed types $\varphi$ can verify that $t\in M$.

Recall that an infinite trace $\pi$ is defined as a word in the language $\left(A\cdot \{\dL,\dR\}\right)^\omega$. A witness for the existence of $t$ will be an infinite word over the alphabet $A\cup\left(\{\dL,\dR\}\times\types\right)$ denoted $\widehat{\pi}$ and called \emph{enrichment of $\pi$ by types}. We require that every even position of $\widehat\pi$ belongs to $A$ and every odd position belongs to $\{\dL,\dR\}\times\types$.
\[\widehat\pi=a_0\cdot(d_0, \tau_0)\cdot a_1\cdot (d_1,\tau_1)\cdot\ldots\]
If $w=a_0\cdot d_0\cdot a_1\cdot d_1\cdot\ldots \cdot d_{n-1}\cdot a_n$ is a labelled trace that is a prefix of $\pi$ we say that $d_n$ is the \emph{final direction of $w$} and $\tau_n$ is the \emph{final type of $w$}.

Let the formula $\psi_\dR$ express that for every trace $w\prec \pi$ with final direction $d$, final type $\tau$, and such that $S$ is the state of $\aut{A}$ (see Lemma~\ref{lem:is_locally}) after reading $w$, the following conditions are satisfied:
\begin{itemize}
\item if $S=S_\dL\times\types\cup\types\times S_\dR$ then \new{$\tau\notin S_{\bar{d}}$},
\item otherwise $\tau\in\pi_i(S)$ where $i=1,2$ for $d=\dL,\dR$ respectively.
\end{itemize}

Note that every $\widehat\pi$ that is \new{an} enrichment of $\pi$ by types induces a total tree $t_{\widehat\pi}$ over the alphabet $A$ where in the \new{$n$-th} side hole of $\pi$ we put some total tree of type $\tau_n$. Note also that $\widehat\pi\models\psi_\dR$ if and only if $t_{\widehat\pi}$ resolves $M$ up to $\pi$. In particular the exact subtrees we put into side holes of $t_\pi$ are irrelevant, we only need to take care of their types.

What remains is to express in MSO logic on $\widehat\pi$ that $t_{\widehat\pi}\in M$. For this we say that there exists an infinite word $\rho$ coding a run the non-deterministic automaton $\aut{B}$ on $t_\pi$. Formally a \emph{word coding a run $\rho$} is defined as a word over the alphabet $Q^B\cup\{\ast\}\cup D^B$ where $D^B$ is the set of all deterministic transitions appearing in the transitions of $\aut{B}$. The elements of $Q^B\cup\{\ast\}$ are supposed to appear on even positions of $\rho$ and elements of $D^B$ are supposed to appear on odd positions of $\rho$:
\[\rho=q_0\cdot b_0\cdot q_1\cdot b_1\cdot\ldots\]

Let the formula $\psi_M$ express the following facts about the combination of words $\widehat\pi\otimes\rho$ in the language \[\left[A\times\left(Q^B\cup\{\ast\}\right)\cdot \left(\left(\{\dL,\dR\}\times\types\right)\times D^B\right)\right]^\omega\]
\begin{itemize}
\item the state $q_0$ equals $q_I$,
\item for every $n$ the transition $b_n$ is one of the deterministic transitions appearing in $\delta^B(q_n,a_n)$,
\item for every $n$ the state assigned to $\bar{d_n}$ by $b_n$ (if any) belongs to $\tau_n$,
\item for every $n$ the state assigned to $d_n$ by $b_n$ (if any) equals $q_{n+1}$, if there is no such state then $q_{n+1}=\ast$,
\item either from some point on $b_n=\ast$ or the parity condition is satisfied by the sequence of states $q_0,q_1,\ldots$.
\end{itemize}

Note, that $\widehat\pi\otimes\rho\models \psi_M$ if and only if $\rho$ encodes an accepting run of $\aut{B}$ on $t_\pi$ that assigns to the $n$'th hole of $t_\pi$ a state belonging to $\tau_n$. Therefore, $\widehat{\pi}\otimes\rho\models\psi_M$ if and only if the run encoded by $\rho$ can be extended to an accepting run of $\aut{B}$ on $t_{\widehat\pi}$.

Let $\varphi$ express for a given infinite trace $\pi$ that there exists an enrichment of $\pi$ by types $\tau_n$ and an encoding of run $\rho$ such that $\widehat\pi\models\psi_\dR$ and $\widehat\pi\otimes\rho\models\psi_M$. Note that $\pi\models\varphi$ if and only if there exists a tree $t=t_{\widehat\pi}\in M$ that realizes $\pi$ and that resolves $M$ up to $\pi$.
\end{proof}

\medskip

We define $\aut{G}_M$ as a product of $\aut{A}$ and $\aut{D}$, with priorities inherited from $\aut{D}$ and the types of transitions ($\lor$, $\land$, etc.) determined by the type of profile computed by $\aut{A}$. More precisely, for $a\in A$, $(p,q) \in Q^\aut{A}\times Q^\aut{D}$, $\tau =\tau^\aut{A}(\delta^{\aut{A}}(p,a))$, define  $\delta\big((p,q), a\big)$ as
\begin{align*} 
\top &\quad \text{if} \quad \tau\in \{ \ast, \trees{A}\times \trees{A}\};\\
\bot &\quad \text{if} \quad\tau= \emptyset;\\
\beta(p, q, a, \dL) &\quad \text{if} \quad\tau = Z_\dL \times \trees{A};\\
\beta(p, q, a, \dR) &\quad \text{if} \quad\tau= \trees{A} \times Z_\dR ;\\
\beta(p, q, a, \dL) \lor \beta(p, q, a, \dR)  &\quad \text{if} \quad
 \tau = Z_\dL \times \trees{A} \cup  \trees{A} \times Z_\dR;\\
\beta(p, q, a, \dL) \land \beta(p, q, a, \dR) &\quad \text{if} \quad\tau = Z_\dL \times Z_\dR;
\end{align*}
where $\beta(p, q, a, d)$ is defined as $\big ( (\delta^\aut{A} (p, ad), \delta^\aut{D} (q, ad)), d\big)$.
Let $q_M=(q_I^\aut{A}, q_I^\aut{D})$.

\begin{theorem}
A regular language $M$ is recognized by a game automaton iff $M$ is locally game and $\lang(\aut{G}_M, q_M) = M$.
\end{theorem}

\begin{proof}
Assume that $M = \lang(\aut{B}, q_I^\aut{B})$ for some game automaton $\aut{B}$ and $q_I^\aut{B}\in Q^\aut{B}$.  By Corollary~\ref{rem:is_locally}, $M$ is locally game. Fix $t\in\trees{A}$ and let $\rho_M=\rho(\aut{G}_M, t, q_M)$ and $\rho_\aut{B}=\rho(\aut{B}, t,q_I^\aut{B})$. By Lemma~\ref{lem:is_locally}, $p_M(w)$ determines the profiles of the corresponding transitions in $\rho_\aut{B}$ and $\rho_M$. Hence, the games associated to these runs are isomorphic if the priorities are ignored. Let $\pi$ be an infinite trace in $t$. By the construction, $\aut{G}_M$ accepts $\pi$ from $q_M$ iff $\pi$ is $M$-correct. By Lemma~\ref{lm:is_branches}, $\pi$ is $M$-correct iff $\aut{B}$ accepts $\pi$ from $q_I^\aut{B}$. It follows that $\rho_\aut{B}$ is accepting iff $\rho_M$ is accepting.
\end{proof}

\medskip

As an immediate corollary we obtain the following.

\begin{theorem}\label{thm:game} 
Given an alternating automaton $\aut{A}$ and a state $q_I$, it is decidable whether $\lang(\aut{A},q_I)$ is recognized by a game automaton.  If so, some game automaton recognizing $\lang(\aut{A},q_I)$ can be effectively constructed from $\aut{A}$ and $q_I$.
\end{theorem}


%% file: journal_TOCL.bbl

\begin{thebibliography}{00}


\ifx \showCODEN    \undefined \def \showCODEN     #1{\unskip}     \fi
\ifx \showDOI      \undefined \def \showDOI       #1{{\tt DOI:}\penalty0{#1}\ }
  \fi
\ifx \showISBNx    \undefined \def \showISBNx     #1{\unskip}     \fi
\ifx \showISBNxiii \undefined \def \showISBNxiii  #1{\unskip}     \fi
\ifx \showISSN     \undefined \def \showISSN      #1{\unskip}     \fi
\ifx \showLCCN     \undefined \def \showLCCN      #1{\unskip}     \fi
\ifx \shownote     \undefined \def \shownote      #1{#1}          \fi
\ifx \showarticletitle \undefined \def \showarticletitle #1{#1}   \fi
\ifx \showURL      \undefined \def \showURL       #1{#1}          \fi

\bibitem[\protect\citeauthoryear{Arnold}{Arnold}{1999}]%
        {arnold_strict}
{Andr{\'{e}} Arnold}. 1999.
\newblock \showarticletitle{The mu-calculus alternation-depth hierarchy is
  strict on binary trees}.
\newblock {\em ITA\/} {33}, 4/5 (1999), 329--340.
\newblock


\bibitem[\protect\citeauthoryear{Arnold and Niwi{\'{n}}ski}{Arnold and
  Niwi{\'{n}}ski}{2001}]%
        {niwinski_rudiments}
{Andr{\'{e}} Arnold} {and} {Damian Niwi{\'{n}}ski}. 2001.
\newblock {\em Rudiments of mu-calculus}.
\newblock Elsevier.
\newblock


\bibitem[\protect\citeauthoryear{Arnold and Niwi{\'{n}}ski}{Arnold and
  Niwi{\'{n}}ski}{2007}]%
        {niwinski_strict}
{Andr{\'{e}} Arnold} {and} {Damian Niwi{\'{n}}ski}. 2007.
\newblock \showarticletitle{Continuous Separation of Game Languages}.
\newblock {\em Fundamenta Informaticae\/} {81}, 1-3 (2007), 19--28.
\newblock


\bibitem[\protect\citeauthoryear{Arnold and Santocanale}{Arnold and
  Santocanale}{2005}]%
        {arnold_separation}
{Andr{\'{e}} Arnold} {and} {Luigi Santocanale}. 2005.
\newblock \showarticletitle{Ambiguous classes in {$\mu$}-calculi hierarchies}.
\newblock {\em TCS\/} {333}, 1--2 (2005), 265--296.
\newblock


\bibitem[\protect\citeauthoryear{Boja{\'{n}}czyk and Place}{Boja{\'{n}}czyk and
  Place}{2012}]%
        {bojanczyk_boolean}
{Miko{\l}aj Boja{\'{n}}czyk} {and} {Thomas Place}. 2012.
\newblock \showarticletitle{Regular Languages of Infinite Trees That Are
  {Boolean} Combinations of Open Sets}. In {\em ICALP}. 104--115.
\newblock


\bibitem[\protect\citeauthoryear{Bradfield}{Bradfield}{1998}]%
        {bradfield_simplifying}
{Julian Bradfield}. 1998.
\newblock \showarticletitle{Simplifying the modal mu-calculus alternation
  hierarchy}. In {\em STACS}. 39--49.
\newblock


\bibitem[\protect\citeauthoryear{B{\"{u}}chi}{B{\"{u}}chi}{1962}]%
        {buchi_decision}
{Julius~Richard B{\"{u}}chi}. 1962.
\newblock \showarticletitle{On a Decision Method in Restricted Second-Order
  Arithmetic}. In {\em Proc. 1960 Int. Congr. for Logic, Methodology and
  Philosophy of Science}. 1--11.
\newblock


\bibitem[\protect\citeauthoryear{Colcombet, Kuperberg, L{\"{o}}ding, and
  {Vanden Boom}}{Colcombet et~al\mbox{.}}{2013}]%
        {colcombet_weak}
{Thomas Colcombet}, {Denis Kuperberg}, {Christof L{\"{o}}ding}, {and} {Michael
  {Vanden Boom}}. 2013.
\newblock \showarticletitle{Deciding the weak definability of B{\"{u}}chi
  definable tree languages}. In {\em CSL}. 215--230.
\newblock


\bibitem[\protect\citeauthoryear{Colcombet and L{\"{o}}ding}{Colcombet and
  L{\"{o}}ding}{2008}]%
        {loding_index_to_bounds}
{Thomas Colcombet} {and} {Christof L{\"{o}}ding}. 2008.
\newblock \showarticletitle{The Non-deterministic {Mostowski} Hierarchy and
  Distance-Parity Automata}. In {\em ICALP (2)}. 398--409.
\newblock


\bibitem[\protect\citeauthoryear{Duparc, Facchini, and Murlak}{Duparc
  et~al\mbox{.}}{2011}]%
        {murlak_weak_game}
{Jacques Duparc}, {Alessandro Facchini}, {and} {Filip Murlak}. 2011.
\newblock \showarticletitle{Definable Operations On Weakly Recognizable Sets of
  Trees}. In {\em FSTTCS}. 363--374.
\newblock


\bibitem[\protect\citeauthoryear{Duparc and Murlak}{Duparc and Murlak}{2007}]%
        {duparc_weak}
{Jacques Duparc} {and} {Filip Murlak}. 2007.
\newblock \showarticletitle{On the Topological Complexity of Weakly
  Recognizable Tree Languages}. In {\em Fundamentals of Computation Theory,
  16th International Symposium, {FCT} 2007, Budapest, Hungary, August 27-30,
  2007, Proceedings}. 261--273.
\newblock


\bibitem[\protect\citeauthoryear{Emerson and Jutla}{Emerson and Jutla}{1991}]%
        {jutla_determinacy}
{Allen Emerson} {and} {Charanjit Jutla}. 1991.
\newblock \showarticletitle{Tree Automata, mu-Calculus and Determinacy}. In
  {\em FOCS'91}. 368--377.
\newblock


\bibitem[\protect\citeauthoryear{Facchini, Murlak, and Skrzypczak}{Facchini
  et~al\mbox{.}}{2013}]%
        {murlak_game_auto}
{Alessandro Facchini}, {Filip Murlak}, {and} {Micha{\l} Skrzypczak}. 2013.
\newblock \showarticletitle{Rabin-Mostowski Index Problem: A Step beyond
  Deterministic Automata}. In {\em LICS}. 499--508.
\newblock


\bibitem[\protect\citeauthoryear{Facchini, Murlak, and Skrzypczak}{Facchini
  et~al\mbox{.}}{2015}]%
        {game_wollic}
{Alessandro Facchini}, {Filip Murlak}, {and} {Micha{\l} Skrzypczak}. 2015.
\newblock \showarticletitle{On the weak index problem for game automata}. In
  {\em WoLLIC}. 93--108.
\newblock


\bibitem[\protect\citeauthoryear{Kechris}{Kechris}{1995}]%
        {kechris_descriptive}
{Alexander Kechris}. 1995.
\newblock {\em Classical descriptive set theory}.
\newblock Springer-Verlag, New York.
\newblock


\bibitem[\protect\citeauthoryear{Kuperberg}{Kuperberg}{2012}]%
        {kuperberg_phd}
{Denis Kuperberg}. 2012.
\newblock {\em Etude de classes de fonctions de co{\^u}t r{\'{e}}guli{\`e}res}.
\newblock Ph.D. Dissertation. Universit{\'{e}} Paris Diderot.
\newblock


\bibitem[\protect\citeauthoryear{K{\"u}sters and Wilke}{K{\"u}sters and
  Wilke}{2002}]%
        {kuster_first_level}
{Ralf K{\"u}sters} {and} {Thomas Wilke}. 2002.
\newblock \showarticletitle{Deciding the First Level of the {$\mu$}-Calculus
  Alternation Hierarchy}. In {\em FSTTCS}. 241--252.
\newblock


\bibitem[\protect\citeauthoryear{Mostowski}{Mostowski}{1984}]%
        {mostowski_standard}
{Andrzej~W. Mostowski}. 1984.
\newblock \showarticletitle{Regular expressions for infinite trees and a
  standard form of automata}. In {\em Symposium on Computation Theory}.
  157--168.
\newblock


\bibitem[\protect\citeauthoryear{Mostowski}{Mostowski}{1991a}]%
        {mostowski_parity_games}
{Andrzej~W. Mostowski}. 1991a.
\newblock {\em Games with forbidden positions}.
\newblock {T}echnical {R}eport. University of Gda{\'{n}}sk.
\newblock


\bibitem[\protect\citeauthoryear{Mostowski}{Mostowski}{1991b}]%
        {mostowski_hierarchies}
{Andrzej~W. Mostowski}. 1991b.
\newblock \showarticletitle{Hierarchies of Weak Automata and Weak Monadic
  Formulas}.
\newblock {\em Theor. Comput. Sci.\/} {83}, 2 (1991), 323--335.
\newblock


\bibitem[\protect\citeauthoryear{Muller, Saoudi, and Schupp}{Muller
  et~al\mbox{.}}{1986}]%
        {MullerSS86}
{David~E. Muller}, {Ahmed Saoudi}, {and} {Paul~E. Schupp}. 1986.
\newblock \showarticletitle{Alternating Automata. The Weak Monadic Theory of
  the Tree, and its Complexity}. In {\em ICALP} {\em (Lecture Notes in Computer
  Science)}, {Laurent Kott} (Ed.), Vol. 226. 275--283.
\newblock


\bibitem[\protect\citeauthoryear{Murlak}{Murlak}{2005}]%
        {Murlak05}
{Filip Murlak}. 2005.
\newblock \showarticletitle{On Deciding Topological Classes of Deterministic
  Tree Languages}. In {\em CSL 2005}. 428--441.
\newblock


\bibitem[\protect\citeauthoryear{Murlak}{Murlak}{2008a}]%
        {murlak_phd}
{Filip Murlak}. 2008a.
\newblock {\em Effective topological hierarchies of recognizable tree
  languages}.
\newblock Ph.D. Dissertation. University of Warsaw.
\newblock


\bibitem[\protect\citeauthoryear{Murlak}{Murlak}{2008b}]%
        {murlak_weak_index}
{Filip Murlak}. 2008b.
\newblock \showarticletitle{Weak index versus {Borel} rank}. In {\em STACS
  2008} {\em (LIPIcs)}, Vol.~1. 573--584.
\newblock


\bibitem[\protect\citeauthoryear{Niwi{\'{n}}ski and Walukiewicz}{Niwi{\'{n}}ski
  and Walukiewicz}{1998}]%
        {niwinski_relating}
{Damian Niwi{\'{n}}ski} {and} {Igor Walukiewicz}. 1998.
\newblock \showarticletitle{Relating Hierarchies of Word and Tree Automata}. In
  {\em STACS}. 320--331.
\newblock


\bibitem[\protect\citeauthoryear{Niwi{\'{n}}ski and Walukiewicz}{Niwi{\'{n}}ski
  and Walukiewicz}{2003}]%
        {niwinski_gap}
{Damian Niwi{\'{n}}ski} {and} {Igor Walukiewicz}. 2003.
\newblock \showarticletitle{A gap property of deterministic tree languages}.
\newblock {\em Theor. Comput. Sci.\/} {1}, 303 (2003), 215--231.
\newblock


\bibitem[\protect\citeauthoryear{Niwi{\'{n}}ski and Walukiewicz}{Niwi{\'{n}}ski
  and Walukiewicz}{2005}]%
        {niwinski_deterministic}
{Damian Niwi{\'{n}}ski} {and} {Igor Walukiewicz}. 2005.
\newblock \showarticletitle{Deciding Nondeterministic Hierarchy of
  Deterministic Tree Automata}.
\newblock {\em Electr. Notes Theor. Comput. Sci.\/}  {123} (2005), 195--208.
\newblock


\bibitem[\protect\citeauthoryear{Otto}{Otto}{1999}]%
        {otto_recursion}
{Martin Otto}. 1999.
\newblock \showarticletitle{Eliminating Recursion in the {$\mu$}-Calculus}. In
  {\em STACS}. 531--540.
\newblock


\bibitem[\protect\citeauthoryear{Rabin}{Rabin}{1969}]%
        {rabin_s2s}
{Michael~O. Rabin}. 1969.
\newblock \showarticletitle{Decidability of second-order theories and automata
  on infinite trees}.
\newblock {\em Trans. of the American Math. Soc.\/}  {141} (1969), 1--35.
\newblock


\bibitem[\protect\citeauthoryear{Rabin}{Rabin}{1970}]%
        {rabin_separation}
{Michael~O. Rabin}. 1970.
\newblock \showarticletitle{Weakly definable relations and special automata}.
  In {\em Proceedings of the Symposium on Mathematical Logic and Foundations of
  Set Theory}. North-Holland, 1--23.
\newblock


\bibitem[\protect\citeauthoryear{Shelah}{Shelah}{1975}]%
        {shelah_composition}
{Saharon Shelah}. 1975.
\newblock \showarticletitle{The Monadic Theory of Order}.
\newblock {\em The Annals of Mathematics\/} {102}, 3 (1975), 379--419.
\newblock


\bibitem[\protect\citeauthoryear{Skurczy{\'{n}}ski}{Skurczy{\'{n}}ski}{1993}]%
        {skurczynski_borel_infinite}
{Jerzy Skurczy{\'{n}}ski}. 1993.
\newblock \showarticletitle{The {Borel} hierarchy is infinite in the class of
  regular sets of trees}.
\newblock {\em Theoretical Computer Science\/} {112}, 2 (1993), 413--418.
\newblock


\bibitem[\protect\citeauthoryear{Urba{\'{n}}ski}{Urba{\'{n}}ski}{2000}]%
        {urbanski_det_buchi}
{Tomasz~Fryderyk Urba{\'{n}}ski}. 2000.
\newblock \showarticletitle{On Deciding if Deterministic {Rabin} Language Is in
  {B{\"{u}}chi} Class}. In {\em ICALP}. 663--674.
\newblock


\bibitem[\protect\citeauthoryear{{Vanden Boom}}{{Vanden Boom}}{2012}]%
        {boom_phd}
{Michael {Vanden Boom}}. 2012.
\newblock {\em Weak Cost Automata over Infinite Trees}.
\newblock Ph.D. Dissertation. University of Oxford.
\newblock


\bibitem[\protect\citeauthoryear{Wadge}{Wadge}{1983}]%
        {wadge_phd}
{William Wadge}. 1983.
\newblock {\em Reducibility and determinateness in the {Baire} space}.
\newblock Ph.D. Dissertation. University of California, Berkeley.
\newblock


\bibitem[\protect\citeauthoryear{Walukiewicz}{Walukiewicz}{2002}]%
        {walukiewicz_low_levels}
{Igor Walukiewicz}. 2002.
\newblock \showarticletitle{Deciding low levels of tree-automata hierarchy}.
\newblock {\em Electr. Notes Theor. Comput. Sci.\/}  {67} (2002), 61--75.
\newblock


\end{thebibliography}
